\pdfoutput=1
\RequirePackage{ifpdf}
\ifpdf 
\documentclass[pdftex]{sigma}
\else
\documentclass{sigma}
\fi

\usepackage{tikz}
\usetikzlibrary{shapes.geometric, arrows}
\tikzstyle{a} = [rectangle, minimum width=3cm, minimum height=1cm, text centered, draw=black, fill=red!30]
\tikzstyle{b} = [rectangle, minimum width=3cm, minimum height=1cm, text centered, draw=black, fill=blue!30]
\tikzstyle{c} = [rectangle, minimum width=3cm, minimum height=1cm, text centered, draw=black, fill=green!30]
\tikzstyle{d} = [rectangle, minimum width=3cm, minimum height=1cm, text centered, draw=black, fill=orange!30]
\tikzstyle{e} = [rectangle, minimum width=3cm, minimum height=1cm, text centered, draw=black, fill=yellow!30]
\tikzstyle{f} = [rectangle, minimum width=3cm, minimum height=1cm, text centered, draw=black, fill=gray!30]
\tikzstyle{arrow} = [line width=3pt, ->,>=triangle 45]
\DeclareMathOperator{\arccosh}{arcCosh}

\numberwithin{equation}{section}

\newtheorem{Theorem}{Theorem}[section]

 { \theoremstyle{definition}

\newtheorem{Example}[Theorem]{Example}
 }

\begin{document}

\allowdisplaybreaks

\renewcommand{\thefootnote}{$\star$}

\newcommand{\arXivNumber}{1505.00527}

\renewcommand{\PaperNumber}{088}

\FirstPageHeading

\ShortArticleName{Examples of Complete Solvability of 2D Classical Superintegrable Systems}

\ArticleName{Examples of Complete Solvability \\ of 2D Classical Superintegrable Systems\footnote{This paper is a~contribution to the Special Issue
on Analytical Mechanics and Dif\/ferential Geometry in honour of Sergio Benenti.
The full collection is available at \href{http://www.emis.de/journals/SIGMA/Benenti.html}{http://www.emis.de/journals/SIGMA/Benenti.html}}}

\Author{Yuxuan~{CHEN}~$^\dag$, Ernie G.~{KALNINS}~$^\ddag$, Qiushi {LI}~$^\dag$ and Willard {MILLER~Jr.}~$^\dag$}

\AuthorNameForHeading{Y.~Chen, E.G.~Kalnins, Q.~Li and W.~Miller~Jr.}

\Address{$^\dag$~School of Mathematics, University of Minnesota, Minneapolis, Minnesota, 55455, USA}
\EmailD{\href{mailto:yc397@cam.ac.uk}{yc397@cam.ac.uk}, \href{mailto:lixx0939@umn.edu}{lixx0939@umn.edu}, \href{mailto:miller@ima.umn.edu}{miller@ima.umn.edu}}

\Address{$^\ddag$~Department of Mathematics, University of Waikato, Hamilton, New Zealand}
\EmailD{\href{mailto:math0236@waikato.ac.nz}{math0236@waikato.ac.nz}}

\ArticleDates{Received May 05, 2015, in f\/inal form October 27, 2015; Published online November 03, 2015}

\Abstract{Classical (maximal) superintegrable systems in $n$ dimensions are Hamiltonian systems with $2n-1$ independent constants of the motion,
globally def\/ined, the maximum number possible.
They are very special because they can be solved algebraically.
In this paper we show explicitly, mostly through examples of 2nd order superintegrable systems in 2~dimensions, how the trajectories can be determined in detail
using  rather elementary algebraic, geometric  and analytic methods applied to the closed quadratic algebra of symmetries of the system,
without resorting to separation of variables techniques
or trying to integrate Hamilton's equations. We treat a family of 2nd order  degenerate systems: oscillator analogies on Darboux, nonzero constant curvature, and f\/lat spaces,
related to one another via contractions,  and obeying Kepler's laws. Then we treat two 2nd order nonde\-ge\-ne\-ra\-te systems, an analogy of a caged Coulomb problem on the 2-sphere and
its contraction to a~Euclidean space caged Coulomb problem. In all cases the symmetry algebra structure provides detailed
information about the trajectories, some
of which are rather complicated. An interesting example is the occurrence of ``metronome orbits'', trajectories conf\/ined to an arc rather than a loop,
which are indicated  clearly from the structure equations but might be overlooked using more traditional  methods. We also treat
the Post--Winternitz system, an example of a classical 4th order superintegrable system that cannot be
solved using separation of variables.
 Finally we treat a superintegrable system, related to the addition theorem for elliptic functions,
whose constants of the motion are only rational in the momenta. It is a system of special interest because its constants of the motion generate a
closed polynomial algebra. This paper contains many new results but we have tried to present most of the materials in a fashion that is easily
accessible to nonexperts, in order to provide  entr\'ee to superintegrablity theory.}


\Classification{20C99; 20C35; 22E70}

\tableofcontents

\renewcommand{\thefootnote}{\arabic{footnote}}
\setcounter{footnote}{0}

\section{Introduction}

Classical and quantum mechanical Hamiltonian systems that can be solved explicitly, both algebraically and
analytically,  and with adjustable parameters, are relatively rare and highly prized. Famous classical examples are the anharmonic oscillator (Lissajous patterns)  and Kepler systems (planetary orbits) and the Hohmann transfer for orbital navigation~\cite{Curtis}.
 Famous quantum examples are the  Coulomb system (energy levels of the hydrogen atom,
leading to the periodic table of the elements), and the quantum isotropic oscillator. The solvability of these systems is related to their
symmetry, not necessarily group symmetry. This higher order symmetry  is captured by the concept of superintegrability.
A natural classical Hamiltonian system (with the Hamiltonian as kinetic energy plus potential energy) on an $n$-dimensional Riemannian space
is said to be (maximally) superintegrable if it admits the maximally possible $2n-1$ functionally independent constants of the motion globally def\/ined (usually required to be polynomial or at least rational
in the momenta). Similarly a quantum Hamiltonian system $H=\Delta_n+V$ (where $\Delta_n$ is a Laplace--Beltrami operator on an $n$-dimensional
Riemannian manifold and $V$ is a potential function) is (maximally) superintegrable if it admits $2n-1$ algebraically independent partial dif\/ferential operators
commuting with $H$. There is a rapidly growing literature concerning these systems, e.g.,
\cite{BEHRR,BEHRR2,ballesteros2013anisotropic,CDR2008,Dask2001,EVA,EvansVerrier2008,FORDY,Zhedanov1992b,
KKM10,VILE,MSVW,Marquette2012,MayrandVinet1,MPW,TempTW,SCQS,TTW1,tsiganov2008maximally}.
In this paper  we consider only classical systems and $n=2$ so all of our systems admit 3 independent constants of the motion.
By taking Poisson brackets  of the classical constants of the motion we generate a  symmetry  algebra, not necessarily a Lie algebra, which is never
abelian.  (This contrasts with integrable Hamiltonian systems which admit $n$ constants of the motion in involution, so that the
symmetry algebra is always abelian.) It is this nonabelian symmetry algebra and its structure that is responsible for
the solvability of superintegrable systems.

To be more explicit, along a specif\/ic classical trajectory each constant of the motion ${\cal L}_j$ takes a f\/ixed value $\ell_j$,
so the trajectory can be characterized as the intersection of
$2n-1$ hypersurfaces ${\cal L}_j=\ell_j$ in the $2n$-dimensional phase space. Thus in principle the path of the trajectory can be
determined algebraically, though not how it is traced out in time. Since the Poisson bracket of two constants of the motion is again a constant of the motion, the nonabelian symmetry
algebra gives us relationships between the symmetries. In the quantum case each bound state eigenspace of the Hamiltonian is invariant under
the action of the symmetry algebra, so that a~knowledge of the irreducible representations of the symmetry algebra gives useful information about the dimensions of the eigenspaces and
of the eigenvalues themselves.

In this paper we illustrate the value of superintegrabilty by studying and solving several families of classical superintegrable systems. Second degree
superintegrable systems
are those whose generating symmetries are all polynomials in the momenta of degree $\le 2$. All of these systems have been classif\/ied for 2 dimensions (as have the systems with nondegenerate potentials in 3 dimensions)~\cite{KC,KKM20061, KKMP}.
They occur on constant curvature spaces
(with 3 Killing vectors, admitting the most symmetries), on the
4 Darboux spaces (admitting 1 Killing vector) and 6~families of Koenigs spaces (admitting no Killing vectors)~\cite{Koenigs, MPW}. Thus,
after the constant curvature spaces, the Darboux spaces admit the
most symmetries.
  The Darboux metrics an be written as
\begin{alignat*}{5}
&D1\colon \quad &&   ds^2=4x\big(dx^2+dy^2\big),\qquad &&  D2\colon \quad && ds^2=\frac{x^2+1}{x^2}\big(dx^2+dy^2\big),&\\
 & D3\colon \quad && ds^2=\frac{e^x+1}{e^{2x}}\big(dx^2+dy^2\big),\qquad && D4(b)\colon \quad && ds^2= \frac{2\cos 2x+b}{\sin^2 2x}\big(dx^2+dy^2\big),&
 \end{alignat*}
An example of a Koenigs space is
\begin{gather*}
 ds^2=\left(\frac{c_1}{x^2+y^2}+\frac{c_2}{x^2}+\frac{c_3}{y^2}+c_4\right)\big(dx^2+dy^2\big).
 \end{gather*}

We  f\/irst treat some analogs of the harmonic oscillator.
(A~treatment of the Kepler system and its analog on the 2-sphere from this point of view
is contained in Chapter~3 of the  article~\cite{MPW}.) We start by studying  a system on the Darboux space $D4(b)$
with a 1-parameter potential. Then by taking limits (contractions~\cite{Wigner, KMP2013}) of this system we
obtain systems on the Darboux space~$D3$, on the Poincar\'e upper half plane, the 2-sphere (the Higgs oscillator~\cite{Higgs}), and
f\/inally the isotropic oscillator on Euclidean space.

The second class of examples are nondegenerate (3-parameter potential) systems. The f\/irst, $S7$ in our listing~\cite{KKMP}, can be regarded as a caged
version of an analog to the
Coulomb potential on the 2-sphere. It contracts to the caged Coulomb system $E16$ in Euclidean space.

The preceding examples can also be studied analytically via separation of variables. However, the  Post--Winternitz system~\cite{PW2011} cannot,
see also~\cite{MPY}.  It is  4th
degree superintegrable but
nonseparable. However, we show that the classical orbits can be found exactly via superintegrabilty.

The last examples are dif\/ferent. They are separable in elliptic coordinates and can be derived via an action-angle construction.
The usual action-angle construction of a superintegrable
and separable system requires the addition theorem for trigonometric or hyperbolic functions and leads to polynomial superintegrability, e.g.,~\cite{KKM10}.
The construction here uses the addition formula for elliptic functions~\cite{DLMF}.
It leads to  classical  systems where some of the constants are nontrivially rational in the canonical momenta.
We present two examples
where, however, the systems
 have  polynomial symmetry
algebras, so we consider them as superintegrable and worthy of study. The new feature here is that the symmetry algebra closes polynomially,
even though the system is only rationally superintegrable. These are the f\/irst examples of such behavior known to us.
Again we can determine the trajectories exactly.

With the exception of the elliptic  systems, all of the superintegrable systems in this paper have been derived and  classif\/ied before;
we take the structure equations as given
and show how superintegrability alone leads to formulas for the trajectories. These formulas and their analysis are new. The elliptic systems have not been found elsewhere to our
knowledge, so we demonstrate the procedure to
derive them.

\looseness=1
This paper is partly pedagogical (the Introduction and the Appendix) but mostly new research (Sections~\ref{section2}--\ref{section5}). In the Appendix we give a brief review of fundamental def\/initions and results from  classical Hamiltonian mechanics, adapted from \cite{MPW}, needed as background for our computations. The point here is to illustrate how, using the structure equations of superintegrable systems alone, we can derive
and classify the trajectories via algebra. Hamilton's equations are used only to determine the
periods of orbits for systems with degenerate potential. Separation of variable techniques are not used, except in the last section; our approach is easy to understand
geometrically. Analogously the spectra of the corresponding quantum systems can be obtained algebraically from the structure relations,
though we do not treat this here.
Some recent papers, e.g.,~\cite{Negro, Ragnisco}, adopt a related but dif\/ferent approach by using ladder operators constructed from the
structure algebra to compute trajectories and spectra for systems with degenerate potential. We know of no prior treatments of the nondegenerate
caged systems~$S7$ and~$E16$, or of the use of the structure algebra to call attention to special orbits, such as those for which~${\cal R}^2=0$.
We point out the
contractions that relate our various superintegrable systems.

\section{Examples of 2D 2nd degree degenerate systems}\label{section2}

Our f\/irst examples are degenerate systems. These  have 1-parameter potentials and always admit a symmetry that is a 1st degree polynomial in the momenta, hence a group symmetry that can be interpreted as  invariance with respect to rotation or translation corresponds to a constant of the motion which leads to an analog of Kepler's 2nd Law.
The only possibilities in 2~dimensions are constant curvature and Darboux spaces~\cite{KKMW}. In~\cite{MPW} there is an example of an analog
of the Kepler problem on the 2-sphere that satisf\/ies Kepler's three laws and then by a limiting process (a contraction)  goes to the Kepler system
in Euclidean space. Here we will treat an analog of the harmonic oscillator on the Darboux space~$D4(b)$ and show that it obeys analogs of Kepler's laws.
Then by taking contractions to superintegrable systems on the Darboux space~$D3$, on the Poincar\'e upper half plane
(equivalent to a system on
a hyperboloid in Euclidean space) and, f\/inally, to the isotropic oscillator system in Euclidean space,
 we will see that using ideas from superintegrability theory alone we can understand the basic properties of these systems:
 conservations laws, explicit trajectories, etc., and how they are related. Fig.~\ref{flowchart} describes the contraction relationships that we
 will exploit.  An important feature of degenerate systems is that they always admit~4 linearly independent symmetries that are 2nd degree in the momenta,
 whereas only 3 can be functionally independent. Thus there must be a relation between these symmetries. This relation emerges from the structure
 algebra obeyed by the symmetries.
 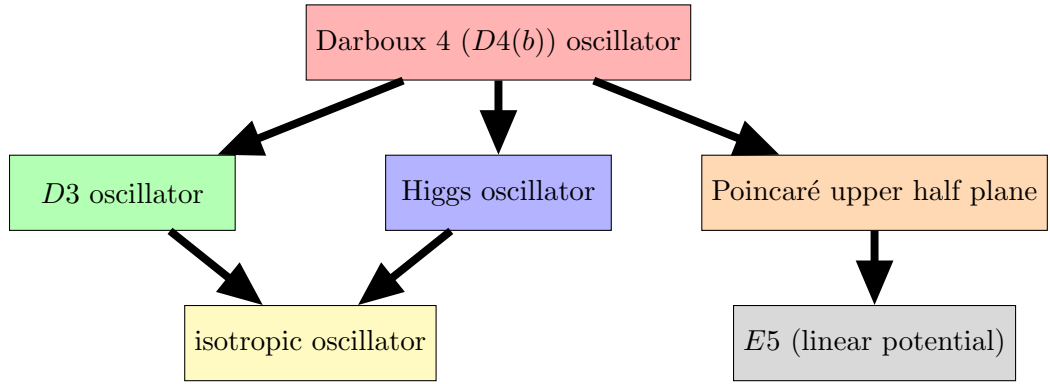
\begin{figure}[t!]
\centering
\begin{tikzpicture}[node distance=2cm]
\node (1) [a] {Darboux 4 ($D4(b)$) oscillator};
\node (3) [b, below of=1] {Higgs oscillator};
\node (2) [c, left of=3, xshift=-3cm] {$D3$ oscillator};
\node(4) [d, right of=3 ,xshift=3cm] {Poincar\'e upper half plane};
\node(5) [e, below of =2, xshift=2.5cm] {isotropic oscillator};
\node(6) [f, below of=4,] {$E5$ (linear potential)};
\draw [arrow] (1) -- (3);
\draw [arrow] (1) -- (2);
\draw [arrow] (1) -- (4);
\draw [arrow] (2) -- (5);
\draw [arrow] (3) -- (5);
\draw [arrow] (4) -- (6);
\end{tikzpicture}
\caption{A diagram showing the contractions of oscillators.}\label{flowchart}
\end{figure}

\subsection[The $D4(b)$ oscillator]{The $\boldsymbol{D4(b)}$ oscillator}\label{sec:D4}

We consider the superintegrable system
\begin{gather}\label{D4bham}
 {\cal H} =\frac{\sinh^2(2x)}{2\cosh(2x)+b}\big(p_x^2+p_y^2\big)+\frac{\alpha}{2\cosh(2x)+b},
 \end{gather}
with a basis of constants of the motion, ${\cal H}$, ${\cal J}=p_y$, and
\begin{gather}
{\cal Y}_1=-\cos(2y)\left(\frac{\sinh^2(2x)}{2\cosh(2x)+b}\big(p_x^2+p_y^2\big)-\cosh(2x)p_y^2\right)\nonumber\\
\hphantom{{\cal Y}_1=}{} -\sin(2y)\sinh(2x)p_xp_y -\frac{\alpha\cos(2y)}{4\cosh^2(x)-2+b}, \label{D4Y1}\\
 {\cal Y}_2 =-\sin(2y)\left(\frac{\sinh^2(2x)}{2\cosh(2x)+b}\big(p_x^2+p_y^2\big)-\cosh(2x)p_y^2\right)\nonumber\\
 \hphantom{{\cal Y}_2 =}{} +\cos(2y)\sinh(2x)p_xp_y -\frac{\alpha\sin(2y)}{4\cosh^2(x)-2+b},\label{D4Y2}
\end{gather}
Here the spaces $D4$ are indexed by a parameter $b>-2$. (In the limit as $b\to -2$ the space becomes a 2-sphere and this system becomes
the Higgs oscillator~\cite{Higgs}.) The variable $y$ can be interpreted as an angle, and the space and potential are periodic in $y$ with period~$\pi$.
The constants of the motion generate an algebra under the Poisson bracket obeying the structure equations
\begin{gather}
\label{D4strusture1}   \{{\cal J},{\cal H}\}=\{{\cal Y}_1,{\cal H}\}=\{{\cal Y}_2,{\cal H}\}=0,\\
\label{D4strusture2}   \{{\cal J},{\cal Y}_1\}=-2{\cal Y}_2,\qquad    \{{\cal J},{\cal Y}_2\}=2{\cal Y}_1,\qquad
\{{\cal Y}_1,{\cal Y}_2\}=4{\cal J}^3+2b{\cal J}{\cal H}-\alpha{\cal J},
\end{gather}
with Casimir
\begin{gather}\label{D4bCasimir} {\cal Y}_1^2+{\cal Y}_2^2={\cal J}^4+{\cal H}^2+b{\cal J}^2{\cal H}-\alpha{\cal J}^2.
\end{gather}
Note that there are 4 linearly independent symmetries ${\cal J}^2$, ${\cal H}$,  ${\cal Y}_1$,  ${\cal Y}_2$ as polynomials in the momenta, but there can be only 3 functionally independent generators. The dependence relation is given by~(\ref{D4bCasimir}).

From the f\/irst two equations (\ref{D4strusture2}) we see that
 $({\cal Y}_1, {\cal Y}_2)$ transforms like a 2-vector under rotations about the 3-axis. The sum  ${\cal Y}_1^2+{\cal Y}_2^2$ is expressed in terms of constants of motion so the sum is constant. We set ${\cal Y}_1^2+{\cal Y}_2^2\equiv \kappa^2$ where $\kappa$ is the length of the 2-vector. Thus we can  choose a preferred coordinate system such that ${\cal Y}_1=\kappa$ and ${\cal Y}_2=0$.

To determine a  trajectory, we need to express $x$ and $y$ in terms of the constants of the motion along the trajectory.
We do this by eliminating the momenta from equations  (\ref{D4bham})--(\ref{D4bCasimir}).  We see that ${\cal Y}_1$ and ${\cal Y}_2$
can be expressed in an alternate form:
\begin{gather}
  {\cal Y}_1=-\cos(2y){\cal H}+\cos(2y)\cosh(2x)p_y^2-\sin(2y)\sinh(2x)p_x p_y, \label{Y1} \\
  {\cal Y}_2=-\sin(2y){\cal H}+\sin(2y)\cosh(2x)p_y^2+\cos(2y)\sinh(2x)p_x p_y.  \label{Y2}
\end{gather}
To eliminate $p_x$, we multiply equation (\ref{Y1}) by $\cos(2y)$, equation~(\ref{Y2}) by~$\sin(2y)$ respectively, and add them together,
which gives us ${\cal Y}_1\cos(2y)+{\cal Y}_2\sin(2y)=-{\cal H}+\cosh(2x)p_y^2$.
Using the facts that ${\cal Y}_1=\kappa$, ${\cal Y}_2=0$ and ${\cal J}=p_y$, we get the orbit equation,
\begin{gather}
\cosh(2x){\cal J}^2-\cos(2y)\kappa={\cal H}.  \label{eq:D4orbit}
\end{gather}

Since $\cosh(2x)$ is an even function in $x$,  there will be two orbits for positive~$x$ and negative~$x$ that are symmetric with respect to $x$-axis.
Therefore,  without loss of generality, we restrict ourselves to positive~$x$.

{\bf Analog of Kepler's second law of planetary motion.} We exploit the fact that there is a 1st degree constant of the motion. We consider
an orbital  motion as taking place in the $s_1-s_2$    plane with polar coordinates $s_1=r\cos\theta$, $s_2=r\sin \theta $ where $r=(2\cosh(2x)+b)/(2\sinh{2x})$ and $\theta=2y$.
In these coordinates the metric is
\begin{gather*}
 ds^2 = 2 \frac{1+2br^2+\sqrt{16r^4+4br^2+1}}{(1+2br^2)^2}dr^2+r^2d\theta^2.
 \end{gather*}
  If $\alpha>0$ the potential is attractive to the origin in the plane;
 if $\alpha<0$ the potential is repulsive and the trajectories are unbounded.
\begin{Theorem}
Trajectories of the $D4(b)$ oscillator sweep out equal areas in equal times with respect to the origin in the $s_1-s_2$ plane.
\end{Theorem}

\begin{proof}
 In the  interval from some initial time $0$ to time $t$ the area swept out by the segment of the
 straight line connecting the origin and the object is $A(t) =\frac12\int_{\theta(0)}^{\theta(t)} r^2(\theta)d\theta$. Thus the rate at which the area is swept out is
\begin{gather}
\frac{dA}{dt}=\frac{dA}{d\theta}\frac{d\theta}{dt}= \frac12 r^2(\theta(t)) \frac{d\theta}{dt}=r^2\frac{dy}{dt},\label{eq:dA/dt}
\end{gather}
since
$dy/dt=(dy/d\theta)\cdot (d\theta/dt)=(1/2)\cdot (d\theta/dt)$.
To get $(dy/dt)$, we take the Poisson bracket of~${\cal H}$ and~$y$, which is
\begin{gather*}
\{{\cal H}, y\}=\frac{dy}{dt}=\frac{2p_y\sinh^2(2x)}{2\cosh(2x)+b}=\frac{2{\cal J}\sinh^2(2x)}{2\cosh(2x)+b},
\end{gather*}
then plug in this expression for $(dy/dt)$ into equation~(\ref{eq:dA/dt}), to get, $(dA/dt)={\cal J}/2$
which is a~constant.
\end{proof}

{\bf Analog of Kepler's third law of planetary motion.}
\begin{Theorem} For $\alpha>0$ the  period $T$ of an orbit is
 \begin{gather}
 T=\frac{1}{2}\pi\left(\frac{2+b}{\sqrt{(-2{\cal H}-b{\cal H}+\alpha)}}+\frac{2-b}{\sqrt{(2{\cal H}-b{\cal H}+\alpha)}}\right).
 \label{period2}
\end{gather}
\end{Theorem}

\begin{proof} The total area swept out as the trajectory goes through one period is
\begin{gather}
 A(T)   = \frac12\int_{0}^{2\pi} r^2(\theta)d\theta\label{eq:Area}\\
\hphantom{A(T)}{}  =\frac{1}{16}{\cal J}^2\pi\frac{(2+b)
\sqrt{-\kappa^2+{\cal H}^2+2{\cal J}^2{\cal H}+{\cal J}^4}+(2-b)\sqrt{-\kappa^2+{\cal H}^2-2{\cal J}^2{\cal H}+{\cal J}^4}}{\sqrt{(-\kappa^2+{\cal H}^2-2{\cal J}^2{\cal H}
+{\cal J}^4)(-\kappa^2+{\cal H}^2+2{\cal J}^2{\cal H}+{\cal J}^4)}}.\nonumber
\end{gather}
However, from the second law we see that $A(T)=(T/2)\cdot{\cal J}$ and, using the fact that
$\kappa^2={\cal J}^4+{\cal H}^2+b{\cal J}^2{\cal H}-\alpha{\cal J}^2$, the period can be expressed in the  form~(\ref{eq:Area}).
\end{proof}

 Now we begin an analysis of the trajectories. We f\/irst assume that $\alpha>0$ so that the potential is attractive.

{\bf Restriction on ${\cal H}$.}  Recall there is a restriction when we try to express $r$ in terms of $\theta$ and other constants of the motion. That is, $\cosh(2x)=({\cal H}+\cos\theta\kappa)/({\cal J}^2)$ should always be larger than $1$ for any $\theta$. Then since $\kappa>0$, we have the restriction ${\cal H}-\kappa>{\cal J}^2$
implying that for  a~closed trajectory,  ${\cal H}$ should always be positive.  Here $\kappa$ is the length of the 2-vector $({\cal Y}_1, {\cal Y}_2)$.
 To be explicit,
$\kappa^2={\cal J}^4+{\cal H}^2+b{\cal J}^2{\cal H}-\alpha{\cal J}^2$. Then by squaring both sides of an alternate form of the restriction equation, ${\cal H}-{\cal J}^2>\kappa$, and plugging in the expression for $\kappa^2$, we get
$(\alpha-(b+2){\cal H}){\cal J}^2>0$. Noticing the fact that $b+2$, ${\cal H},$ and ${\cal J}^2$ are all larger than $0$, we reach  the conclusion,
${\cal H}<\alpha/(b+2)$
for a bounded orbit. For larger values of ${\cal H}$ the trajectory is unbounded.

{\bf Case of bounded trajectory.} To plot the trajectories on the $s_1-s_2$ plane it is convenient to write $s_1$ and $s_2$ both in terms of $\theta$,
which is
\begin{gather} s_1=\frac{1}{2}\sqrt{\frac{{\cal J}^2(2{\cal H}+2\cos\theta\cdot\kappa+b{\cal J}^2)}{({\cal H}+\cos\theta\cdot\kappa)^2-{\cal J}^4}}
 \cos\theta,\nonumber\\
  s_2=\frac{1}{2}\sqrt{\frac{{\cal J}^2(2{\cal H}+2\cos\theta\cdot\kappa+b{\cal J}^2)}{({\cal H}+\cos\theta\cdot\kappa)^2-{\cal J}^4}}
  \sin\theta. \label{eq:D4s1}
 \end{gather}
(Note: it is useful to divide the  $r$ part of $s_1$ and $s_2$ by ${\cal J}^4$ and introduce new constants ${\cal H'}={\cal H}/{\cal J}^2$
and $\kappa'=\kappa/{\cal J}^2$ so we can eliminate one constant.)
Since in $s_1$ and $s_2$, ${\cal J}$ always appears in the form of~${\cal J}^2$, so for every positive~${\cal J}$,
there will always be a duplicate case for the correspon\-ding~$-{\cal J}$, and  without loss of generality, we can restrict our discussion to ${\cal J}>0$.

A typical plot of the trajectory on the $s_1$-$s_2$ plane is  Fig.~\ref{6.pdf},
\begin{figure}[t!]\centering
 \includegraphics[width=70mm]{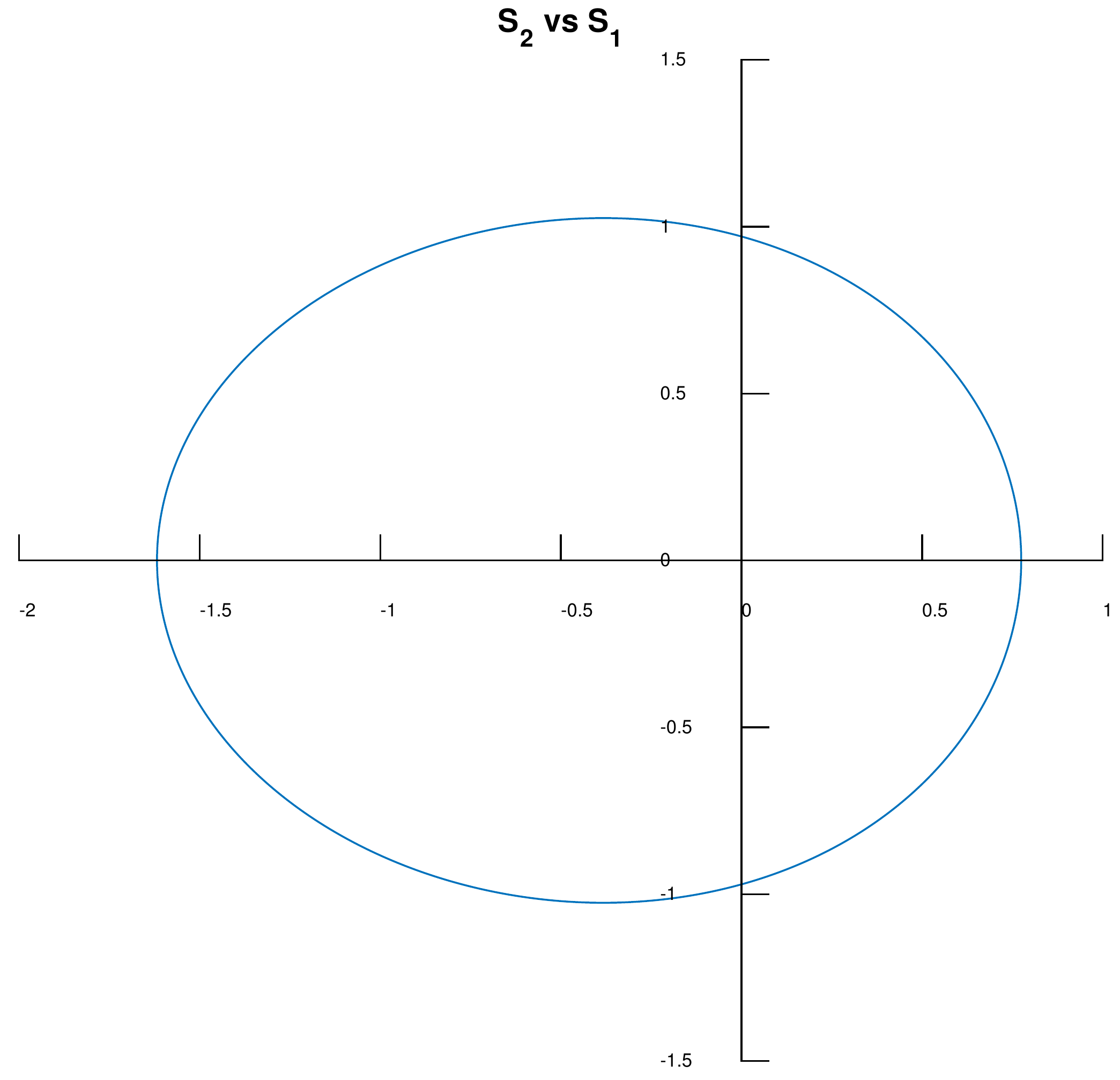}
\caption{Orbit plot for ${\cal J}=10$, ${\cal H}=130$, $\kappa=20$, $b=0$ (equal axes).}\label{6.pdf}
\end{figure}
and as one or more of ${\cal H}$, ${\cal J}$ and~$b$ becomes larger and larger, they will dominate the~$r$ term and make~$r$ less
susceptible to the change of $\cos\theta$. Thus the plots will look more and more circular.
Furthermore, if  $\kappa$ becomes very large, the plots will tend to move towards the negative~$s_1$ direction and appear in an elongated
form as in Fig.~\ref{8.pdf}.
\begin{figure}[t!]\centering
 \includegraphics[width=80mm]{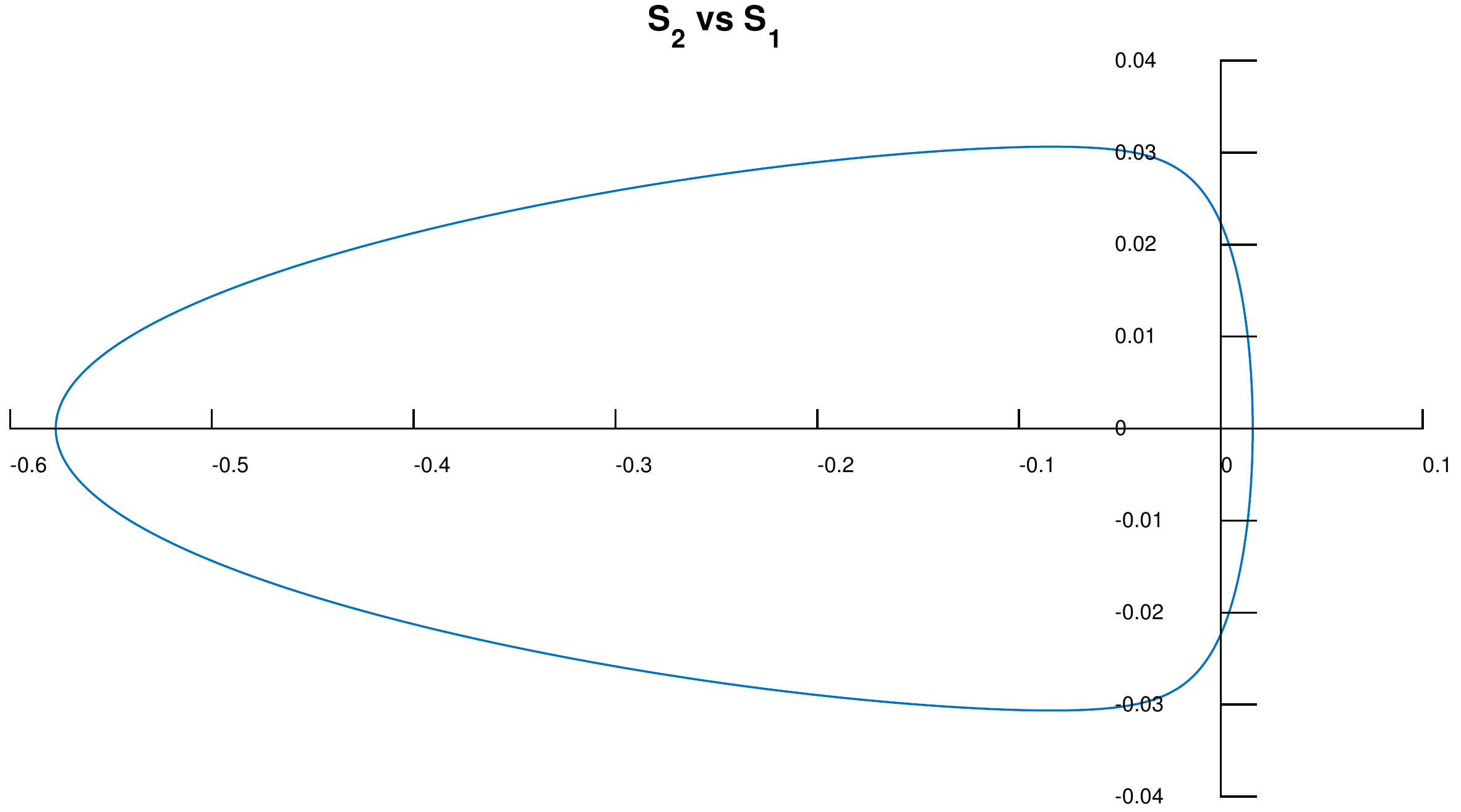}
\caption{Orbit plot for ${\cal J}=1$, ${\cal H}=1002$, $\kappa=1000$, $b=0$ (axes are not set equal).}\label{8.pdf}
\end{figure}
The plot will always be symmetric with respect to $s_1$-axis,  because $s_1(\theta)$ is even  and $s_2(\theta)$ is odd in $\theta$.
Also, since the plot will always lean towards the negative~$s_1$ side, larger~$r$ always occurs at smaller~$s_1$ value.

To trace out an orbit in time, we would need a starting point to integrate Hamilton's equations. Conventionally, we choose a point that is
closest to the origin, which we call perihelion.
From the plots of the orbits and an easy analysis of  equations~(\ref{eq:D4s1}), it is obvious that this
point is the intersection of the trajectory
with the positive $s_1$ axis. The point on the $x$-$y$ plane that corresponds to this  perihelion is $(x_0, 0)$ where
$x_0=\frac{1}{2} \arccosh(\frac{\kappa+{\cal H}}{{\cal J}^2})$.  Plugging $y=0$ into equation (\ref{Y2}) gives $p_xp_y=0$. Realizing $p_y={\cal J}\neq0$, we
know $p_x$ must be equal to $0$.  We see that the perihelion points on the phase space trajectory are uniquely determined by the
constants of the motion as follows:
\begin{gather}
\cosh(2x)=\frac{{\cal H}+\kappa}{{\cal J}^2}, \qquad y=0, \qquad p_x=0, \qquad p_y={\cal J},\label{eq:perihelion}
\end{gather}
and we can see that if we know the perihelion point on the phase space as in equation~(\ref{eq:perihelion}), a~set of
constants of the motion can be uniquely determined.

{\bf Case of escape velocity.} This is the case when ${\cal H}={\cal J}^2+\kappa$,
 a bifurcation point on the momentum map. In polar coordinates, as  $\theta\to \pi$ we have $r\to\infty$, whereas for $s_1$, $s_2$ in
 equations~(\ref{eq:D4s1}), $s_1\to-\infty$  and $s_2\to 0$.
 A plot for  $s_2$ {\it vs.} $s_1$ is given in Fig.~\ref{12.pdf}.
\begin{figure}[t!]\centering
 \includegraphics[width=80mm]{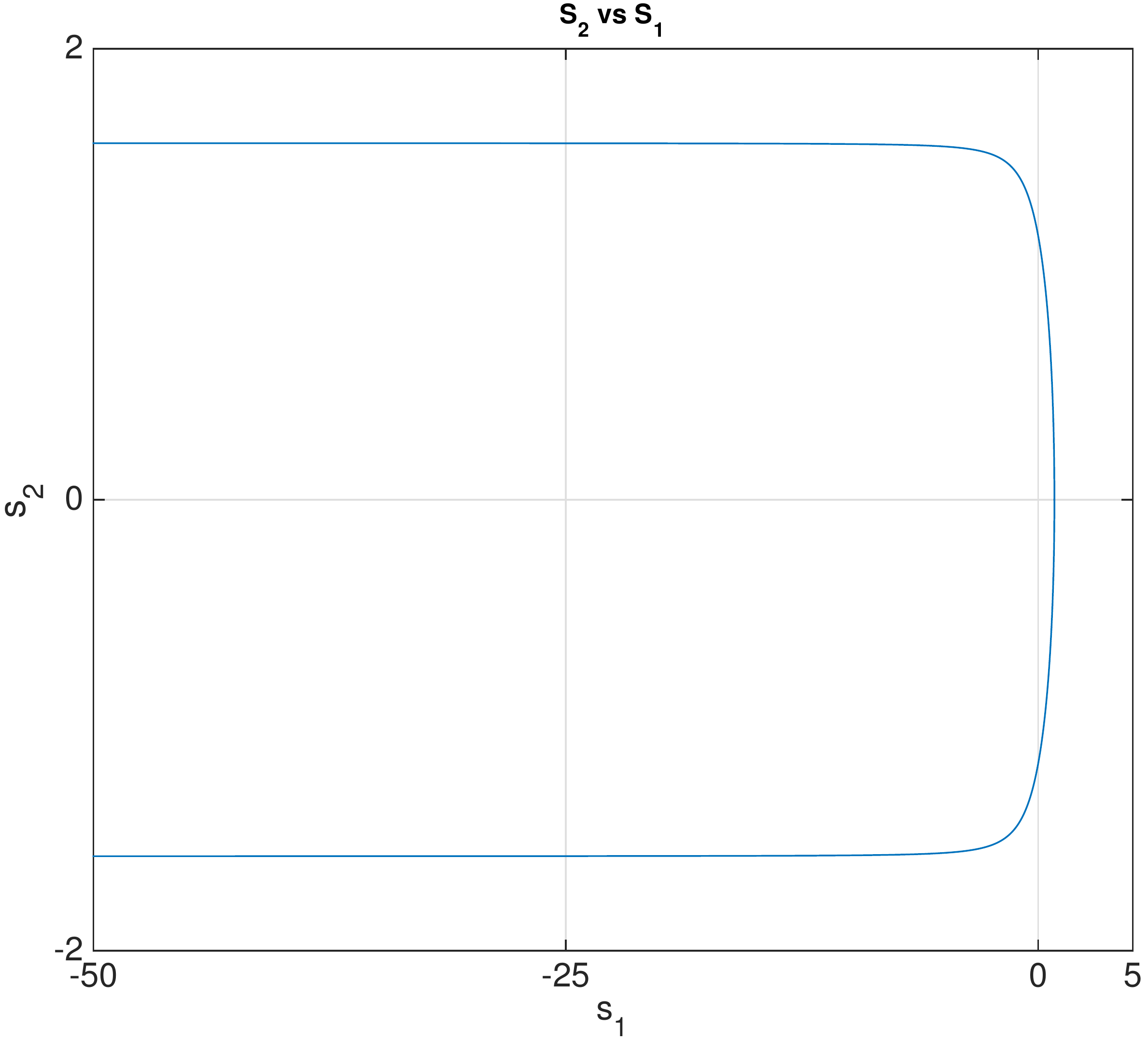}
\caption{$s_2$ \textit{vs} $s_1$ for ${\cal J}=10$, ${\cal H}=120$, $\kappa=20$, $b=0$.}\label{12.pdf}
\end{figure}

In Fig.~\ref{12.pdf}, the two tails will extend to inf\/inity in the direction of negative $s_1$-axis, and they will get closer and closer to the $s_1$-axis for smaller and smaller $s_1$ value but they will never touch the $s_1$-axis.

{\bf Case of unbounded trajectory.} Here, there is more than one value of $y$, i.e., $\theta/2$ that has no corresponding real value of~$x$
in  equation~(\ref{eq:D4orbit}).

\begin{Example}[${\cal H}\le{\cal J}^2$] In this case, since ${\cal H}\le{\cal J}^2$, $r$ will blow up before $s_1$ goes becomes negative,
so the entire trajectory will be bounded on the positive-$s_1$ side of the plane. A boundary case (${\cal H}={\cal J}^2$) is plotted in
Fig.~\ref{14.pdf}.
\begin{figure}[t!]\centering
 \includegraphics[width=80mm]{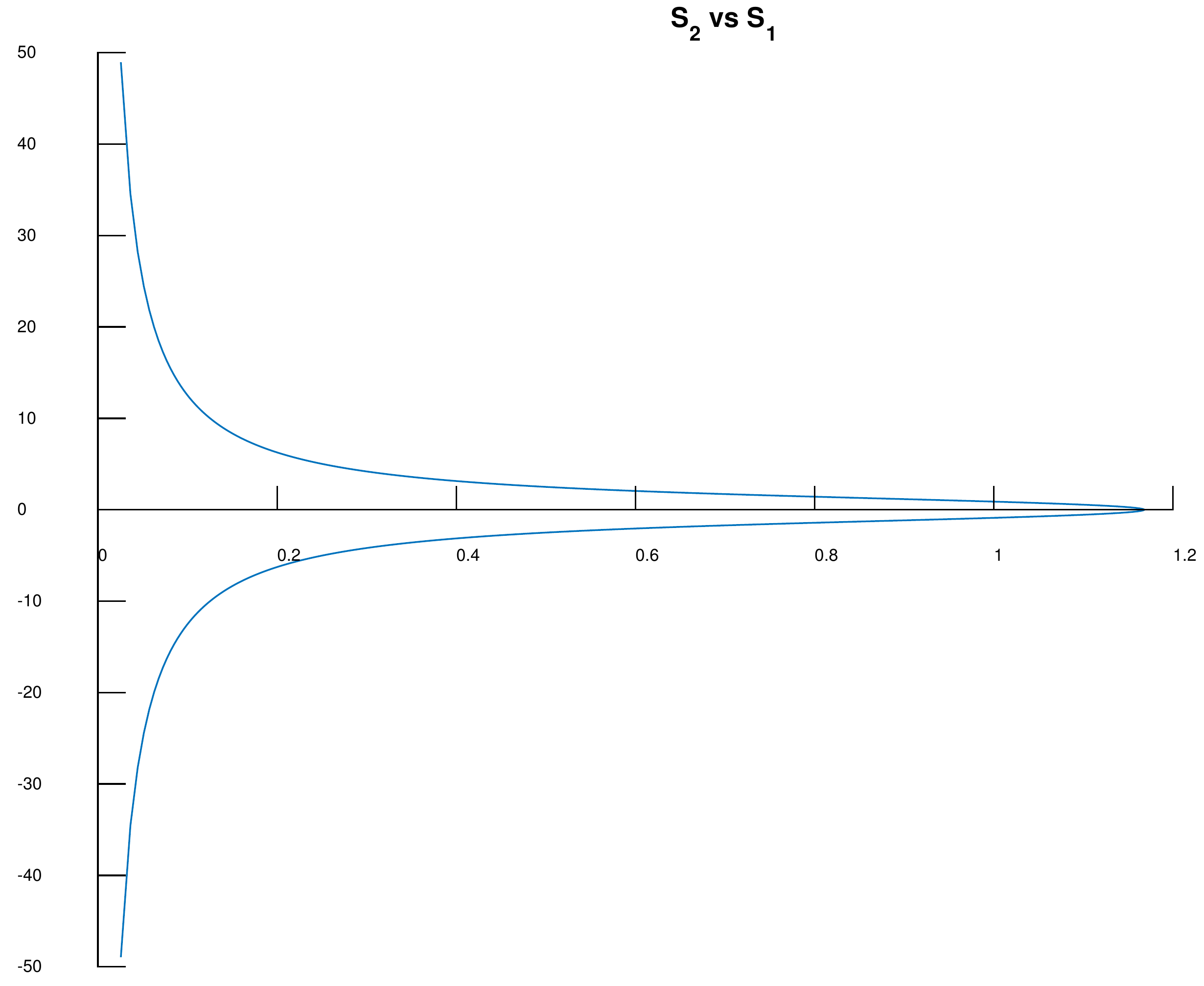}
\caption{$s_2$ \textit{vs} $s_1$ for ${\cal J}=10$, ${\cal H}=100$, $\kappa=20$, $b=0$.}\label{14.pdf}
\end{figure}
For Fig.~\ref{14.pdf} the constants are ${\cal J}=10$, ${\cal H}=100$, $\kappa=20$, $b=0$, so for
$0\le\theta<\pi/2$ and $3\pi/2<\theta\le 2\pi$, corresponding real values of $x$ exist and there is a trajectory.
For $\theta=\pi/2$ and $3\pi/2$, $r$ will goes to inf\/inity and at the same time, $s_2$ will go to inf\/inity and $s_1$  to $0$ which
is represented by the two tails in Fig.~\ref{14.pdf}.
\end{Example}

\begin{Example}[${\cal J}^2 \le {\cal H} <{\cal J}^2+\kappa$] In this case, the trajectory will extend to the negative-$s_1$ side of
the plane. A typical plot for an orbit is given in Fig.~\ref{15.pdf}.
\begin{figure}[t!]\centering
 \includegraphics[width=85mm]{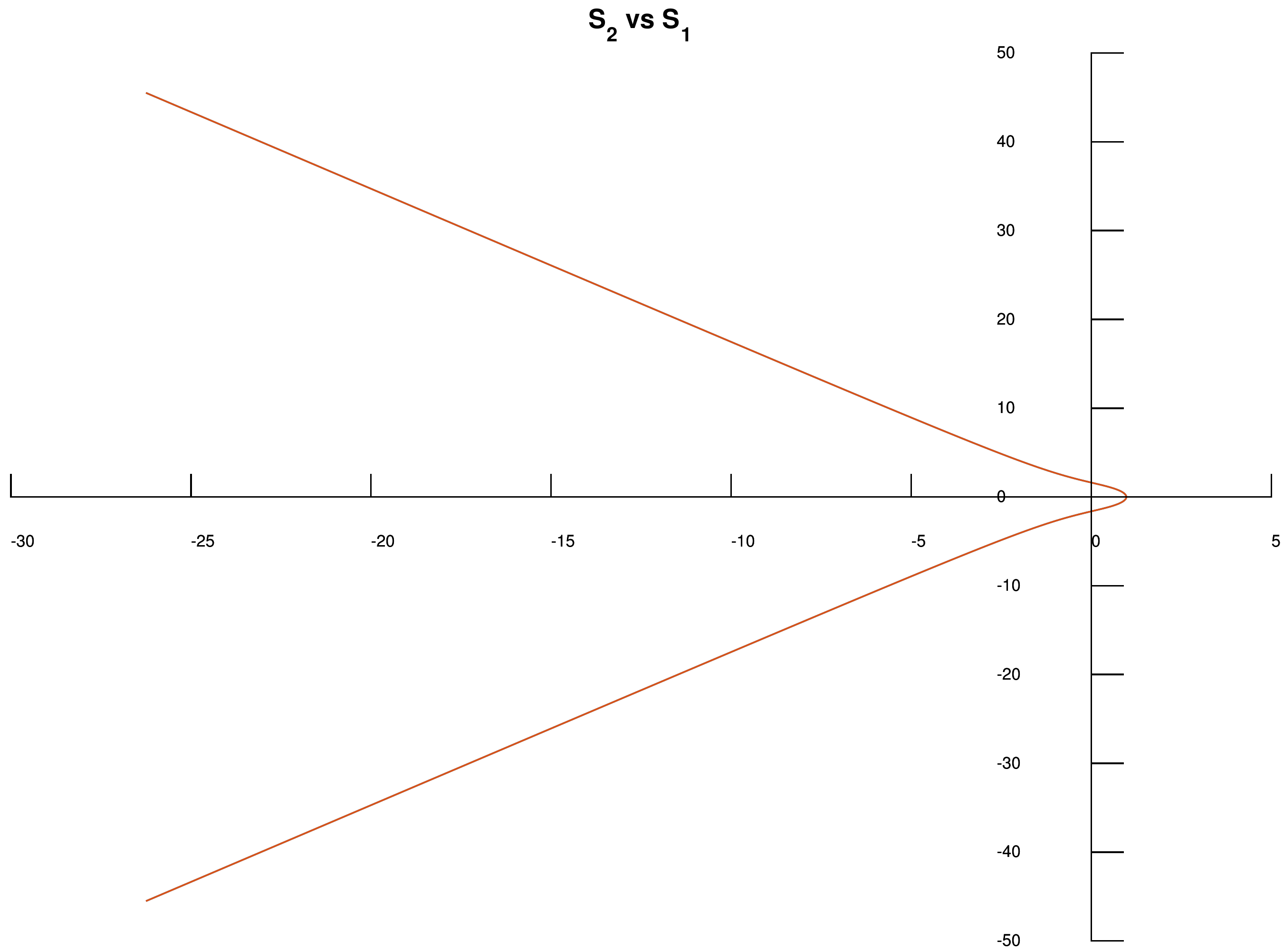}
\caption{$s_2$ \textit{vs} $s_1$ for ${\cal J}=10$, ${\cal H}=110$, $\kappa=20$, $b=0$.}\label{15.pdf}
\end{figure}
\end{Example}

\subsection{Contraction to the Poincar\'e upper half plane}

Using the Hamiltonian ${\cal H}$, (\ref{D4bham}) modif\/ied by a constant as
${\cal H}'={\cal H}- [\alpha/(2+b) ]$,
we let $x=\epsilon Y$, $y=\epsilon X$, $\alpha=- [(2+b)\beta ]/[2\epsilon^2]$,
${\cal J}={\cal K}/\epsilon$, ${\cal Y}_1-(2{\cal J}^2)/(\epsilon^2)+{\cal H}\approx 2{\cal X}_1$, ${\cal Y}_2\approx (2{\cal X}_2)/(\epsilon)$,
and go to the limit as $\epsilon\to 0$.
Then
\begin{gather*}
\begin{aligned}
&{\cal H}'=\frac{4Y^2}{b+2}\big(p_X^2+p_Y^2\big)+\beta Y^2,  &  &{\cal K}=p_X,\\
&{\cal X}_1=\big(X^2-Y^2\big)p_X^2-2XYp_Xp_Y -\frac{\beta X^2}{4}(b+2), &
&{\cal X}_2=Xp_X^2+Yp_Xp_Y+\frac{\beta(b+2)}{4} X.
\end{aligned}
\end{gather*}
The structure equations become
\begin{gather*}
 \{{\cal K},{\cal X}_1\}=-2{\cal X}_2,\qquad \{{\cal K},{\cal X}_2\}={\cal K}^2+\frac{\beta}{4}(b+2),\qquad
\{{\cal X}_1,{\cal X}_2\}=-2{\cal K}{\cal X}_1+\frac{b+2}{2}{\cal K}{\cal H}',
\end{gather*}
and the Casimir is
\begin{gather}\label{PoincareCasimir}
{\cal X}_2+(\frac{\beta(b+2)}{4}+{\cal K}^2){\cal X}_1-\frac{b+2}{4}{\cal K}^2{\cal H}'=0.
\end{gather}
If the potential is turned of\/f this is essentially the Poincar\'e upper half plane model of hyperbolic geometry; for $b=4$ it is exactly that.
The space here consists of all real  points $(X,Y)$ with $Y>0$. In this limit there is no longer any periodicity.
Using these identities to eliminate the momenta we arrive at the equation for the trajectories:
\begin{gather}\label{Poincaretrajectories}
\left(\frac{4{\cal K}^2+\beta(b+2)}{2}X-2{\cal X}_2\right)^2+{\cal K}^2\big(4{\cal K}^2
+\beta(b+2)\big)Y^2-{\cal H}'{\cal K}^2(b+2)=0.
\end{gather}
Notice that the orbit equation is even in $Y$, so the trajectory will always be symmetric with respect to the $X$-axis.
When restricting to the Poincar\'e upper half plane with $Y>0$, the trajectory will be {\it the upper half of the trajectory traced out by the orbit equation}.
If $\beta>0$, so that the potential is attractive to the boundary $Y=0$,  this describes the portion of the  ellipse $(x^2/A^2)+(y^2/B^2)=1$ in the upper half plane, where
\begin{gather*}
 A^2=\frac{4{\cal H}'{\cal K}^2(b+2)}{(4{\cal K}^2+\beta(b+2))^2},\qquad B^2=\frac{{\cal H}'(b+2)}{4{\cal K}^2+\beta(b+2)},\\ x=X-\frac{4{\cal X}_2}{4{\cal K}^2+\beta(b+2)},\qquad  y=Y.
 \end{gather*}

If $\beta\le0$ the potential is repulsive. If however, $4{\cal K}^2+(b+2)\beta>0$ the trajectory~(\ref{Poincaretrajectories}) is again a portion of an ellipse. There is a special case that when $\beta=0$, we will have $A^2=B^2$ such that the trajectory will be upper half of a~circle. A plot of a trajectory when $\beta=0$ is given in Fig.~\ref{27.pdf}. And an example of elliptic trajectory is given in Fig.~\ref{28.pdf}.
\begin{figure}[t!]\centering
 \includegraphics[width=90mm]{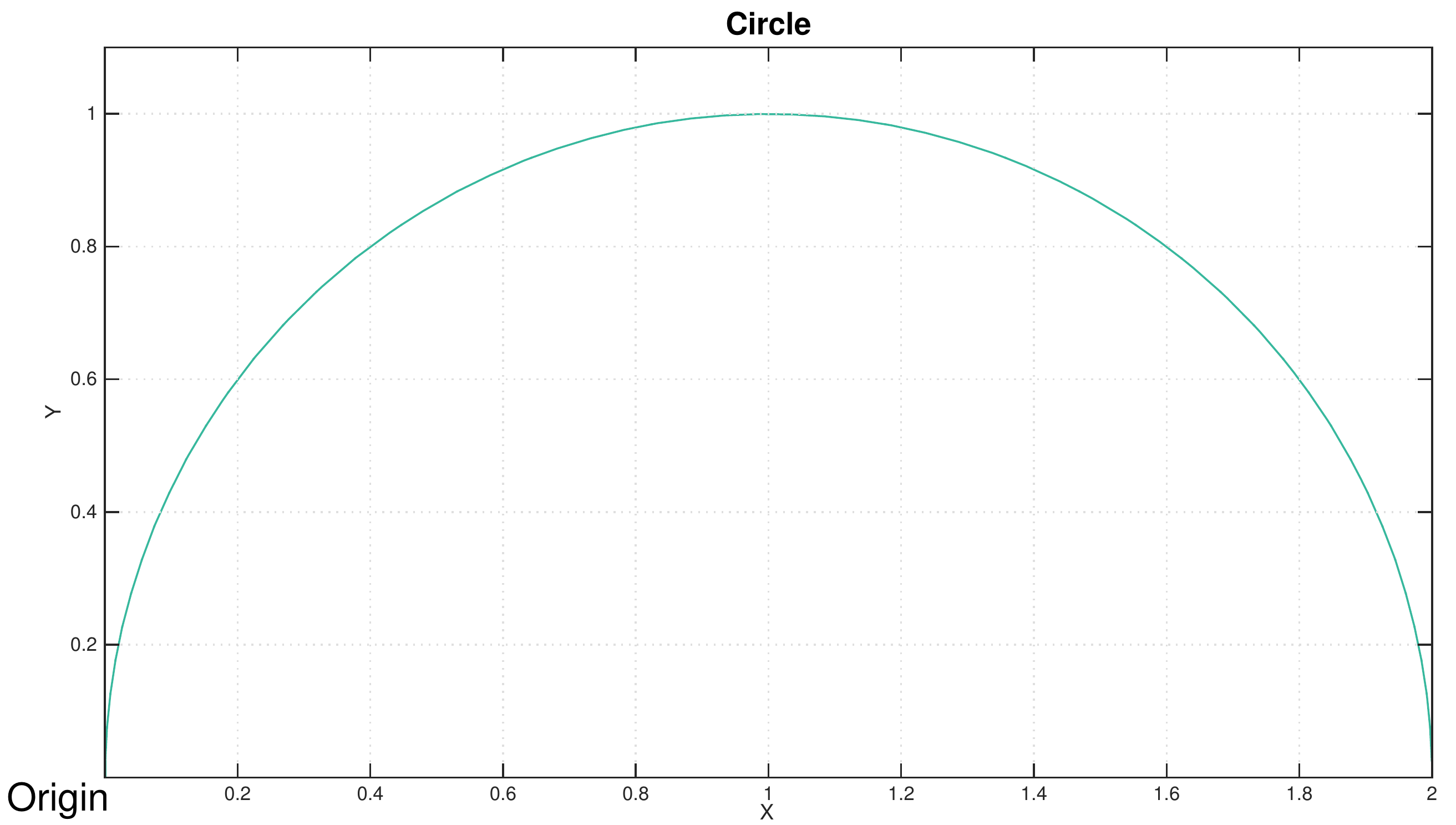}
\caption{Circular trajectory on Poincar\'e upper half plane with $\beta=0$, $b=0$, ${\cal H'}=2$, ${\cal K}=1$, ${\cal X}_2=1$.}\label{27.pdf}
\end{figure}
\begin{figure}[t!]\centering
 \includegraphics[width=80mm]{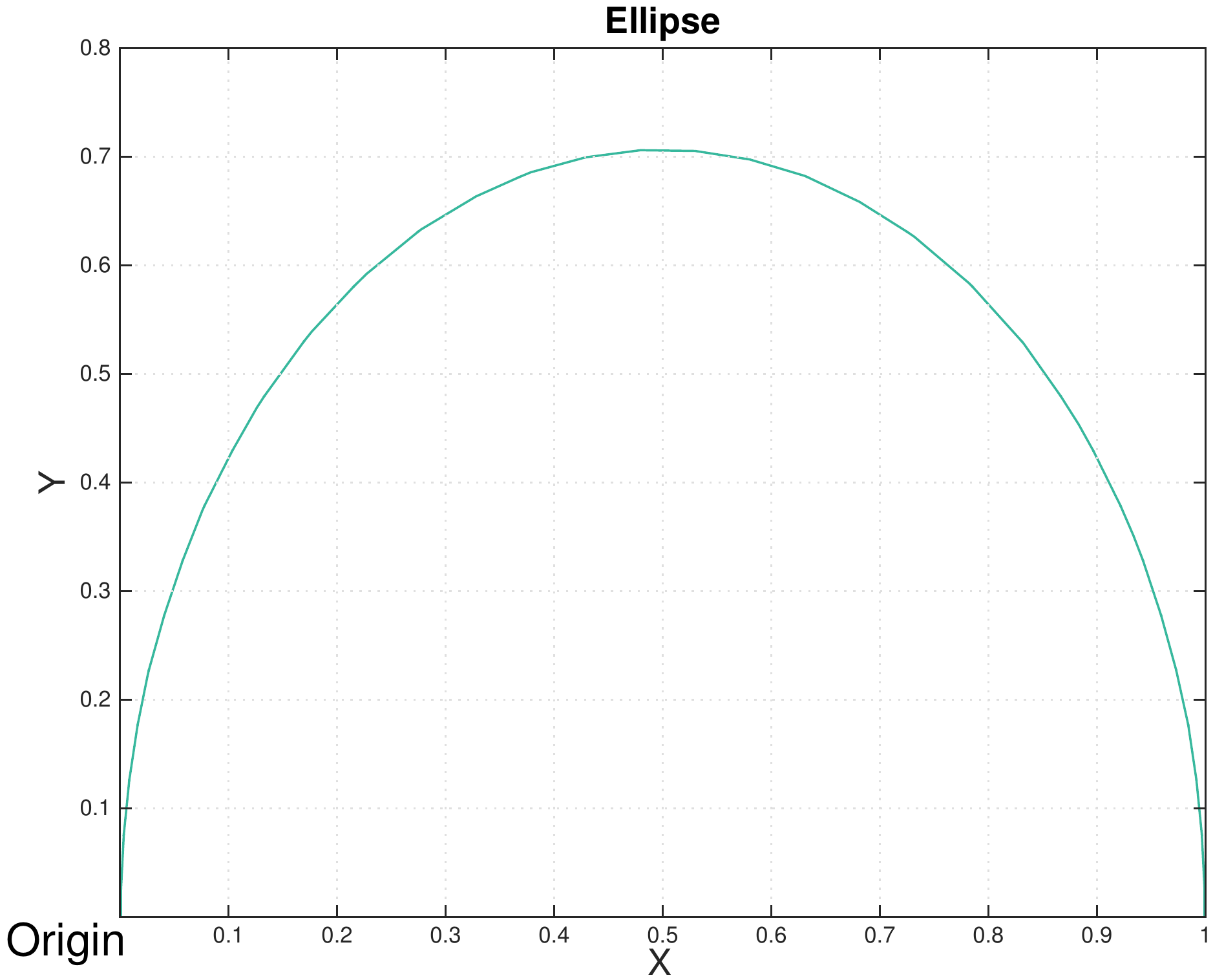}
\caption{Elliptic trajectory on Poincar\'e upper half plane with $\beta=2$, $b=0$, ${\cal H'}=2$, ${\cal K}=1$, ${\cal X}_2=1$.}\label{28.pdf}
\end{figure}

If $4{\cal K}^2+(b+2)\beta<0$   the trajectory is the upper sheet $y>0$ of the hyperbola   $(x^2/A^2)-(y^2/{B'}^2)=1$ where
\begin{gather*}
 A^2=\frac{4{\cal H}'{\cal K}^2(b+2)}{(4{\cal K}^2+\beta(b+2))^2},\qquad { B'}^2=-\frac{{\cal H}'(b+2)}{4{\cal K}^2+(b+2)\beta},\\
x=X-\frac{4{\cal X}_2}{4{\cal K}^2+\beta(b+2)},\qquad
  y=Y.
  \end{gather*}
\begin{figure}[t!]\centering
 \includegraphics[width=85mm]{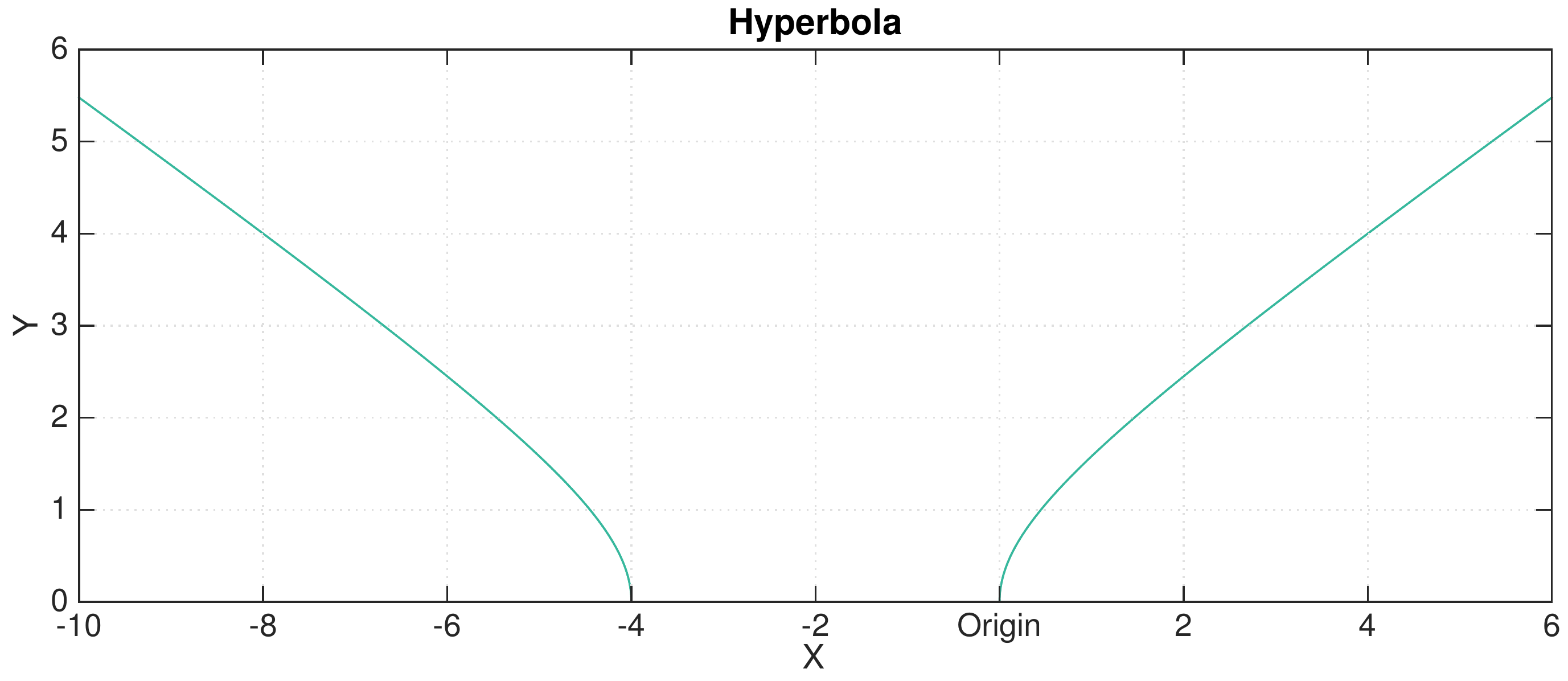}
\caption{Hyperbolic trajectory on Poincar\'e upper half plane with $\beta=-3$, $b=0$, ${\cal H'}=2$, ${\cal K}=1$, ${\cal X}_2=1$.}\label{29.pdf}
\end{figure}

If $4{\cal K}^2+(b+2)\beta=0$  equation~(\ref{Poincaretrajectories}) becomes degenerate. From~(\ref{PoincareCasimir}) we now f\/ind
${\cal X}_1=-Y^2{\cal K}^2-2XYp_Y{\cal K}$ and  have ${\cal H}'= [4Y^2/(b+2) ] \cdot p_Y^2\geq0
$. Eliminating  $p_Y$ from these two expressions we f\/ind the equation
$({\cal K}^2Y^2+{\cal X}_1)^2-{\cal H}'{\cal K}^2(b+2)X^2=0$
for the unbounded trajectories. These are upper halves of parabolas with orbit equation
$ Y^2= -({\cal X}_1/{\cal K})\pm \sqrt{{\cal H}'(b+2)}X$,
and we notice that a certain subset of ${\cal H'}$, ${\cal X}_1$, ${\cal K}$ and $b$ would correspond to two half-parabolas which are symmetric with respect to $y$-axis.
\begin{figure}[t!]\centering
 \includegraphics[width=90mm]{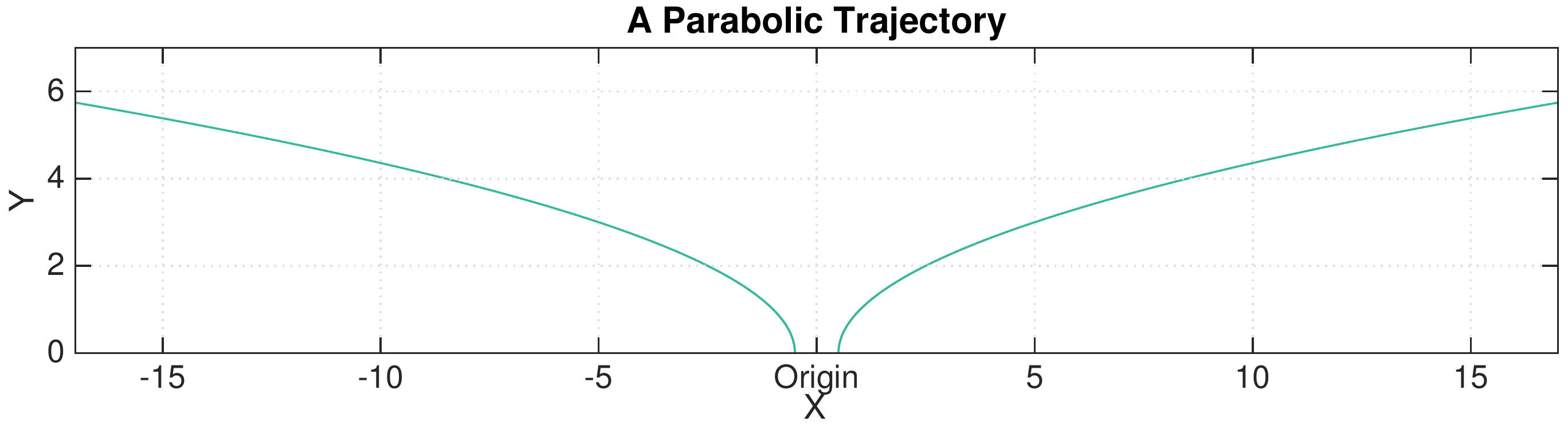}
\caption{Parabolic trajectory on Poincar\'e upper half plane with $\beta=-2$, $b=0$, ${\cal H'}=2$, ${\cal K}=1$, ${\cal X}_1=1$.}\label{30.pdf}
\end{figure}
Interestingly, when ${\cal X}_1/{\cal K}\le0$, the two half-parabolas will cross  as in Fig.~\ref{31.pdf}.
\begin{figure}[t!]\centering
 \includegraphics[width=90mm]{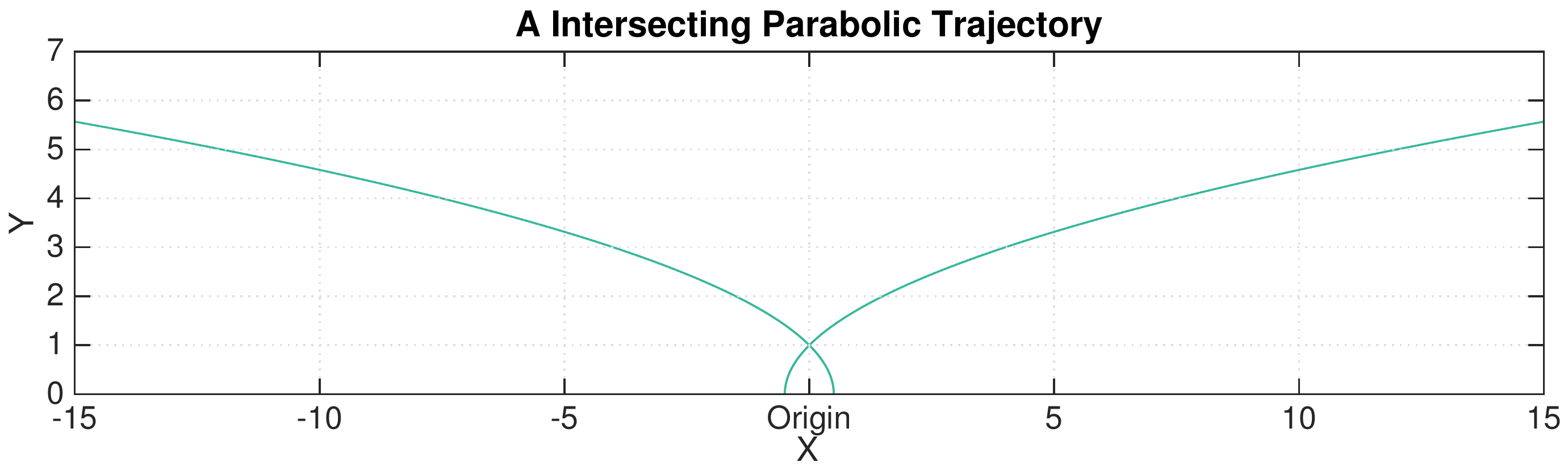}
\caption{Parabolic trajectory on Poincar\'e upper half plane with $\beta=-2$, $b=0$, ${\cal H'}=2$, ${\cal K}=-1$, ${\cal X}_1=1$.}\label{31.pdf}
\end{figure}
A special case occurs when ${\cal H'}=0$ so that the trajectory becomes a straight line.  Since  ${\cal H}'=(4Y^2\cdot p_Y^2)/(b+2)$, we see
this implies that $p_Y=0$ so the trajectory must be parallel to the $x$-axis.
Furthermore, if ${\cal H'}=0$ and ${\cal X}_1/{\cal K}=0$, the trajectory will be
exactly the $y$-axis. Moreover, for ${\cal H'}=0$ and ${\cal X}_1/{\cal K}<0$, there will be no trajectory.

We mention that this  Poincar\'e upper half plane model can be mapped to the unit disk, the Poincar\'e disk, as well as to the upper sheet of the 2-sheet hyperboloid is 3-space.
Further, this model can be contracted to the f\/lat space system ${\cal H}=p_x^2+p_y^2+\alpha x$, $E5$ in our listing~\cite{KKMP}. Indeed, if we set
 \begin{gather*}
 X=y,\qquad Y=x+\frac{\sqrt{b+2}}{2\epsilon},\qquad \beta =\frac{\alpha  \sqrt{b+2}}{\epsilon},\qquad     {\cal H}(\epsilon)=\epsilon^2\left({\cal H}'-\frac{\alpha (b+2)^{3/2}}{4\epsilon}\right),
 \end{gather*}
then $\lim\limits_{\epsilon\to 0} {\cal H}(\epsilon)=p_x^2+p_y^2+\alpha x$. We don't go into detail here but we give another detailed example in the next section.

\subsection[Contraction to the  $D3$ oscillator]{Contraction to the  $\boldsymbol{D3}$ oscillator}
The oscillator in the space Darboux 3 ($D_3$) has  Hamiltonian
\begin{gather}\label{Darboux3a} {\cal H}=\frac12\frac{e^{2x}}{e^x+1}\big(p_x^2+p_y^2\big)+\frac{\beta}{e^x+1}
\end{gather}
with a single Killing vector ${\cal J}=p_y$ and a symmetry algebra basis, $\{{\cal H},{\cal J}^2,{\cal X}_1,{\cal X}_2\}$. Here,
\begin{gather*}
 {\cal X}_1=\frac12 e^x\sin y  p_xp_y+\frac14\frac{e^{2x}}{e^x+1}\cos y  p_x^2-\frac14 \frac{e^x(e^x+2)}{e^x+1}\cos y  p_y^2+\frac12\frac{\beta\cos y}{e^x+1},\\
{\cal X}_2=-\frac12 e^x\cos y  p_xp_y+\frac14\frac{e^{2x}}{e^x+1}\sin y  p_x^2-\frac14 \frac{e^x(e^x+2)}{e^x+1}\sin y  p_y^2+\frac12\frac{\beta\sin y}{e^x+1}.
\end{gather*}
The structure relations are
\begin{gather*}
 \{{\cal J},{\cal X}_1\}=-X_2,\qquad \{{\cal J},{\cal X}_2\}={\cal X}_1,\qquad \{{\cal X}_1,{\cal X}_2\}=\frac12 {\cal J}{\cal H}-\frac{\beta}{2}{\cal J},
 \end{gather*}
 and there is the functional relation
$4{\cal X}_1^2+{\cal X}_2^2-\frac14{\cal H}^2-\frac12 {\cal J}^2{\cal H}+\frac{\beta}{2}{\cal J}^2=0$.
\begin{Theorem} The $D3$ oscillator system is a contraction of the $D4$ oscillator system.
\end{Theorem}

\begin{proof}
We can get (\ref{Darboux3a}) as a limiting case of (\ref{D4bham}) by taking
 \begin{gather}\label{D3Contra}
 x=\frac{x'}{2}-\frac12\ln(\epsilon),\qquad y=\frac{y'}{2},\qquad b=\frac{1}{\epsilon},\qquad {\cal H}'=8\epsilon{\cal H},\qquad
  \alpha=\frac{\beta}{8\epsilon^2}.
  \end{gather}
Then we have ${\cal H}'=\frac12  \frac{e^{2x'}}{e^{x'}+1}(p_{x'}^2+p_{y'}^2)$,
${\cal J}'=2{\cal J}=p_{y'}$, and ${\cal Y}'_1\approx -4\epsilon {\cal Y}_1$, ${\cal Y}'_2\approx -4\epsilon {\cal Y}_2$.

First, it is obvious that ${\cal J}'=2{\cal J}=p_{y'}$. Then as $\epsilon \to 0$, since
\begin{gather*}
\cosh(2x)=\frac{e^{x'-\ln(\epsilon)}+e^{-x'+\ln(\epsilon)}}{2}=\frac{e^{x'}+\epsilon^2e^{-x'}}{2\epsilon},\qquad \sinh(2x)=\frac{e^{x'}-\epsilon^2e^{-x'}}{2\epsilon},
\end{gather*}
plugging these two identities into the Hamiltonian equation (\ref{D4bham}) and also using equation (\ref{D3Contra}), we have
\begin{gather*}
{\cal H}'= 8\epsilon\left(\frac{\frac{e^{2x'}-2\epsilon^2+\epsilon^4e^{-2x'}}{4\epsilon^2}}{\frac{e^{x'}+\epsilon^2e^{-x'}}{\epsilon}+\frac{1}{\epsilon}}\right)(\frac14 p_{x'}^2+\frac14 p_{y'}^2)+8\epsilon\left(\frac{\frac{\beta}{8\epsilon^2}}{\frac{e^{x'}+\epsilon^2e^{-x'}}{\epsilon}+\frac{1}{\epsilon}}\right) \\
\hphantom{{\cal H}'}{}
 =\frac12\frac{e^{2x'}-2\epsilon^2+\epsilon^4e^{-2x'}}{e^{x'}+\epsilon^2e^{-x'}+1}\big(p_{x'}^2+p_{y'}^2\big)+\frac{\beta}{e^{x'}+\epsilon^2e^{-x'}+1}
=\frac12  \frac{e^{2x'}}{e^{x'}+1}\big(p_{x'}^2+p_{y'}^2\big)+\frac{\beta}{e^x+1},
\end{gather*}
because the terms with $\epsilon$ in front of them will all go to $0$ as $\epsilon\to 0$. This ${\cal H}'$ is exactly the Hamiltonian of the Darboux-3 oscillator.

Similarly, by substituting the identities in equation (\ref{D3Contra}) into equations (\ref{Y1}) and (\ref{Y2}), we have,
\begin{gather*}
{\cal Y}'_1=-4\epsilon(-\cos y')\frac{e^{2x'}-2\epsilon^2+\epsilon^4e^{-2x'}}{4\epsilon(e^{x'}+\epsilon^2e^{-x'}+1)}
\left(\frac14 p_{x'}^2+\frac14 p_{y'}^2\right)-\cos y'\frac{\beta}{8\epsilon(e^{x'}+\epsilon^2e^{-x'}+1)}\\
\hphantom{{\cal Y}'_1=}{}
 +\cos y'\left(\frac{e^{x'}+\epsilon^2e^{-x'}}{2\epsilon}\right)\frac14 p_{y'}^2-\sin y'
 \left(\frac{e^{x'}-\epsilon^2e^{-x'}}{2\epsilon}\right)\frac14 p_{x'}p_{y'}
=\frac14 \frac{e^{2x'}}{e^{x'}+1}\cos(y')p_{x'}^2\\
\hphantom{{\cal Y}'_1=}{}
-\frac14 \left(\frac{2e^{x'}(e^{x'}+1)}{e^{x'}+1}-\frac{e^{2x'}}{e^{x'}+1}\right)\cos(y')p_{y'}^2
+\frac12 e^{x'}\sin(y')p_{x'}p_{y'}+\frac12\frac{\beta\cos y'}{e^{x'}+1}\\
\hphantom{{\cal Y}'_1}{}
=\frac14 \frac{e^{2x'}}{e^{x'}+1}\cos(y')p_{x'}^2-\frac14 \frac{e^{x'}(e^{x'}+2)}{e^{x'}+1}\cos(y')\ p_{y'}^2+\frac12 e^{x'}\sin(y')p_{x'}p_{y'}+\frac12\frac{\beta\cos y'}{e^{x'}+1},
\end{gather*}
as $\epsilon\to 0$.
In the same sense,
\begin{gather*}
{\cal Y}'_2=\frac14 \frac{e^{2x'}}{e^{x'}+1}\sin(y')p_{x'}^2-\frac14 \frac{e^{x'}(e^{x'}+2)}{e^{x'}+1}\sin(y')\ p_{y'}^2
 -\frac12 e^{x'}\cos(y')p_{x'}p_{y'}+\frac12\frac{\beta\sin y'}{e^{x'}+1}.
 \end{gather*}
We notice they are exactly the same vectors as ${\cal X}_1$ and ${\cal X}_2$. So we have proved that in this limiting case, Darboux-4 space contracts to Darboux-3.
\end{proof}

Using this transformation, we can get the orbit equation for the $D3$ oscillator from the orbit equation for the $D4(b)$ oscillator. First,
since $\kappa'^2={\cal Y'}_1^2+{\cal Y'}_2^2=16\epsilon^2({\cal Y}_1^2+{\cal Y}_2^2)=16\epsilon^2\kappa^2$ and assuming $\epsilon>0$,
we have $\kappa'=4\epsilon\kappa$. Then putting $\kappa=\frac{1}{4\epsilon} \kappa'$ and the identities~(\ref{D3Contra})
into  (\ref{eq:D4orbit}), we derive
$\big(\frac{e^{x'}+\epsilon^2e^{-x'}}{2\epsilon}\big)\big(\frac{{\cal J'}}{2}\big)^2-\cos(y')(\frac{\kappa'}{4\epsilon})=\frac{{\cal H'}}{8\epsilon}$.
Thus,
$e^{x'}{\cal J'}^2-2\cos(y')\kappa'={\cal H'}$
as $\epsilon\to 0$.

\subsubsection[Embedding $D3$ in 3D Minkowski space]{Embedding $\boldsymbol{D3}$ in 3D Minkowski space}

We can embed $D3$ as a surface in 3D Minkowski space with coordinates $X$, $Y$, $Z$ in such a way as to preserve rotational symmetry. For example, let
\begin{gather*}
 X=2\sqrt{2}e^{-\frac{x}{2}}\sqrt{1+e^{-x}}\cos \frac{y}{2},\qquad Y=2\sqrt{2}e^{-\frac{x}{2}}\sqrt{1+e^{-x}}\sin \frac{y}{2}, \\
 Z=\frac{\sqrt{6}}{12}\ln\left(\frac{\frac{\sqrt{3}}{6}(6+5e^x){\sqrt{3+2e^{2x}+5e^x}+1}}{\frac{\sqrt{3}}{6}(6+5e^x)\sqrt{3+2e^{2x}+5e^x}-1}
\right)-e^{-x}\sqrt{2}\sqrt{3+2e^{2x}+5e^x}.
\end{gather*}
Then
$dX^2+dY^2-dZ^2=\frac{2(e^x+1)}{e^{2x}}(dx^2+dy^2)$. Such embeddings are not unique.

{\bf Discussion of trajectories.} To see how the trajectory behaves on this surface, we impose the orbit equation $e^{x'}{\cal J'}^2-2\cos(y')\kappa'={\cal H'}$ to $X$, $Y$ and $Z$.
From the orbit equation, we have
\begin{gather}
e^x=\frac{{\cal H}+2\cos y \cdot\kappa}{{\cal J}^2}=\frac{{\cal H}}{{\cal J}^2}+\frac{2\kappa}{{\cal J}^2}\cos y=a+b\cos y,\label{eq:D3orbit2}
\end{gather}
where we ignore all the primes on the letters for convenience. Also, we treat $a=\frac{{\cal H}}{{\cal J}^2}$ and $b=\frac{2\kappa}{{\cal J}^2}$ as two new constants, so the number of constants is reduced by~1. Then we can express $X$, $Y$ and $Z$ all in terms of $y$.
We notice that since $e^x>0$, the trajectory will be closed when $a>b$ and unbounded when $a<b$.

\begin{Example} Plots of an``elliptical'' shape trajectory and its overhead 2-D view are given in Figs.~\ref{18.jpg} and~\ref{19.jpg}. (The size of the ``cone" that appears on the graphics is determined by the minimum value we choose for~$e^x$ in MATLAB. The smaller the minimum value,
the larger the ``cone''. For these two plots, the minimum value for~$e^x$ we choose is~$1$ and we adjust to dif\/ferent sizes
of ``cones'' for other examples to make the plots clearer.)
\end{Example}

\begin{figure}[t!]\centering
 \includegraphics[width=85mm]{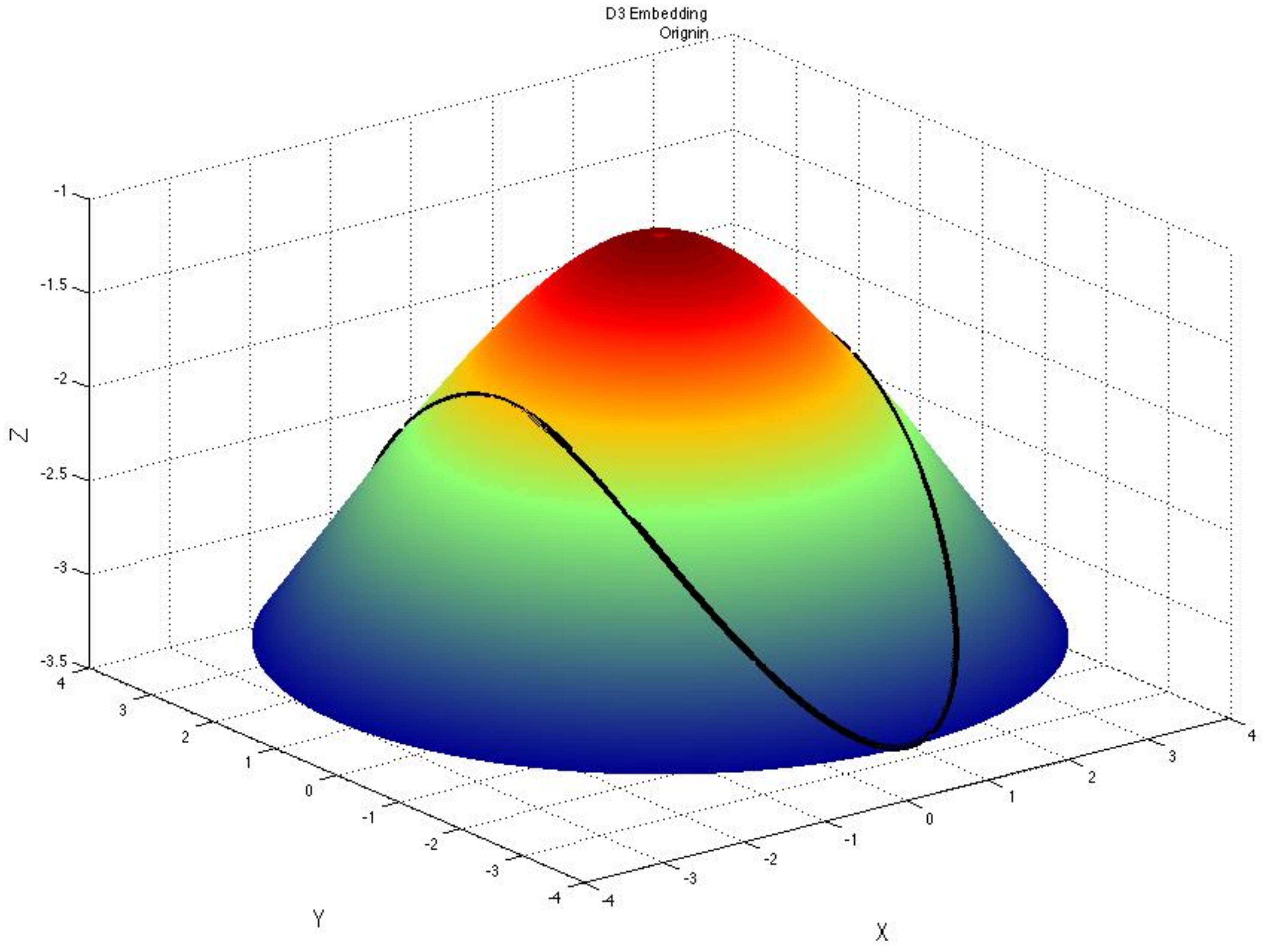}
\caption{3-D View of trajectory with $a=2$, $b=1$.}\label{18.jpg}
\end{figure}
\begin{figure}[t!]\centering
 \includegraphics[width=85mm]{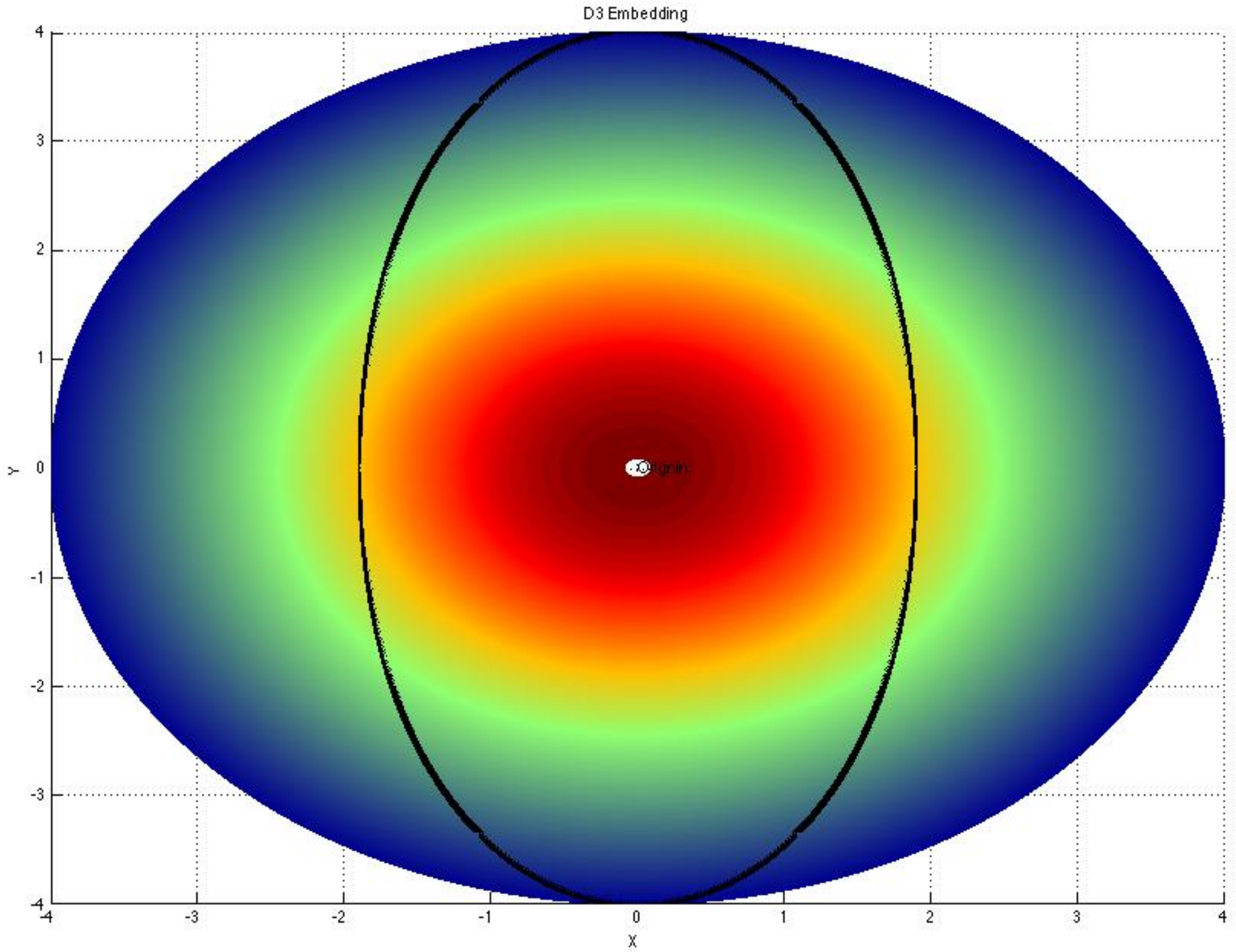}
\caption{Overhead view of trajectory with $a=2$, $b=1$.}\label{19.jpg}
\end{figure}

\begin{Example}
As $a$ becomes smaller, the ``ellipse'' is prolonged, see Figs.~\ref{20.jpg} and~\ref{21.jpg}. Seen from above the shape of the trajectory  two merging ellipses. (The minimum value for~$e^x$ we choose is~$0.1$ for these two plots.)
\end{Example}

\begin{figure}[t!]\centering
\includegraphics[width=85mm]{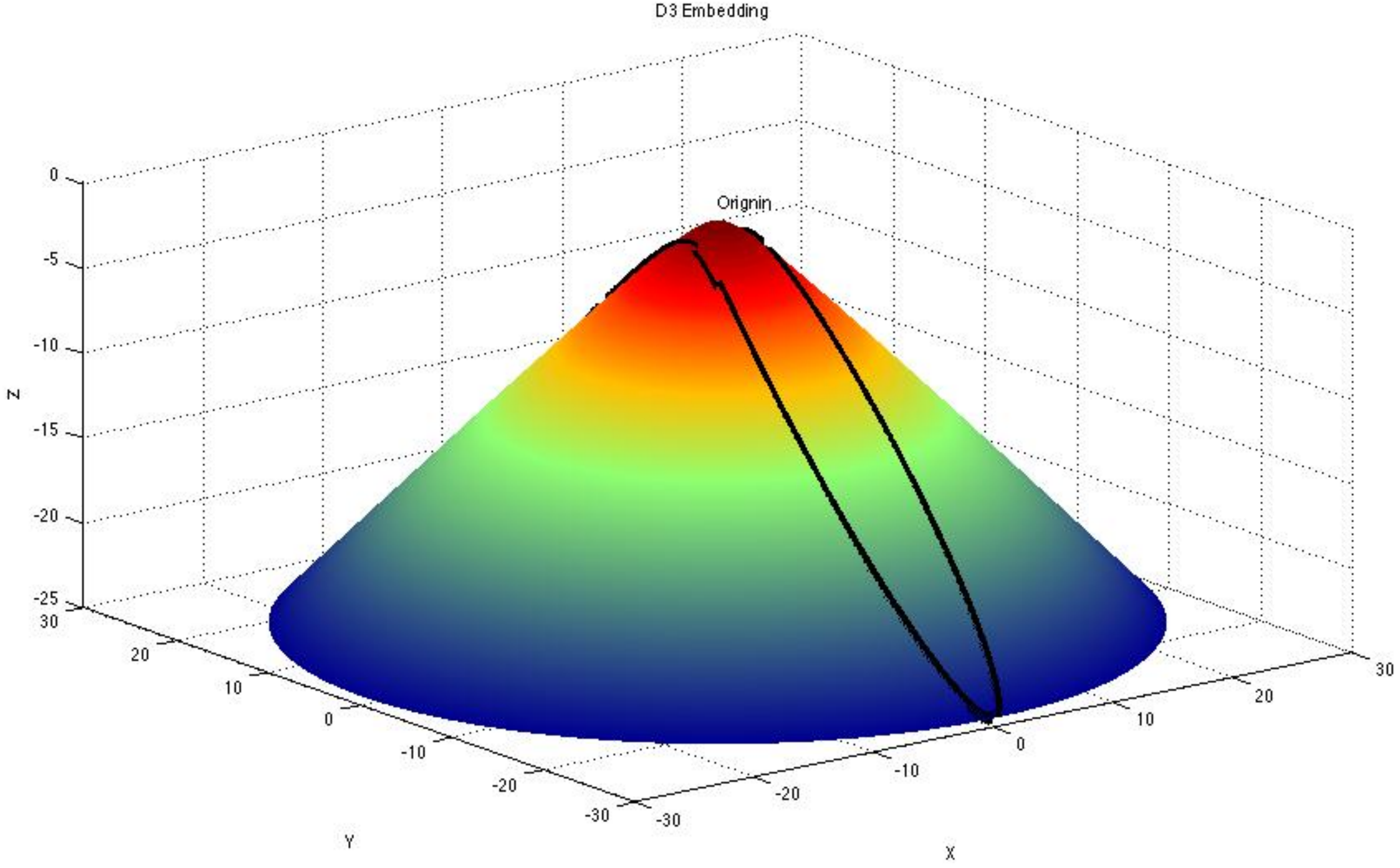}
\caption{3-D View of trajectory with $a=1.1$, $b=1$.}\label{20.jpg}
\end{figure}
\begin{figure}[t!]\centering
 \includegraphics[width=85mm]{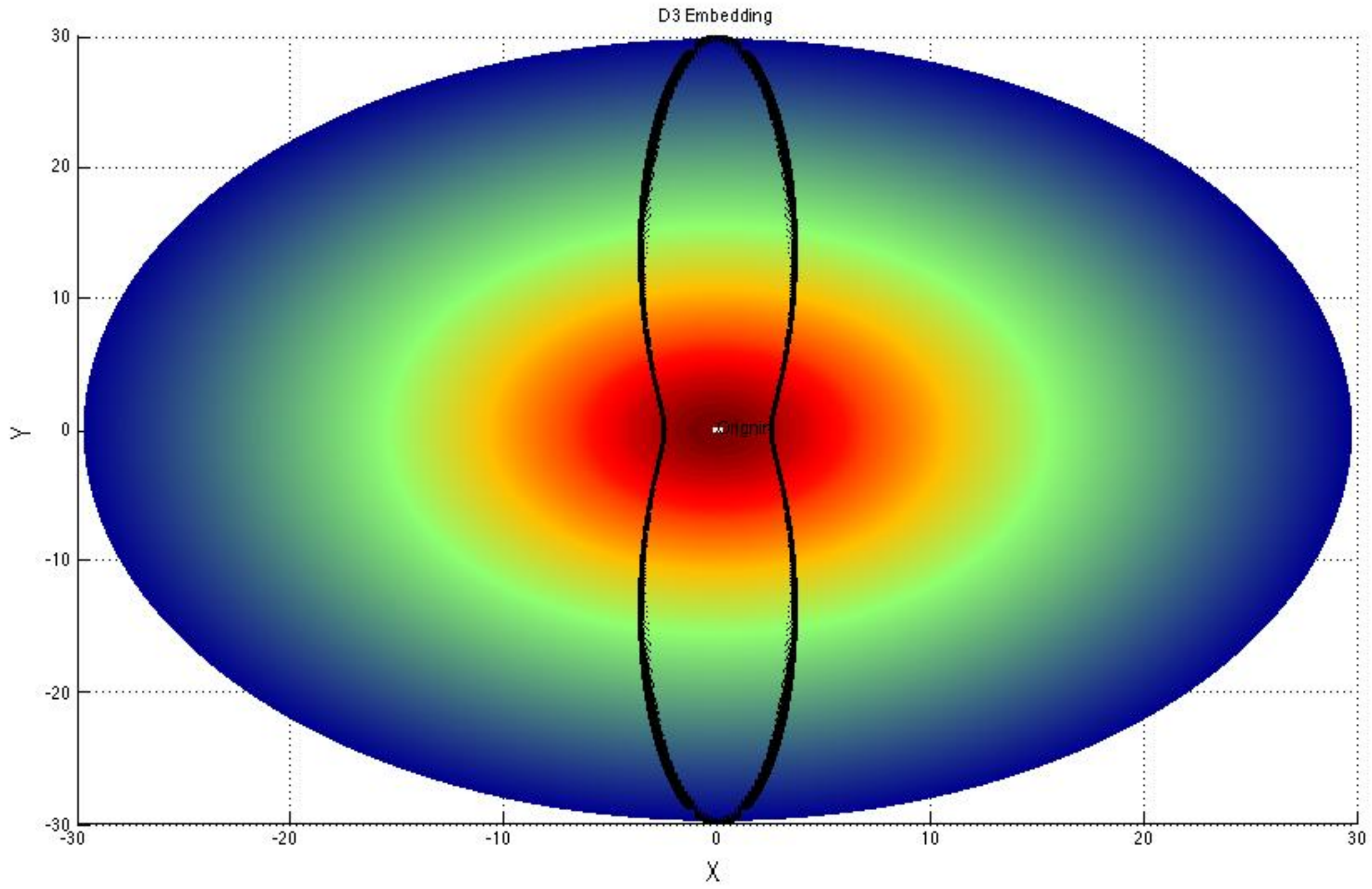}
\caption{Overhead view of trajectory with $a=1.1$, $b=1$.}\label{21.jpg}
\end{figure}

\begin{Example}[escape velocity] For $a=b$, we  encounter the boundary case between bounded trajectory and
unbounded trajectory. A plot of the case is given in Fig.~\ref{22.jpg}.
The shape of the trajectory when looking from above is like two very ``thin'' parabolas meeting at the top of the ``cone''.
(The minimum value for~$e^x$  we choose is $0.001$ for these two plots.)
 \end{Example}

 \begin{figure}[t!]\centering
 \includegraphics[width=85mm]{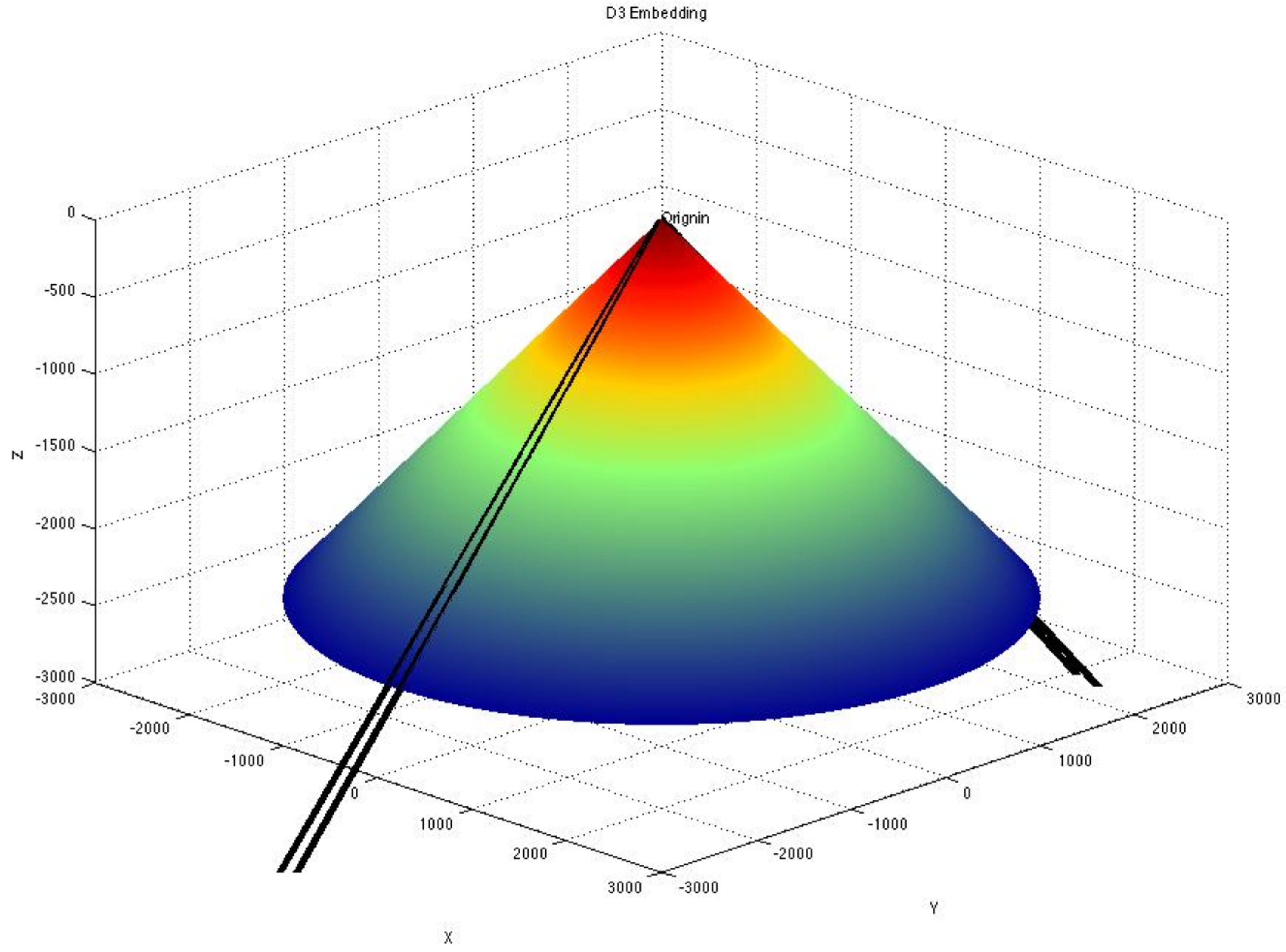}
\caption{3-D View of escape velocity trajectory with $a= b=1$.}\label{22.jpg}
\end{figure}

\begin{Example}[unbounded trajectory] \looseness=-1 When $a<b$, there are  no real $x$ values corresponding to a~range of values of $y$,  clearly seen
from equation (\ref{eq:D3orbit2}), so the trajectory is unbounded. Examples are given in Figs.~\ref{24.jpg} and~\ref{26.jpg}
which are a standard 3-D view and a closeup overhead view.
\end{Example}

\begin{figure}[t!]\centering
\includegraphics[width=85mm]{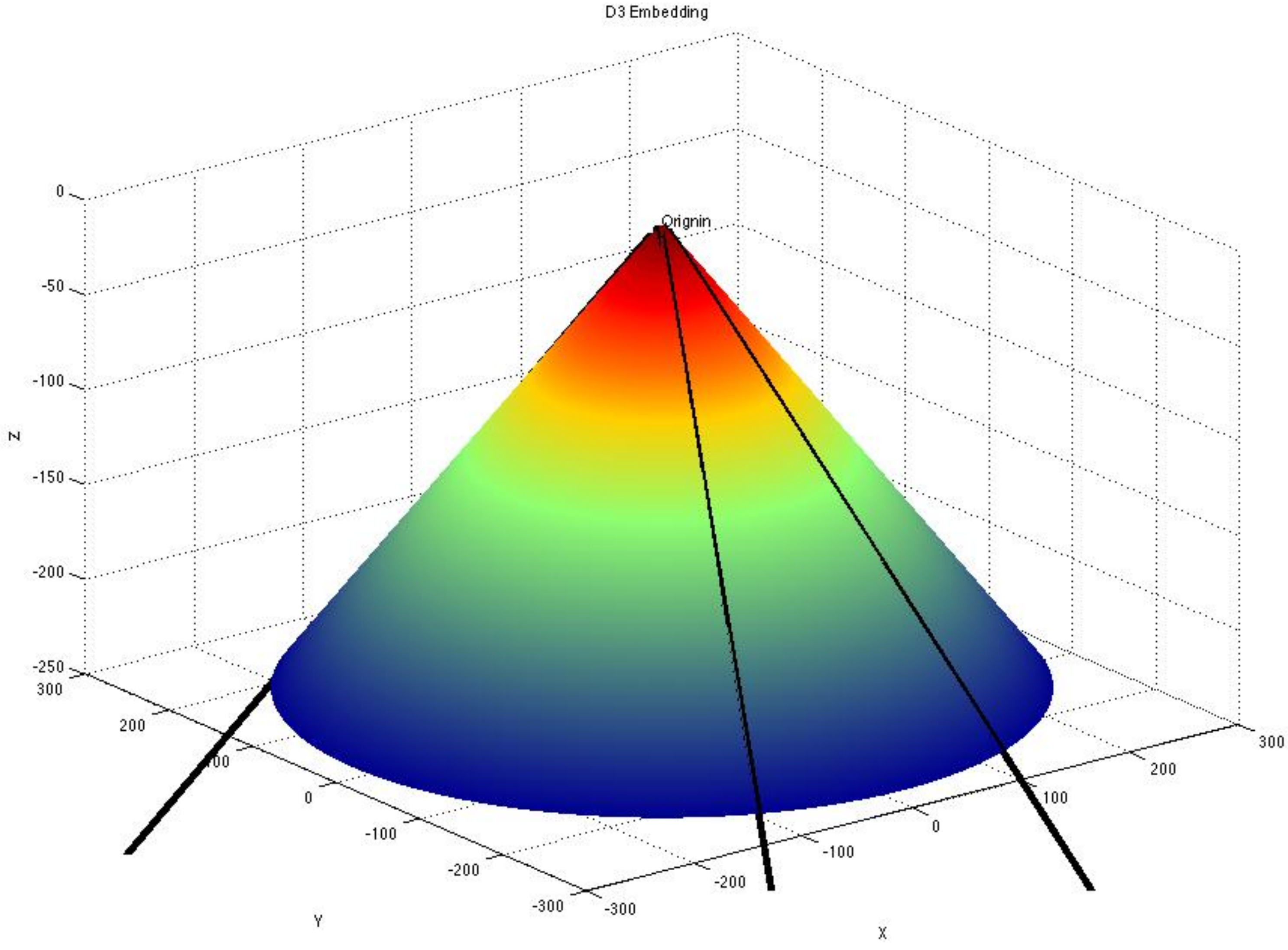}
\caption{3-D View of trajectory with $a=0.7$, $b=1$.}\label{24.jpg}
\end{figure}
\begin{figure}[t!]\centering
 \includegraphics[width=85mm]{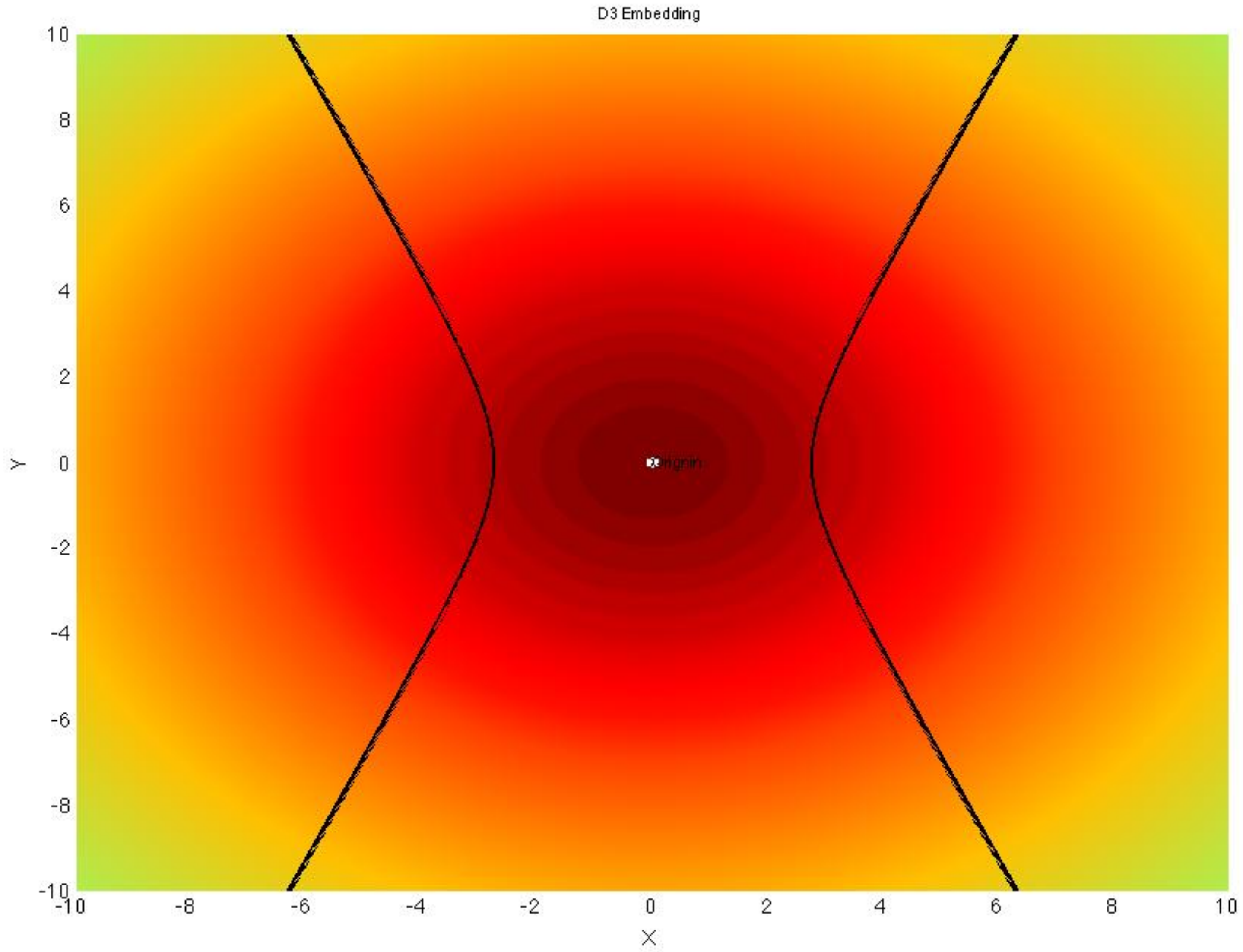}
\caption{Closeup overhead view of trajectory  with $a=0.7$, $b=1$.}\label{26.jpg}
\end{figure}

\subsubsection[Contraction of the $D3$ oscillator to the isotropic oscillator]{Contraction of the $\boldsymbol{D3}$ oscillator to the isotropic oscillator}

The isotropic oscillator in Euclidean space has the Hamiltonian
\begin{gather}\label{isotropicoscillator}
{\cal H}=p_x^2+p_y^2+\omega^2\big(x^2+y^2\big)
\end{gather}
and basis symmetries
\begin{gather*}
 {\cal K}=xp_y-yp_x,\qquad {\cal L}_1=\frac12\big(p_x^2-p_y^2\big)+\frac12\omega^2\big(x^2-y^2\big),\qquad {\cal L}_2=p_xp_y+\omega^2xy,\qquad {\cal K}^2.
 \end{gather*}
The structure equations are
\begin{gather*}
 \{{\cal L}_1,{\cal K}\}=2{\cal L}_2,\qquad \{{\cal L}_2,{\cal K}\}=-2{\cal L}_1,\qquad \{{\cal L}_1,{\cal L}_2\}=-2\omega^2{\cal K}.
 \end{gather*}
The functional relation is
${\cal L}_1^2+{\cal L}_2^2=\frac14{\cal H}^2-\omega^2{\cal K}^2$.
We can obtain this system as a limit of the D3 oscillator as follows: In~(\ref{Darboux3a}) we set
 $x=x'+\ln(\frac{1}{\epsilon})$, $y=y'$, $ {\cal H}'= 2\epsilon{\cal H}$, $\beta=(2\omega^2)/\epsilon^2$.
Then  as $\epsilon \to 0$ we have $K=p_{y'}$ and
\begin{gather*}
{\cal H}'=e^{x'}(p_{x'}^2+p_{y'}^2)+\frac{4\omega^2}{e^{x'}}, \\
X_1' = \epsilon X_1=\frac{e^{x'}}{2}\left(\sin y'  p_{x'}p_{y'}+\frac12\cos y'  \big(p_{x'}^2-p_{y'}^2\big)\right)+\frac{\omega^2\cos y'}{e^{x'}}, \\
X_2' =  \epsilon X_2= \frac{e^{x'}}{2}\left(-\cos  y' p_{x'}p_{y'}+\frac12\sin  y' \big(p_{x'}^2-p_{y'}^2\big)\right)+\frac{\omega^2\sin y'}{e^{x'}}.
\end{gather*}
In terms of f\/lat space Cartesian coordinates $X=r\cos\theta$, $Y=r\sin\theta$ we have
\begin{gather*}
   e^{x'}=\frac{4}{r^2},\qquad y'=2\theta,\qquad {\cal H}'=p_X^2+p_Y^2+\omega^2\big(X^2+Y^2\big), \\
   X_1'= \frac14\big(p_X^2-p_Y^2\big)+\frac14 \omega^2\big(X^2-Y^2\big),\qquad  X_2'= \frac12p_Xp_Y+\frac12\omega^2XY,
\end{gather*}
with ${\cal K}=\frac12(Xp_Y-Yp_X)$.
From the expressions of ${\cal H'}$ and $X'_1$, we get
\begin{gather}
p_x^2=\frac{{\cal H'}+4X'_1-2\omega^2X^2}{2},\qquad p_y^2=\frac{{\cal H'}-4X'_1-2\omega^2Y^2}{2}.\label{eq:pxpy}
\end{gather}
Then from the expression of $X'_2$, we have $p_xp_y=2X'_2-\omega^2XY$. Squaring this equation and equating it to the product of the two equations~\eqref{eq:pxpy}  we obtain the orbit equation,
\begin{gather*}
2\omega^2 \big(X^2+Y^2\big)({\cal H'}-4X'_1)-16\omega^2X'_2XY={\cal H'}-16X'^2_1-16X'^2_2.
\end{gather*}
The shape of the orbit depends only on  $8X'_2/({\cal H'}-4X'_1)$, so we are really investigating the equation $X^2+Y^2-aXY=1$
where $a=8X'_2/({\cal H'}-4X'_1)$. We  take the right side to be $1$ since we now only care about the shape, not the scale.

When $a=0$, i.e., $X'_2=0$, the trajectory is just a circle. When $0<a<2$, i.e., $0<8X'_2/({\cal H'}-4X'_1)<2$, the trajectory is a tilted ellipse
with $y=x$ as the major axis. In Fig.~\ref{33.jpg}, we show three trajectories for dif\/ferent values of  $a$, and also the major axis.
As $a$ increases, the ellipse becomes more elongated in the major axis direction.
\begin{figure}[t!]\centering
  \includegraphics[width=80mm]{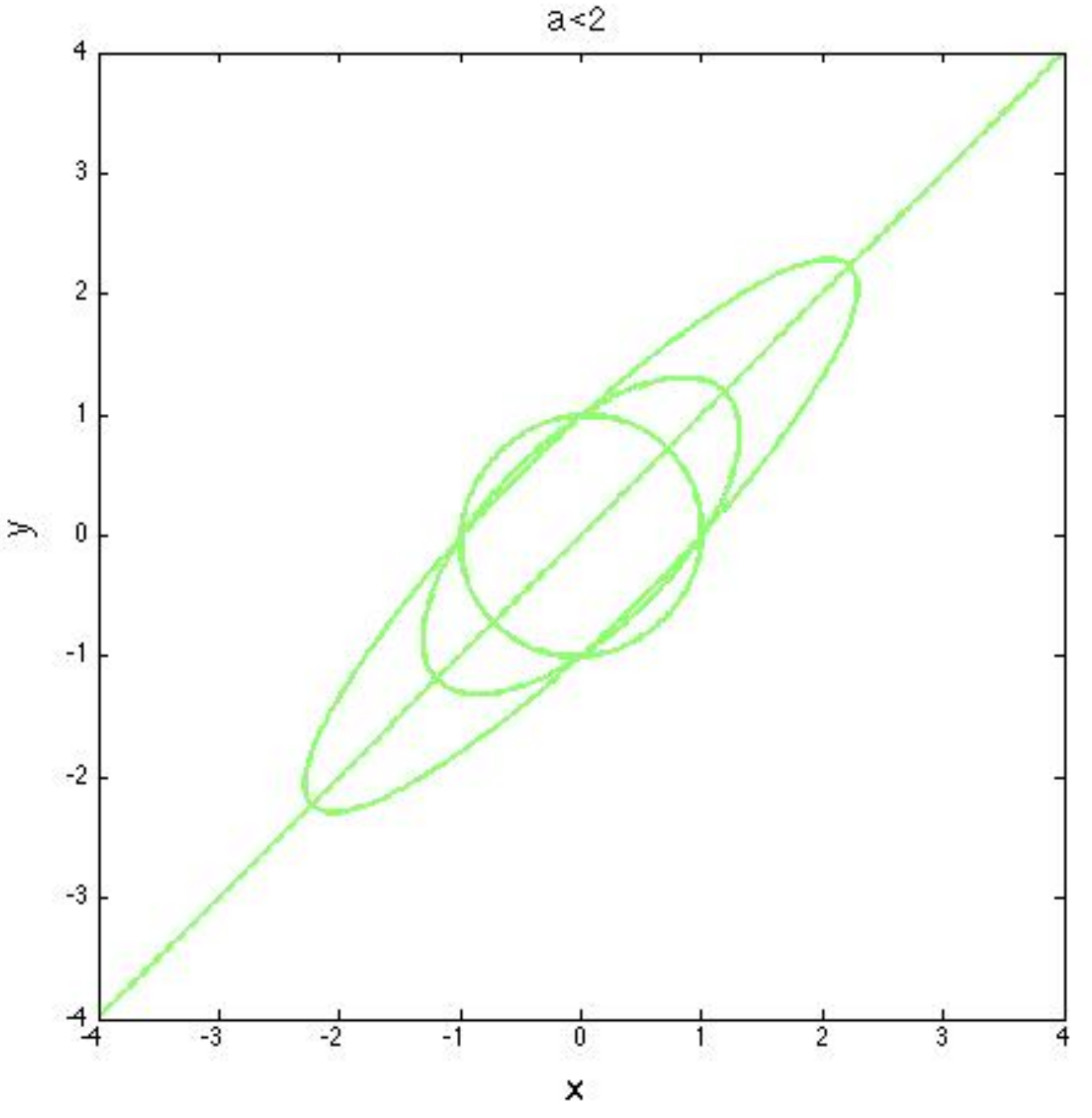}
\caption{Elliptical trajectories.}\label{33.jpg}
\end{figure}

When $a=2$, i.e., ${\cal H'}-4X'_1=4X'_2$, there is a bifurcation point on the momentum map. The trajectory splits into  two parallel straight
lines.
For $a>2$, i.e., $8X'_2/({\cal H'}-4X'_1)>2$, the trajectories are hyperbolas with $y=x$ as the symmetry axis. In Fig.~\ref{35.jpg},
we present three hyperbolic trajectories,  shown in darker colors as $a$ increases.
\begin{figure}[t!]\centering
\includegraphics[width=80mm]{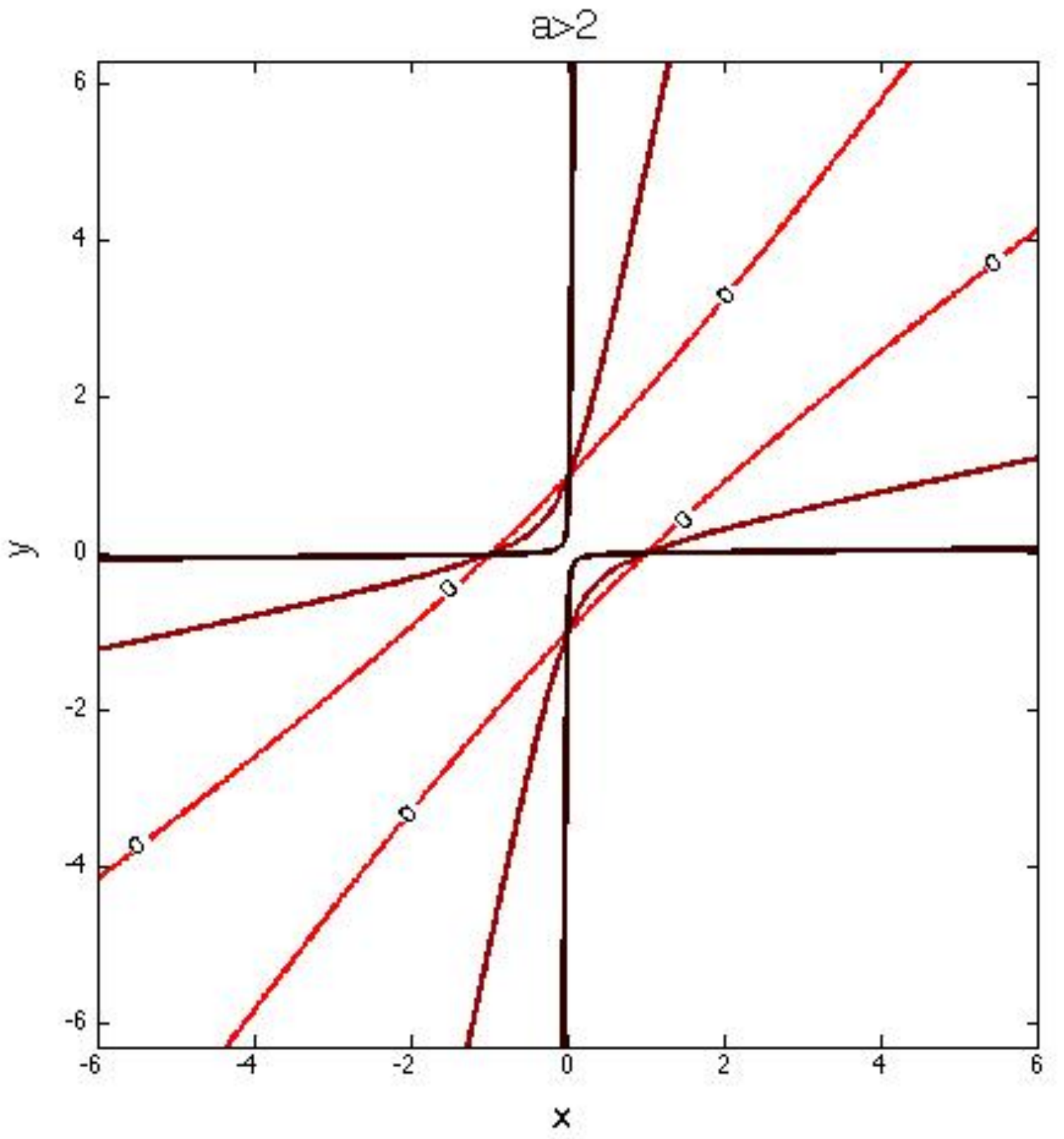}
\caption{Hyperbolic trajectories.}\label{35.jpg}
\end{figure}
It is also possible that $8X'_2/({\cal H'}-4X'_1)<0$, in which case the trajectories will be symmetric about the line $y=-x$.

\subsection[The Higgs oscillator $S3$]{The Higgs oscillator $\boldsymbol{S3}$}

 The classical  system $S3$ on the 2-sphere is determined
by  Hamiltonian \cite{Higgs}
\begin{gather}\label{Hamiltonian}
{\cal H}={\cal J}_1^2+{\cal J}_2^2+{\cal J}_3^2+\frac{\alpha\big(s_1^2+s_2^2+s_3^2\big)}{s_3^2},
\end{gather}
where
${\cal J}_1=s_2p_{3}-s_3p_{2}$ and ${\cal J}_2$, ${\cal J}_3$ are cyclic permutations of this expression. For computational convenience we have embedded the 2-sphere in Euclidean 3-space. Thus we can write
\begin{gather*}
{\cal H}'=p_1^2+p_2^2+p_3^2+\frac{\alpha}{s_3^2}=\frac{{\cal H} +(s_1p_1+s_2p_2+s_3p_3)^2}{s_1^2+s_2^2+s_3^2}
\end{gather*}
 and use the Euclidean space Poisson bracket
$\{{\cal F}, {\cal G}\}=\sum\limits_{i=1}^3(-\partial_{s_i}{\cal F}\partial_{p_i}{\cal G}+\partial_{p_i}{\cal F}\partial_{s_i}{\cal G})$
for our computations, but at the end we restrict to the unit  sphere:  $s_1^2+s_2^2+s_3^2=1$ and $s_1p_1+s_2p_2+s_3p_3=0$. The Hamilton equations for the trajectories $s_j(t)$, $p_j(t)$ in phase space are
\begin{gather*}
\frac{ds_j}{dt}=\{ {\cal H},s_j\},\qquad \frac{dp_j}{dt}=\{ {\cal H},p_j\},\qquad j=1,2,3.
\end{gather*}
 The classical basis for the constants of the motion is
\begin{gather}\label{constants}
{\cal L}_1={\cal J}_1^2+\alpha\frac{s_2^2}{s_3^2},\qquad  {\cal L}_2={\cal J}_1{\cal J}_2-\alpha\frac{s_1s_2}{s_3^2},\qquad
{\cal X}={\cal J}_3.
\end{gather}
The structure relations are
\begin{gather}\label{structure4}
\{{\cal X},{\cal L}_1\}=-2{\cal L}_2,\qquad  \{{\cal X},{\cal L}_2\}=2{\cal L}_1-{\cal H}+{\cal X}^2+\alpha,\qquad
\{{\cal L}_1,{\cal L}_2\}=-2({\cal L}_1+\alpha){\cal X},\!\!\!\!\!
\end{gather}
and the Casimir relation is
\begin{gather*}
{\cal L}_1^2 + {\cal L}_2^2- {\cal L}_1 {\cal H}+ {\cal L}_1 {\cal X}^2+ \alpha{\cal X}^2+ \alpha{\cal L}_1=0.
\end{gather*}
To analyze the classical trajectories it is convenient to replace the basis elements ${\cal L}_1$, ${\cal L}_2$ in the algebra with the new basis set
\begin{gather}\label{Ssymmetries}
{\cal S}_1=\frac12\left({\cal J}_1^2-{\cal J}_2^2-\frac{\alpha\big(s_1^2-s_2^2\big)}{s_3^2}\right),\qquad
{\cal S}_2=\frac12\left(2{\cal J}_1{\cal J}_2-\frac{2\alpha s_1s_2}{s_3^2}\right).
\end{gather}
Thus
${\cal S}_1={\cal L}_1+\frac12({\cal X}^2-{\cal H}+\alpha)$, ${\cal S}_2={\cal L}_2$.
Now the f\/irst two equations~(\ref{structure4}) become
$ \{{\cal J}_3,{\cal S}_1\}=-2{\cal S}_2,\quad  \{{\cal J}_3,{\cal S}_2\}=2{\cal S}_1$, so $({\cal S}_1,{\cal S}_2)$ transforms as a 2-vector with respect to rotations about the 3-axis. Indeed, a rotation through the angle $\beta $ about the 3-axis  rotates the vector by $2\beta$.
The remaining relations  now become
\begin{gather*}
\{{\cal S}_1,{\cal S}_2\}=-{\cal X}\big({\cal H}-{\cal X}^2+\alpha\big),\qquad {\cal S}_1^2+{\cal S}_2^2=\frac14\big({\cal X}^2-{\cal H}+\alpha\big)^2-\alpha {\cal X}^2\equiv \kappa^2.
\end{gather*}
This verif\/ies that the length $\kappa$ of the 2-vector is unchanged under a rotation and shows that $({\cal S}_1,{\cal S}_2)$  is similar to  the Laplace--Runge--Lenz vector for the Kepler problem in Euclidean space. (However, a better superintegrable
2-sphere analog of the Kepler problem is the potential $\alpha s_3/\sqrt{s_1^2+s_2^2}$, called $S6$ in our listing~\cite{KKMP}.)
In analogy with the choice of periaptic coordinates to simplify the Kepler problem, we choose the preferred coordinate system such that the vector  $({\cal S}_1,{\cal S}_2)$ points along the 1-axis, i.e., ${\cal S}_2=0$, ${\cal S}_1=\kappa\ge 0$. Then  equations~(\ref{Hamiltonian}) and~(\ref{constants}) become
\begin{gather}
 {\cal J}_1{\cal J}_2=\frac{\alpha s_1s_2}{s_3^2},\qquad {\cal J}_1^2-{\cal J}_2^2=2\kappa+\frac{\alpha\big(s_1^2-s_2^2\big)}{s_3^2},\nonumber\\
 {\cal J}_3^2={\cal X}^2,\qquad
{\cal J}_1^2+{\cal J}_2^2+{\cal J}_3^2={\cal H}-\frac{\alpha}{s_3^2}.\label{Jrelations1}
\end{gather}
The last 3 equations can be solved to give
\begin{gather*}
{\cal J}_1^2 =  \kappa+\frac12\big({\cal H}-{\cal X}^2-\alpha\big)-\frac{\alpha s_2^2}{s_3^2}, \nonumber\\
 {\cal J}_2^2= -\kappa+\frac12\big({\cal H}-{\cal X}^2-\alpha\big)-\frac{\alpha s_1^2}{s_3^2},\qquad
 {\cal J}_3^2=  {\cal X}^2. 
  \end{gather*}
Substituting these results into the square of the f\/irst equation~(\ref{Jrelations1}) and simplifying, we get the result (assuming $\alpha\ne 0$)
\begin{gather}\label{cone}
\left[\kappa+\frac{{\cal H}-{\cal X}^2-\alpha}{2}\right]s_1^2+\left[-\kappa+\frac{{\cal H}-{\cal X}^2-\alpha}{2}\right]s_2^2-{\cal X}^2s_3^2=0.
\end{gather}
This is the equation of a cone $As_1^2+Bs_2^2+Cs_3^2=0$. The trajectories lie on the intersection of this cone and the unit sphere $s_1^2+s_2^2+s_3^2=1$. Thus we  get conic sections again, as in the Euclidean Kepler problem,  but this time the sections are intersections with the unit sphere, rather than planes. The possible types of trajectory will depend on the signs of~$A$,~$B$,~$C$.

The projection on the $s_1-s_2$ plane is the curve
\begin{gather}\label{S3projection}
\left({\frac{{\cal H}+{\cal X}^2-\alpha}{2}+\kappa}\right)s_1^2
+\left(\frac{{\cal H}+{\cal X}^2-\alpha}{2}-\kappa\right)s_2^2={\cal X}^2.
\end{gather}

The identities
\begin{gather}
\label{kappaident1} \left[\kappa+\frac12({\cal H}-{\cal X}^2-\alpha)\right]\left[-\kappa+\frac12({\cal H}-{\cal X}^2-\alpha)\right]=\alpha{\cal X}^2,\\
\label{kappaident2} \left[\kappa+\frac12({\cal H}+{\cal X}^2-\alpha)\right]\left[-\kappa+\frac12({\cal H}+{\cal X}^2-\alpha)\right]={\cal H}{\cal X}^2,\\
\label{kappaident3} \left[\kappa+\frac12({\cal H}-{\cal X}^2+\alpha)\right]\left[-\kappa+\frac12({\cal H}-{\cal X}^2+\alpha)\right]=\alpha{\cal H}
\end{gather}
will be important in the analysis of trajectories to follow.

{\bf Analog of Kepler's second law of planetary motion.}  We see from equations~(\ref{S3projection}) and~(\ref{kappaident2}) that for the nonzero angular momentum $\cal X$ the projection of the motion on the unit circle in the $s_1-s_2$ plane is a segment of an ellipse, hyperbola or straight line and that none of these curves pass through the center $(s_1,s_2)=(0,0)$ of the circle. Thus as the particle moves along its trajectory $(s_1(t),s_2(t),s_3(t))$, the line segment connecting the projection $(s_1(t),s_2(t))$ to the center of the circle sweeps out an area. Introducing polar coordinates $s_1(t)=r(\phi(t))\cos\phi(t)$, $s_2(t)=r(\phi(t))\sin\phi(t)$ we see that in the  interval from some initial time~$0$ to time~$t$ the area swept out is $A(t) =\frac12\int_{\phi(0)}^{\phi(t)} r^2(\phi)d\phi$. Thus the rate at which the area is swept out is
$ \frac{dA}{dt}=\frac{dA}{d\phi}\frac{d\phi}{dt}= \frac12 r^2(\phi(t)) \frac{d\phi}{dt}$.
Now note from Hamilton's equations that along the trajectory
\begin{gather*}
{\cal X}={\cal J}_3=s_1(t)p_2(t)-s_2(t)p_1(t)= s_1\frac{ds_2}{dt}-s_2\frac{ds_1}{dt}=r^2\frac{d\phi}{dt}.
\end{gather*}
Thus $\frac{dA}{dt}=\frac{\cal X}{2}$, a constant.
\begin{Theorem} If the angular momentum is nonzero, the projection of the trajectory on the unit circle in the $s_1-s_2$ plane sweeps out equal areas in equal times.
 \end{Theorem}

Now we begin an analysis of the trajectories. There are two cases, depending on whether the potential is repulsive $(\alpha>0)$ or attractive $(\alpha<0)$.

 {\bf Case 1: $\alpha>0$.}
This is the case of a repulsive potential. The equator of the sphere repels the particle.
Thus motion is conf\/ined to a hemisphere. We start an analysis of the types of trajectories.

If ${\cal X}\ne 0$ and using the fact that all trajectories will be periodic for a repulsive force, we see that  there will necessarily be at least one point on the phase space trajectory for which $p_2=0$. Substituting into the phase space conditions~(\ref{Jrelations1}) it is easy to see that this is possible only for $s_1=0$. Solving all the equations completely we f\/ind that these points on the phase space trajectory are uniquely determined by the constants of the motion as follows:
\begin{gather}
s_1=0,\qquad s_2^2=\frac{\kappa+\frac12\big({\cal H}+{\cal X}^2-\alpha\big)}{{\cal H}},\qquad p_1^2=-\kappa+\frac12\big({\cal H}+{\cal X}^2-\alpha\big),\nonumber\\
 p_2=0,\qquad p_3=0.\label{orbitdef1}
\end{gather}
This prescription gives us a point on each trajectory, comparable to the aphelion for the Kepler system, where we can start to trace out the orbit. It remains to analyze the possible orbits. To see what is the available parameter space for the case ${\cal X}\ne 0$, note that f\/irst we must require $\kappa\ge 0$. Noting the identity~(\ref{kappaident1})
we see that the quantities in brackets must have the same sign. However, if that sign is negative then the cone~(\ref{cone})
will degenerate to a point and not intersect the sphere. Thus to obtain trajectories it is necessary that the constants of the motion satisfy
$-\kappa+\frac12({\cal H}-{\cal X}^2-\alpha)>0$, $ \kappa\ge 0$.
On the other hand, if these conditions are satisf\/ied, we see from equations~(\ref{orbitdef1}) that there exist trajectories for which the corresponding constants of the motion are assumed. Thus the conditions are necessary and suf\/f\/icient.

{\bf An Analog of Kepler`s 3rd law of planetary motion.}   For nonzero angular momentum~$\cal X$ the projection of the motion on the unit circle in the $s_1-s_2$ plane is  an ellipse~(\ref{S3projection}) enclosing the center of the circle. The area of this ellipse is easily seen to be $\pi {\cal X}/{\sqrt{\cal H}}$. Let $T$ be the period of the trajectory. Thus $T$ is the length of time for the projection of the trajectory to trace out the complete ellipse and $A(T)=\pi {\cal X}/{\sqrt{{\cal H}}}$. Since the area of the ellipse is swept out at the constant rate $\frac{dA}{dt}=\frac{\cal X}{2}$  we have
$ A(T)=\frac{dA}{dt}T=\frac{{\cal X}T}{2}$.
Equating these two expressions for $A(t)$ and solving for~$T$ we f\/ind
$T=2\pi/\sqrt{{\cal H}}$.
\begin{Theorem}\label{th:period} For ${\cal X}\ne 0$ the period of an orbit is $T=2\pi/\sqrt{{\cal H}}$.
\end{Theorem}

If ${\cal X}=0$ then from (\ref{cone}) and~(\ref{S3projection}) we see that the projection of the motion in the $s_1-s_2$ plane is the straight line segment $s_1=0$. Thus this motion takes place in a plane and the trajectory is a portion of a great circle passing through the poles of the sphere.  Here, $\kappa=\frac12({\cal H}-\alpha)$, ${\cal H}\ge \alpha>0$. The motion is periodic and the points of closest approach to the equator are those such that $s_2^2=({\cal H}-\alpha)/{\cal H}$. The period is again $T=2\pi/\sqrt{{\cal H}}$. There is a special case of equilibrium at a pole when ${\cal H}=\alpha$.

 {\bf Case 2: $\alpha<0$.}
This is the case of an attractive  potential, where the equator of the sphere attracts  the particle. Again,  motion is conf\/ined to a hemisphere. We f\/irst consider the case where ${\cal X}\ne 0$.  Without loss of generality we can assume ${\cal X}>0$. Then the projection of the motion is  counter-clockwise.  Now $\kappa>0$ and from the identity~(\ref{kappaident2})
we see that there will be 3~classes of trajectories, depending on the value of~$\cal H$.

If ${\cal H}>0$ then the coef\/f\/icients of $s_1^2$ and $s_2^2$ in the projection formula (\ref{S3projection}) must have the same sign, necessarily positive for a real trajectory. Thus
\begin{gather*}
{\cal H}>0\colon \quad  {\cal H}+{\cal X}^2-\alpha>2\kappa>{\cal H}-{\cal X}^2-\alpha,\qquad 2\kappa+{\cal H}-{\cal X}^2-\alpha>0,
\end{gather*}
and the projection will be a segment of an ellipse. It is straight-forward to check that the ellipse always intersects the circle at 4 points.
All of these trajectories lead to annihilation at the equator. There will necessarily be a ``perihelion'' point on each trajectory:
$s_2=0$, $p_1=p_3=0$, $s_1^2={\cal X}^2/\left[\kappa +\frac12({\cal H}+{\cal X}^2-\alpha)\right]$, $p_2^2=\kappa+\frac12({\cal H}+{\cal X}^2-\alpha)$.

If ${\cal H}<0$ then the coef\/f\/icient of $s_1^2$ is positive and the coef\/f\/icient of $s_2^2$ is negative in the projection formula~(\ref{S3projection}). Thus
\begin{gather*}
{\cal H}<0\colon \quad  {\cal H}+{\cal X}^2-\alpha-2\kappa<0,\qquad {\cal H}+{\cal X}^2-\alpha+2\kappa>0,
\end{gather*}
and the projections will be segments of hyperbolas. Each branch of the  hyperbola intersects the circle at exactly 2~points. Again,
these trajectories lead to annihilation at the equator. Again there is a ``perihelion'' point on each trajectory:
$s_2=0$, $p_1=p_3=0$, $s_1^2={\cal X}^2/\big[\kappa +\frac12({\cal H}+{\cal X}^2-\alpha)\big]$, $p_2^2=\kappa+\frac12({\cal H}+{\cal X}^2-\alpha)$.

If ${\cal H}=0$ then $\kappa=\frac12({\cal X}^2-\alpha)$ so the coef\/f\/icient of $s_1^2$ is positive and the coef\/f\/icient of $s_2^2$ is 0 in the projection formula~(\ref{S3projection}). Thus
\begin{gather*}
{\cal H}=0\colon \quad  {\cal X}^2-\alpha=2\kappa
\end{gather*}
and the projections will be segments of straight lines. $s_1^2={\cal X}^2/({\cal X}^2-\alpha)$. Here $p_1=0$ and $p_2^2={\cal X}^2-\alpha$,
a constant in accordance with the analog of Kepler's Second Law.  These trajectories lead to annihilation at the equator. There is a ``perihelion'' point
on each trajectory:
$s_2=0$, $p_1=p_3=0$, $ s_1^2={\cal X}^2/({\cal X}^2-\alpha)$, $p_2^2={\cal X}^2-\alpha$.

Finally, we suppose ${\cal X}=0$, so that the motion takes place in a plane through the poles. If ${\cal H}-\alpha>0$ then the motion takes place in the plane $s_1=0$. All trajectories annihilate at the equator. The projected trajectories each pass through the center of the circle, and we have $p_2^2={\cal H}-\alpha-{\cal H}s_2^2$. If ${\cal H}-\alpha<0$ then the motion takes place in the plane~$s_2=0$. All trajectories annihilate at the equator. The projected trajectories do not pass through the center of the circle, and we have $p_1^2={\cal H}-\alpha-{\cal H}s_1^2$. If ${\cal H}=\alpha$ then there are possible trajectories on any plane passing through the poles. If we go to new rotated~$s_1$,~$s_2$ coordinates such that the motion takes place in the plane $s_1=0$, then $p_1=0$ and $p_2^2=-\alpha s_2^2$. All trajectories annihilate at the equator, except for unstable equilibria at the poles.

\subsubsection[Contraction of $D4(b)$ to the Higgs oscillator on the sphere]{Contraction of $\boldsymbol{D4(b)}$ to the Higgs oscillator on the sphere}

By letting $\epsilon\to 0$ in $b=-2+\epsilon^2$ and def\/ining new constant $\alpha=4\beta$, the Hamiltonian equa\-tion~(\ref{D4bham}) becomes
\begin{gather*}
 {\cal H} =\frac{\sinh^2(2x)}{2\cosh(2x)-2}\big(p_x^2+p_y^2\big)+\frac{4\beta}{2\cosh(2x)-2}
=\frac{(2\sinh x\cosh x)^2}{(2\sinh x)^2}\big(p_x^2+p_y^2\big)+\frac{4\beta}{(2\sinh x)^2} \\
\hphantom{{\cal H}}{} =\cosh^2x\big(p_x^2+p_y^2\big)+\frac{\beta(\cosh^2 x-\sinh^2 x)}{\sinh^2x}
=\cosh^2x\big(p_x^2+p_y^2\big)+ \frac{\beta\cosh^2 x}{\sinh^2x}-\beta.
\end{gather*}
We can ignore the constant term $-\beta$ and write{\samepage
\begin{gather*}
{\cal H}'=\cosh^2x\big(p_x^2+p_y^2\big)+ \frac{\beta\cosh^2 x}{\sinh^2x},
\end{gather*}
such that ${\cal H}'={\cal H}+\beta$.  This is the Higgs oscillator on the 2-sphere.}

In terms of the Euclidean embedding of the sphere we have
$ s_1\!=\!{\cos y}/{\cosh x}$, $s_2\!=\!{\sin y}/{\cosh x}$, $s_3={\sinh x}/{\cosh x}$,
so $s_1^2+s_2^2+s_3^2=1$.
If we set $s_1=r'\cos \theta'$,  $s_2=r'\sin \theta'$ as in polar coordinates, then $r'=1/{\cosh x}$ and $\theta'=y$. However,
we notice that in the analog of Kepler's 2nd law part of the Section~\ref{sec:D4},
$r={\sqrt{2\cosh(2x)+b}}/({2\sinh{2x}})$ and $\theta=2y$. As $\epsilon\to 0$, $r=1/({2\cosh x})$. Thus $r'=2r$ and $\theta'=\theta/2$,
so this contraction space is a double covering of the original~$D4(b)$ space. We should be careful of this fact in the further computations.
Expressed in the phase space $(s_1, s_2, s_3, p_1, p_2, p_3)$, we have
${\cal H'}={\cal J}_1^2+{\cal J}_2^2+{\cal J}_3^2+\frac{\beta}{s_3^2}$,
where ${\cal J}_1=s_2p_{3}-s_3p_{2}$ and~${\cal J}_2$,~${\cal J}_3$ are cyclic permutations of this expression.
Considering the potential part of the Hamiltonian, we have a new basis set,
${\cal Y'}_1={\cal J}_1^2-{\cal J}_2^2-\beta(s_1^2-s_2^2)/s_3^2$,
${\cal Y'}_2=2{\cal J}_1{\cal J}_2-(2\beta s_1s_2/s_3^2)$, ${\cal J'}={\cal J}_3$.

 We will show that in this limit, the orbit equation (\ref{eq:D4orbit}) will yield the orbit equation~(\ref{S3projection}).   First, we transform $\cosh(2x)$ and $\cos(2y)$ in equation~(\ref{eq:D4orbit}) to express them in
terms of~$s_1$ and~$s_2$. ($s_3$ can be expressed in terms of~$s_1$ and~$s_2$.)
From the fact that $\cosh^2 x=1/(s^2_1+s^2_2)$, we have
\begin{gather*}
\cosh(2x)=2\cosh^2 x-1=\frac{2}{s^2_1+s^2_2}-1,\qquad
\cos(2y)=\cos^2 y-\sin^2 y=\frac{s_1^2-s_2^2}{s^2_1+s^2_2}.
\end{gather*}  Then equation (\ref{eq:D4orbit}) becomes
$ [2/(s^2_1+s^2_2)-1 ]{\cal J'}^2- [(s_1^2-s_2^2)/(s^2_1+s^2_2) ]\kappa'={\cal H},$
where $\kappa'$ is the length of the two vector $({\cal Y'}_1, {\cal Y'}_2)$ and we set ${\cal Y'}_2=0$ such that $\kappa'={\cal Y'}_1$.
Multiplying the equation by $(s^2_1+s^2_2)$, plugging in the fact that ${\cal H}={\cal H'}-\beta$ and doing some rearrangements, we obtain
\begin{gather}\label{D4-HiggsOrbit}
\left(\frac{{\cal H'}+{\cal J'}^2-\beta+\kappa'}{2}\right)s_1^2+\left(\frac{{\cal H'}+{\cal J'}^2-\beta-
\kappa'}{2}\right)s_2^2={\cal J'}^2,
\end{gather}
where we realize that $\kappa'={\cal Y'}_1$ which is 2 times the basis ${\cal S}_1$ in equation (\ref{Ssymmetries}),
so $\kappa'=2\kappa$ where~$\kappa$ is the one in equation~(\ref{S3projection}). Therefore, we can see that~(\ref{D4-HiggsOrbit}) is
the same orbit equation as~(\ref{S3projection}).

{\bf Analog of Kepler's second law of planetary motion.} Following the same procedure as in Section \ref{sec:D4}, and noticing that
$r'=2r$ and $\theta'=\theta/2$, we have,
$\frac{dA'}{dt}= \frac12 \left(4r^2(\theta(t))\right) \frac{d\theta}{2dt}=r^2\frac{d\theta}{dt}=2\frac{dA}{dt}$.
Therefore,
$\frac{dA'}{dt}= {\cal J'},$
which matches the expression  for the  Higgs oscillator.

{\bf Period of the orbit.} With the same procedure as in Section~\ref{sec:D4} and  $r'=2r$ and $\theta'=\theta/2$, we f\/ind  area swept out as a trajectory goes through one period $T'$ is
$A'=\frac12\int_{0}^{2\pi} r'^2(\theta')d\theta'=2\int_{0}^{2\pi} r^2(\theta)d\theta=4A$
where $A$ is the area in equation (\ref{eq:Area}).
Also, from the second law we see that $A'=\frac12 {\cal J'}T'$. Hence,
\begin{gather*}
T'=\frac{A'}{{\cal J'}}=\frac{4A}{{\cal J'}}=2T=\pi\left(\frac{2+b}{\sqrt{(-2{\cal H}-b{\cal H}+\alpha)}}+\frac{2-b}{\sqrt{(2{\cal H}-b{\cal H}+\alpha)}}\right),
\end{gather*}
where $T$ is the same as  equation~(\ref{period2}).

In the limit as $\epsilon\to 0$, we take $b=-2$ into the above equation of $T'$ and get ,
\begin{gather*}
{\cal T'}={\pi}\frac{4}{\sqrt{4{\cal H}+\alpha}}=\frac{2\pi}{\sqrt{4({\cal H'}-\beta)+4\beta}}=\frac{\pi}{\sqrt{{\cal H'}}},
\end{gather*}
which dif\/fers from the period in Theorem~\ref{th:period} by a factor of $\frac12$, due to the fact that the contraction is a double covering of the original $D4$ space.

\subsubsection{Contraction of the Higgs to the isotropic oscillator}

This contraction has a simple geometric interpretation, the contraction of a 2-sphere to a plane. We can consider the Higgs oscillator as living
in a 2-dimensional bounded ``universe'' of radius $1$ in some set of units. Suppose an observer is situated in this universe ``near''
the attractive north pole. Our observer uses a system of units with unit length $\epsilon$ where $0<\epsilon\ll 1$ and we suppose
that using these units  the universe appears f\/lat to the observer. Thus in the observer's units we have $s_1=\epsilon X$,
$s_2=\epsilon Y$, $s_3=\sqrt{1-\epsilon^2(X^2+Y^2)}=1-(X^2+Y^2)\epsilon^2/2+O(\epsilon^4)$. Here, $\epsilon^2$ is so small that to
 the observer it appears that $s_3=1$. Thus, to the observer, it appears that the universe is the plane $s_3=1$ with local Cartesian coordinates
 $(X,Y)$. We compare the actual system on the 2-sphere with the system as it appears to the observer and we assume that $\epsilon^2$ is so small
 that it can be neglected, unless we are dividing by it. Now we have
\begin{gather*}
 s_1=\epsilon X,\qquad s_2=\epsilon Y,\qquad s_3\approx 1,\qquad p_1=\frac{p_X}{\epsilon},\qquad p_2=\frac{p_Y}{\epsilon},\qquad p_3\approx -(Xp_X+Yp_Y).
 \end{gather*}
We def\/ine new constants $\omega^2$, $h$ by $\alpha=\omega^2/\epsilon^4$, ${E}-{\cal X}^2-\omega^2/\epsilon^4=h/\epsilon^2$, where ${\cal J}_3={\cal X}$.
Then substituting these results into the Higgs Hamiltonian equation ${\cal H}\equiv {\cal J}_1^2+{\cal J}_2^2+{\cal J}_3^2+\alpha/s_3^2=E$, we f\/ind
$ (p_X^2+p_Y^2)/\epsilon^2+\omega^2/\big(\epsilon^2\sqrt{X^2+Y^2}\big)=h/\epsilon^2$.
Multiplying both sides of this equation by~$\epsilon^2$  we obtain the Hamiltonian equation for the Euclidean space isotropic oscillator
\begin{gather*}
{\tilde{\cal H}}\equiv p_X^2+p_Y^2+{\omega^2}\big(X^2+Y^2\big)= {h},
\end{gather*}
in agreement with (\ref{isotropicoscillator}).
Using the same procedure we f\/ind that the constants of the motion become ${\cal K}={\cal X}=xp_y-yp_x$, and
\begin{gather*}
{\tilde{\cal L}}_1=\lim_{\epsilon\to 0}\epsilon^2 {\cal L}_1=p_Y^2+\omega^2Y^2,
\qquad  {\tilde{\cal L}}_2=\lim_{\epsilon\to 0}\epsilon^2 {\cal L}_2=-p_Xp_Y-\omega^2XY,
\end{gather*}
an alternate basis for the symmetries of the isotropic oscillator.
Similarly, the structure equations for the Higgs oscillator go in the limit to the structure equations of the isotropic oscillator.

\section{Examples  of 2D 2nd degree nondegenerate systems}\label{section3}
\subsection[The system $S7$ on the sphere]{The system $\boldsymbol{S7}$ on the sphere}

 The classical  system $S7$ on the 2-sphere \cite{KKMP}  is determined
by the Hamiltonian
\begin{gather*}
{\cal H}={\cal J}_1^2+{\cal J}_2^2+{\cal J}_3^2+\frac{a_1 s_1}{s_2^2\sqrt{s_1^2+s_2^2}}+\frac{a_2}{s_2^2} +\frac{a_3s_3}
{\sqrt{s_1^2+s_2^2}},
\end{gather*}
where
${\cal J}_1=s_2p_{3}-s_3p_{2}$ and ${\cal J}_2$, ${\cal J}_3$ are cyclic permutations of this expression.
We have embedded the 2-sphere in Euclidean 3-space, so
\begin{gather*}
p_1^2+p_2^2+p_3^2={\cal J}_1^2+{\cal J}_2^2+{\cal J}_3^2 +\frac{(s_1p_1+s_2p_2+s_3p_3)^2}{s_1^2+s_2^2+s_3^2}
\end{gather*}
 and we can use the Poisson bracket
$\{{\cal F}, {\cal G}\}=\sum\limits_{i=1}^3(-\partial_{s_i}{\cal F}\partial_{p_i}{\cal G}+\partial_{p_i}{\cal F}\partial_{s_i}{\cal G})$
for our computations, but at the end we restrict to the unit  sphere:  $s_1^2+s_2^2+s_3^2=1$ and $s_1p_1+s_2p_2+s_3p_3=0$. The Hamilton equations for the trajectories $s_j(t)$, $p_j(t)$ in phase space are
\begin{gather*}
\frac{ds_j}{dt}=\{ {\cal H},s_j\},\qquad \frac{dp_j}{dt}=\{ {\cal H},p_j\},\qquad j=1,2,3.
\end{gather*}
 The classical basis for the constants of the motion is
\begin{gather*}
{\cal L}_1={\cal J}_3^2-\frac{a_1 \sqrt{s_1^2+s_2^2} s_1}{s_2^2}+\frac{a_2\big(s_1^2+s_2^2\big)}{s_2^2},\\
 {\cal L}_2=-{\cal J}_1{\cal J}_3+\frac{a_1 s_3\big(s_2^2+2s_1^2\big)}{2s_2^2\sqrt{s_1^2+s_2^2}}
+\frac{a_2s_1s_3}{s_2^2}+\frac{a_3s_1}{2\sqrt{s_1^2+s_2^2}},
\end{gather*}
The structure relations are ${\cal R}=\{ {\cal L}_1,{\cal L}_2\}$ and
\begin{gather*}
\{{\cal R},{\cal L}_1\}=a_1a_3+4{\cal L}_1{\cal L}_2, \\ \{{\cal R},{\cal L}_2\}
=-6{\cal L}_1^2-2{\cal L}_2^2+4{\cal H}{\cal L}_1-2a_2{\cal H}+4a_2{\cal L}_1+\frac{a_3^2-a_1^2}{2},\\
{\cal R}^2+4{\cal L}_1^3-4{\cal L}_1^2{\cal H}+4{\cal L}_2^2{\cal L}_1
-4a_2{\cal L}_1^2+4a_2{\cal H}{\cal L}_1-a_1^2{\cal H}\\
\qquad{} -\big(a_3^2-a_1^2\big){\cal L}_1+2a_1a_3{\cal L}_2+a_2a_3^2=0.
\end{gather*}

\subsubsection[$S7$ in polar coordinates]{$\boldsymbol{S7}$ in polar coordinates}

In terms of polar coordinates $r$, $\theta$ where
$ s_1=r\cos\theta$, $s_2=r\sin\theta$, $s_3=\pm \sqrt{1-r^2}$,
and $0\le r\le 1$, $0\le \theta <\pi$,  we have
\begin{gather}
 {\cal H}=\frac{p_\theta^2}{r^2}+\big(1-r^2\big)p_r^2+\frac{a_1\cos\theta}{r^2\sin^2\theta}+\frac{a_2}{r^2\sin^2\theta}\pm\frac{a_3\sqrt{1-r^2}}{r},\nonumber\\
  {\cal L}_1=p_\theta^2+\frac{2a_1\cos\theta+2a_2}{1-\cos(2\theta)}, \label{L2polar}\\
 {\cal L}_2=\pm \sqrt{1-r^2}\left(\frac{\cos\theta\, p_\theta}{r}+\sin\theta  p_r\right)p_\theta
 \pm\sqrt{1-r^2}\frac{(a_1+a_1\cos^2\theta+2a_2\cos\theta)}{2r(1-\cos^2\theta)}+\frac12 a_3\cos\theta.\nonumber
\end{gather}
Here, $\pm$ is interpreted as $+$ in the northern hemisphere and $-$ in the southern hemisphere.
 Hamilton's equations give
 \begin{gather*}
   {\dot r}=2\big(1-r^2\big)p_r,\qquad {\dot \theta}=\frac{2p_\theta}{r^2},\qquad {\dot p}_\theta=\frac{a_1+2a_2\cos\theta+a_1\cos^2\theta}{r^2\sin^3\theta},\\
   {\dot p}_r=\frac{2p_\theta^2}{r^3}+2rp_r^2+\frac{2a_1\cos\theta}{r^3\sin^2\theta}+\frac{2a_2}{r^3\sin^2\theta}
 \pm\frac{a_3}{\sqrt{1-r^2}}\pm\frac{a_3\sqrt{1-r^2}}{r^2}.
 \end{gather*}

 {\bf Assumptions.}  We require $a_3<0$, and $a_1,a_2>0$.  and restrict our attention to trajectories caged in one
 hemisphere (eastern or western):
 \begin{gather*} -1 <s_3<1, \qquad  0<s_2<1,\qquad -1<s_1<1,\qquad {\rm or}\qquad
  1\ge r>0, \qquad  \pi>\theta> 0.
  \end{gather*}
 Note that the north pole is attractive and the south pole is repulsive.

\subsubsection{Trajectories}

Eliminating $p_\theta$, $p_r$ in the expressions for ${\cal H}$, ${\cal L}_1$, ${\cal L}_2$ we f\/ind the implicit equation for the trajectories:
\begin{gather*}
0= \cos^2\theta a_3^2r^2-4{\cal L}_1^2r^2
\cos^2\theta +4{\cal H}r^2{\cal L}_1\cos^2\theta-4\cos\theta a_3r^2{\cal L}_2-8\cos\theta {\cal L}_1{\cal L}_2r\sqrt{1-r^2}\\
\hphantom{0=}{} -4{\cal L}_1a_1\cos\theta  r^2-2a_1\cos\theta  a_3r\sqrt{1-r^2}+4{\cal H}r^2a_1\cos\theta-r^2a_1^2+4{\cal L}_1^2-4{\cal L}_1a_2\\
\hphantom{0=}{}
-4a_1{\cal L}_2r\sqrt{1-r^2}+a_1^2+4{\cal L}_2^2r^2-4a_2a_3\sqrt{1-r^2}r
+4{\cal H }r^2a_2-4{\cal H}r^2{\cal L}_1\\
\hphantom{0=}{}
+4a_3\sqrt{1-r^2}r{\cal L}_1.
\end{gather*}
This is a quadratic equation for $\cos\theta$ as a function of $r$, with  solutions
\begin{gather}\label{orbiteqnsupper}
\cos\theta=\frac{N(r)\pm 2\sqrt{S(r)}}{D(r)},
\end{gather}
in the upper hemisphere, where
\begin{gather*}
 N= a_1a_3\sqrt{1-r^2}+2a_3r{\cal L}_2-2{\cal H} ra_1+4{\cal L}_1{\cal L}_2\sqrt{1-r^2}+2r{\cal L}_1a_1, \\
 S=\big({\cal H}r^2-\sqrt{1-r^2}  ra_3-{\cal L}_1\big){\cal R}^2,\qquad
D= r\big(4{\cal H}{\cal L}_1-4{\cal L}_1^2+a_3^2\big).
\end{gather*}
We note that
\begin{gather*}
N^2-4S= \big(a_3^2+4{\cal H}{\cal L}_1-4{\cal L}_1^2\big)\big(4{\cal L}_2^2r^2-4a_3\sqrt{1-r^2}ra_2+4{\cal L}_1a_3\sqrt{1-r^2}r+4{\cal H}r^2a_2 \\
\hphantom{N^2-4S=}{} -4{\cal H}r^2{\cal L}_1-4{\cal L}_1a_2-r^2a_1^2-4a_1{\cal L}_2\sqrt{1-r^2}r+4{\cal L}_1^2+a_1^2\big).
\end{gather*}
In the lower hemisphere the trajectories are given by
\begin{gather}
 \cos\theta=\frac{{\tilde N}(r)\pm 2\sqrt{{\tilde S}(r)}}{{\tilde D}(r)},\nonumber\\
 {\tilde  N}= -a_1a_3\sqrt{1-r^2}+2a_3r{\cal L}_2-2{\cal H} ra_1-4{\cal L}_1{\cal L}_2\sqrt{1-r^2}+2r{\cal L}_1a_1,\label{orbiteqnslower} \\
 {\tilde  S}=\big({\cal H}r^2+\sqrt{1-r^2}  ra_3-{\cal L}_1\big){\cal R}^2,\qquad
{\tilde D}= r\big(4{\cal H}{\cal L}_1-4{\cal L}_1^2+a_3^2\big).\nonumber
\end{gather}
Since ${\cal R}^2\ge0$ we see that if ${\cal H}\le {\cal L}_1$ then ${\tilde S}\le 0$. Thus for this case, the trajectory is never in the lower
hemisphere. If
${\cal H}\ge {\cal L}_1$  an elementary analysis shows that ${\tilde S}(r)>0$ exactly in an interval $0<\alpha\le r\le 1$ and is nondecreasing on that interval.

 {\bf Case 1}: $a_2>a_1>0$. All trajectories are closed and periodic. We must have ${\cal L}_1>0$ for trajectories and we initially assume ${\cal R}^2>0$.
The equations for perigee and apogee (distance from the projection of the trajectory in the equatorial plane to the origin) are
\begin{gather*}
 r^2 =\frac{a_3^2+2{\cal H}{\cal L}_1\pm a_3\sqrt{a_3^2+4{\cal H}{\cal L}_1-4{\cal L}_1^2}}{2(a_3^2+{\cal H}^2)}.
 \end{gather*}
Thus, in order to have physical trajectories we must have
\begin{gather*}
a_3^2+4{\cal H}{\cal L}_1-4{\cal L}_1^2\ge 0\qquad {\rm and}\qquad  a_3^2+2{\cal H}{\cal L}_1\ge 0.
\end{gather*}
Suppose the trajectory touches the equatorial plane at $r=1$, $\theta=\theta_0$. Then taking a limit in~(\ref{L2polar}) as the trajectory goes to the
boundary we f\/ind that if the trajectory touches the equatorial plane it does so at an angle $\cos \theta$ which is a solution of the quadratic equation
\begin{gather*}
 \left(\frac14 a_3^2+({\cal H}-{\cal L}_1){\cal L}_1\right)\cos^2\theta   -(a_3{\cal L}_2-a_1({\cal H}-{\cal L}_1))
\cos\theta+{\cal L}_2^2-({\cal H}-{\cal L}_1)({\cal L}_1-a_2)=0.
\end{gather*}
Indeed, $\cos\theta_0={\big[2 (a_3{\cal L}_2+a_1({\cal L}_1-{\cal H}) )\pm 2\sqrt{({\cal H}-{\cal L}_1){\cal R}^2}\big]}/ [{a_3^2
 +4{\cal L}_1({\cal H}-{\cal L}_1)} ]$.
 Thus, a necessary condition for the trajectory to reach the equatorial plane is that
 $ {\cal H}\ge {\cal L}_1$.

\looseness=-1
If the trajectory just touches the equatorial plane but doesn't go into the lower hemisphere then we must have ${\cal H}={\cal L}_1$ and
 $ -\frac{a_3}{2}>{\cal L}_2>\frac{a_3}{2}$.
An example of touching the equatorial plane is Fig.~\ref{Fig21.pdf}.
\begin{figure}[t!]\centering
 \includegraphics[width=61mm]{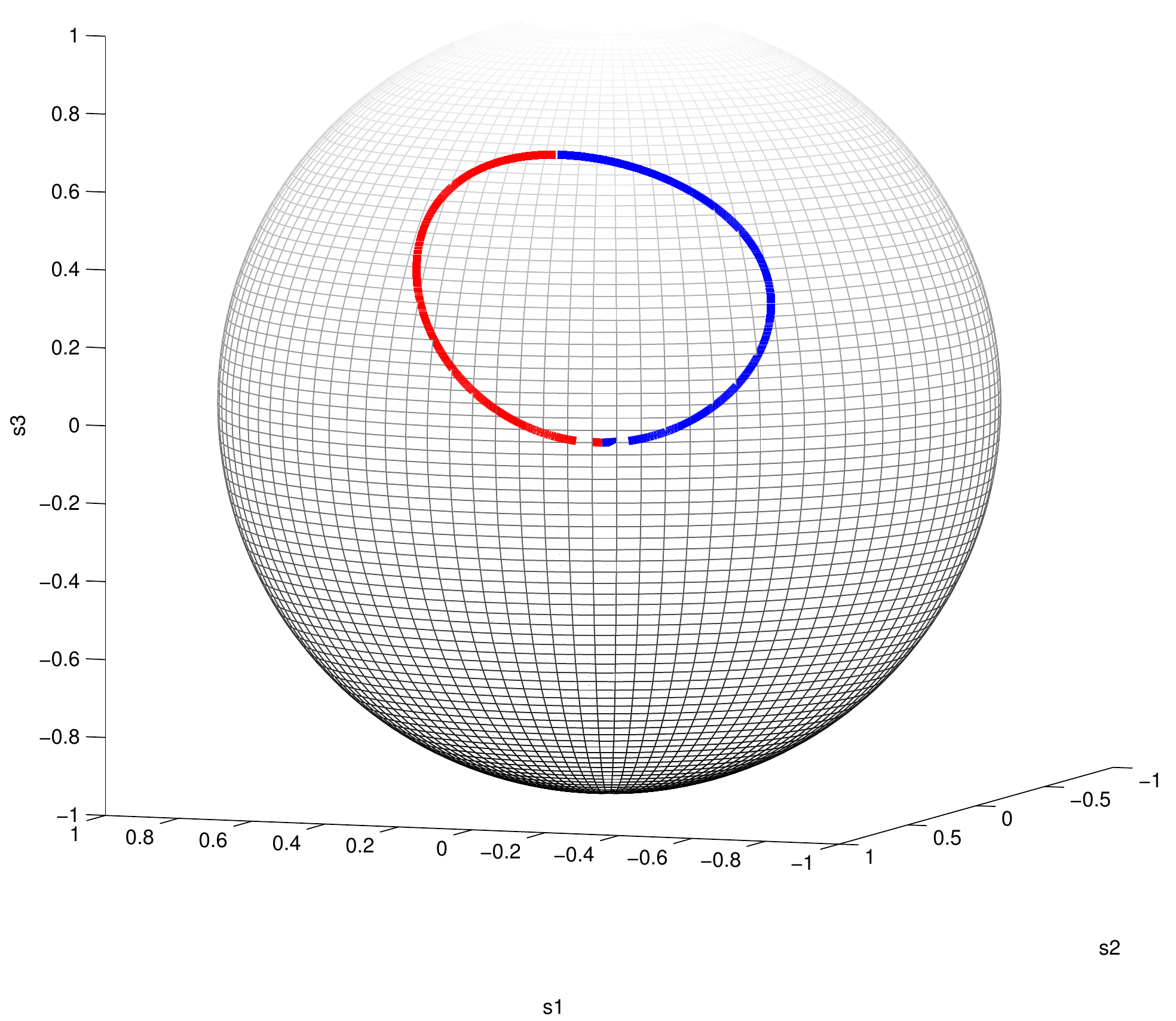}
\caption{Touching the equatorial plane: $a_1=3$, $a_2=5$, $a_3=-6$, $L_1=6$, $H=6$, $L_2=1$.}\label{Fig21.pdf}
\end{figure}
The dif\/ferent colors in the accompanying graphs correspond to the dif\/ferent curves~(\ref{orbiteqnsupper}),~(\ref{orbiteqnslower})
that make up the trajectories.
 If the trajectory passes through the equatorial plane then we must have
 $ {\cal H}>{\cal L}_1$  and $  4{\cal L}_1^2-4a-2{\cal L}_1+a_1^2\ge 0$.
 The angle of crossing lies in the interval
 \begin{gather*} -\frac{a_1}{2{\cal L}_1}- \frac{1}{2{\cal L}_1}\sqrt{4{\cal L}_1^2-4a_2{\cal L}_1+a_1^2}\le \cos\theta_0\le-\frac{a_1}{2{\cal L}_1}
 + \frac{1}{2{\cal L}_1}\sqrt{4{\cal L}_1^2-4a_2{\cal L}_1+a_1^2},
 \end{gather*}
 and for each $\theta_0$ in that interval the possible values of ${\cal L}_2$ are
\begin{gather*}
 {\cal L}_2=\frac{a_3\cos\theta_0}{2}\pm \sqrt{({\cal L}_1-{\cal H})(\cos^2\theta_0\ {\cal L}_1-{\cal L}_1+a_1\cos\theta_0\ +a_2)}.
 \end{gather*}
Note \looseness=-1 that for each choice of the constants of the motion, there are always two crossing angles. Thus if a trajectory crosses into the
lower hemisphere, it must
return to the upper hemisphere: No trajectory remains conf\/ined to the lower hemisphere.
Examples are Figs.~\ref{Fig22.pdf} and~\ref{Fig23.pdf}.
\begin{figure}[t!]\centering
 \includegraphics[width=61mm]{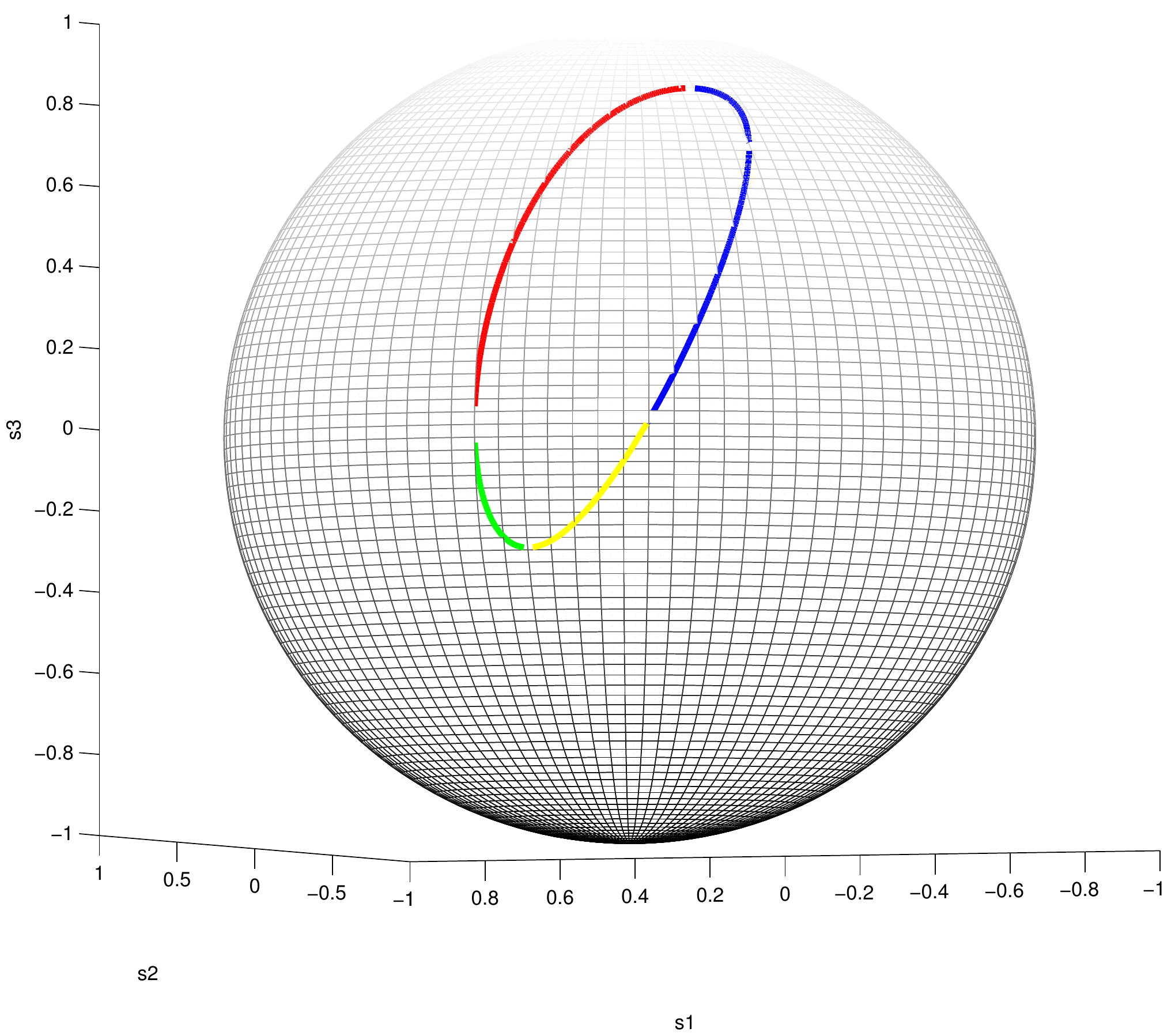}
\caption{Case 1: $a_1=3$, $a_2=4$, $a_3=-5$, $L_1=4$, $L_2=0$, $H=6$.}\label{Fig22.pdf}
\end{figure}
\begin{figure}[t!]\centering
\includegraphics[width=61mm]{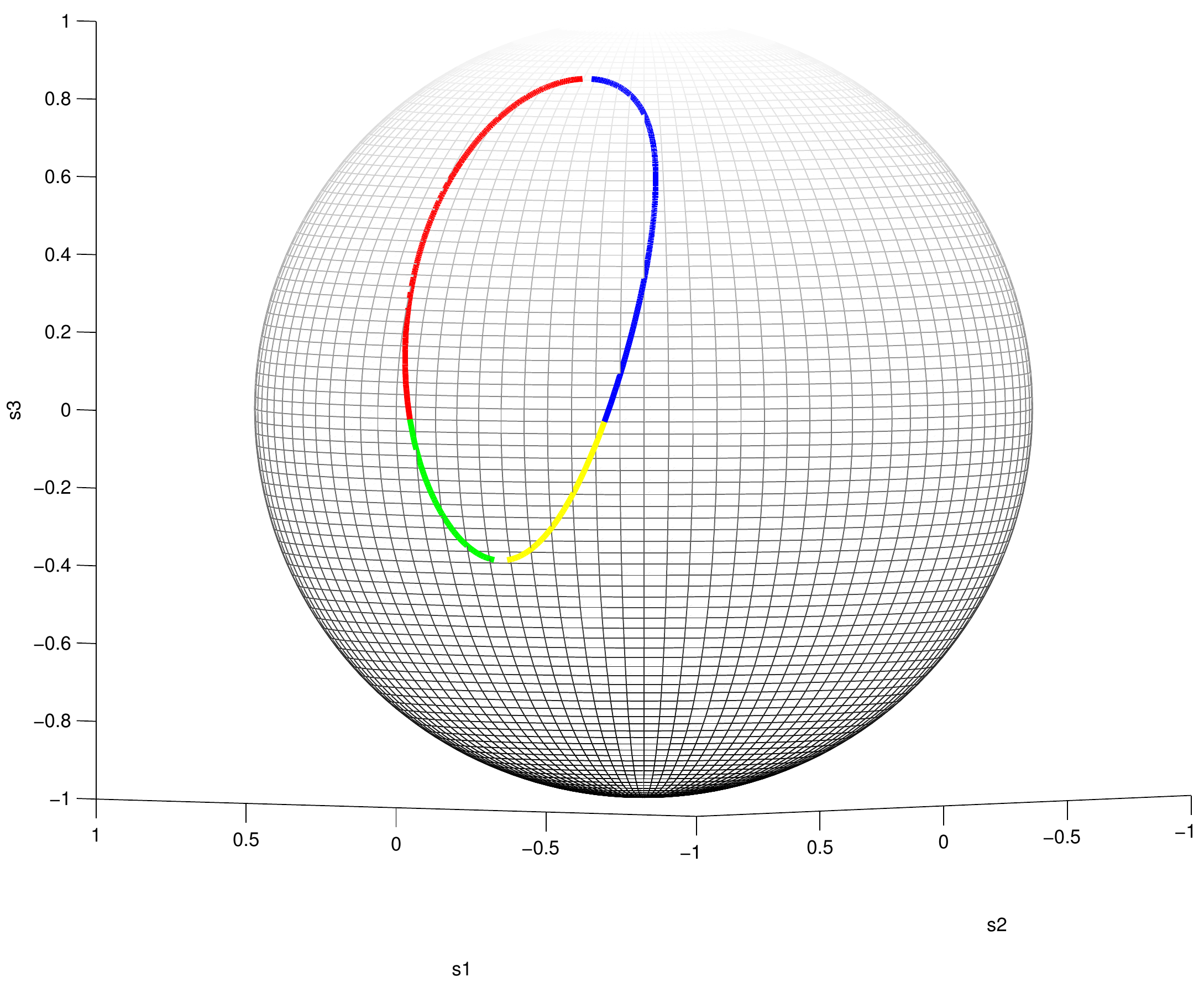}
\caption{Case 1: $a_1=2$, $a_2=3$, $a_3=-4$, $L_1=3$, $L_2=0.5$, $H=5$.}\label{Fig23.pdf}
\end{figure}
If the  angle is f\/ixed, we can use the formula to get ${\cal L}_2$: Examples are Figs.~\ref{Fig24.pdf}
and~\ref{Fig25.pdf}.
\begin{figure}[t!]\centering
\includegraphics[width=62mm]{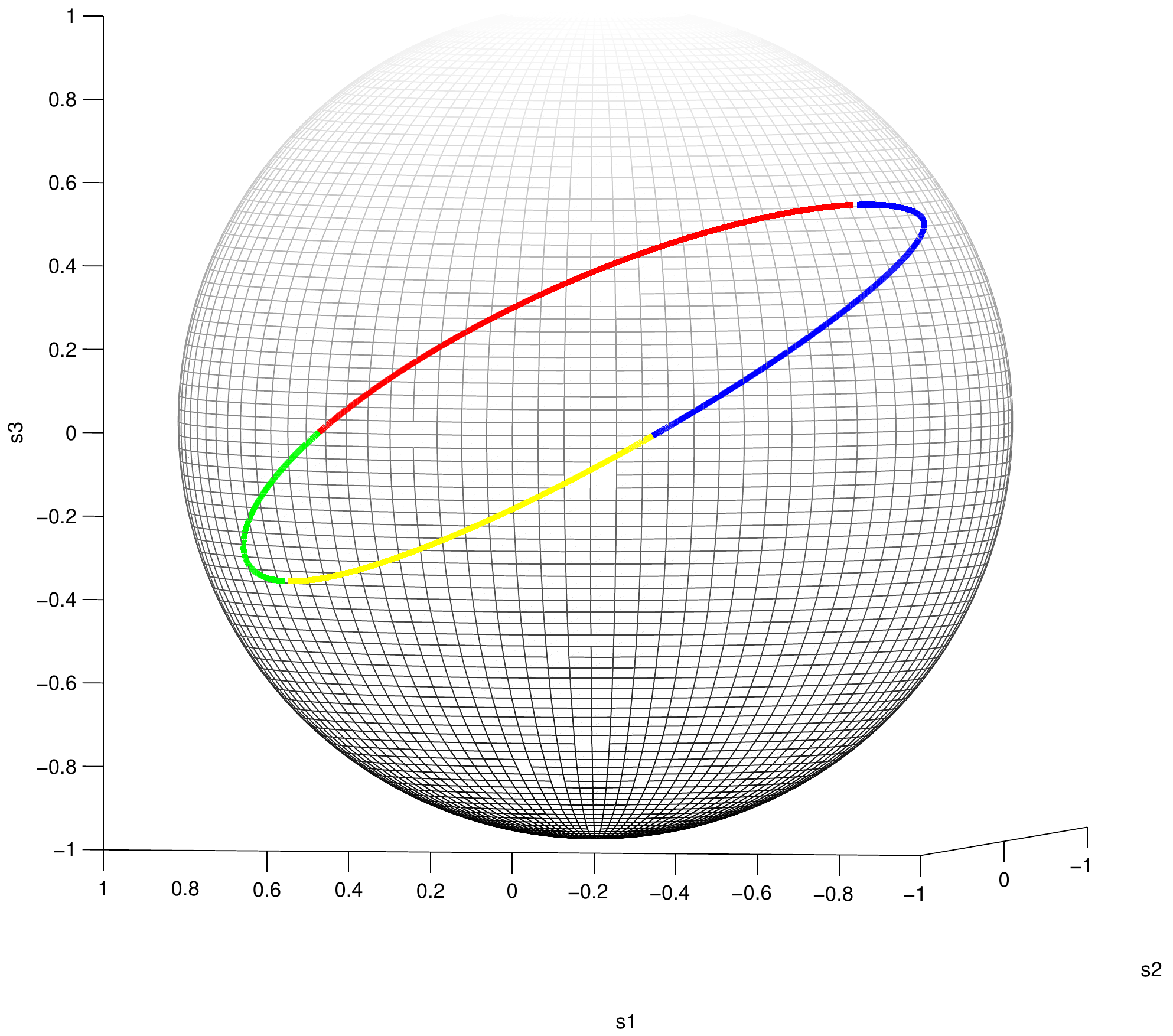}
\vspace{-1.5mm}
\caption{$a_1=8$, $a_2=10$, $a_3=-10$, $L_1=40$, $H=50$,  $L_2=-15.1491$.}\label{Fig24.pdf}
\end{figure}
\begin{figure}[t!]\centering
 \includegraphics[width=62mm]{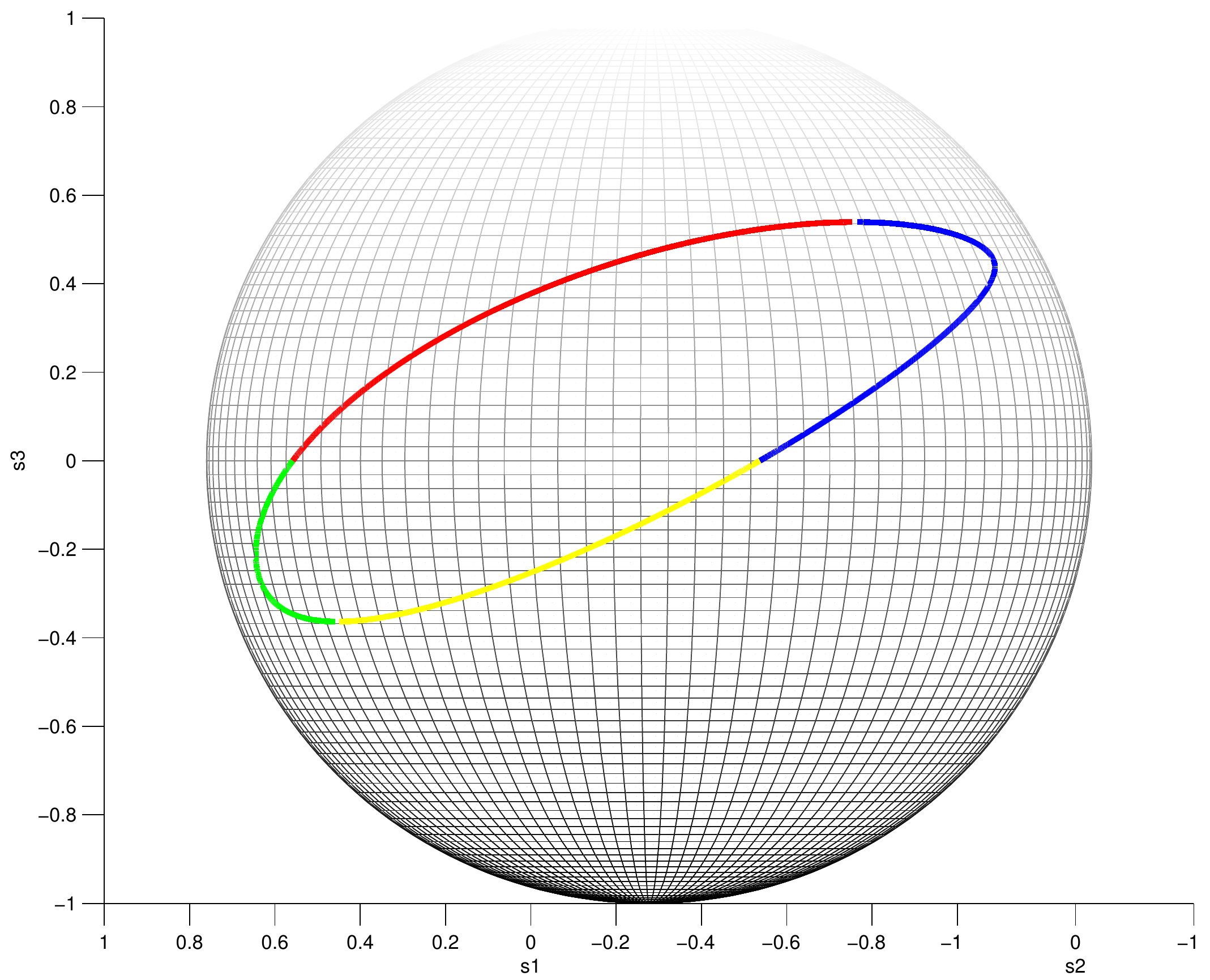}
\caption{$a_1=8$, $a_2=10$, $a_3=-10$, $L_1=40$, $H=50$, $L_2=-12.9919$.}\label{Fig25.pdf}
\end{figure}
If ${\cal R}^2=0$ then $S(r)\equiv 0$. The previous analysis is correct, except the trajectory is a~single arc, rather than a~loop.
We call this a metronome orbit. The particle moves back and  forth along the arc with a~f\/ixed period.
We give no more details here because we will study the analogous systems in our treatment of~$E16$.

{\bf Case 2}: $a_1>a_2$. Then the portion $-1<s_1<-a_2/a_1$ of the $s_1$-axis becomes attractive, the trajectories are not closed; each end of the trajectory impacts this interval, assuming  ${\cal R}^2>0$. If ${\cal H} >{\cal L}_1$ then one impact occurs in the northern hemisphere and one impact in the southern hemisphere. If ${\cal H}\le {\cal L}_1$ both impacts occur in the northern hemisphere.

An example for ${\cal H} >{\cal L}_1$ is  Fig.~\ref{Fig26.pdf},  and for case ${\cal H}<{\cal L}_1$ is Fig.~\ref{Fig27.pdf}.
If ${\cal R}^2=0$ then the trajectory is a single arc that impacts the interval $-1<s_1<-a_2/a_1$ of the $s_1$-axis in the northern hemisphere.
\begin{figure}[t!]\centering
 \includegraphics[width=62mm]{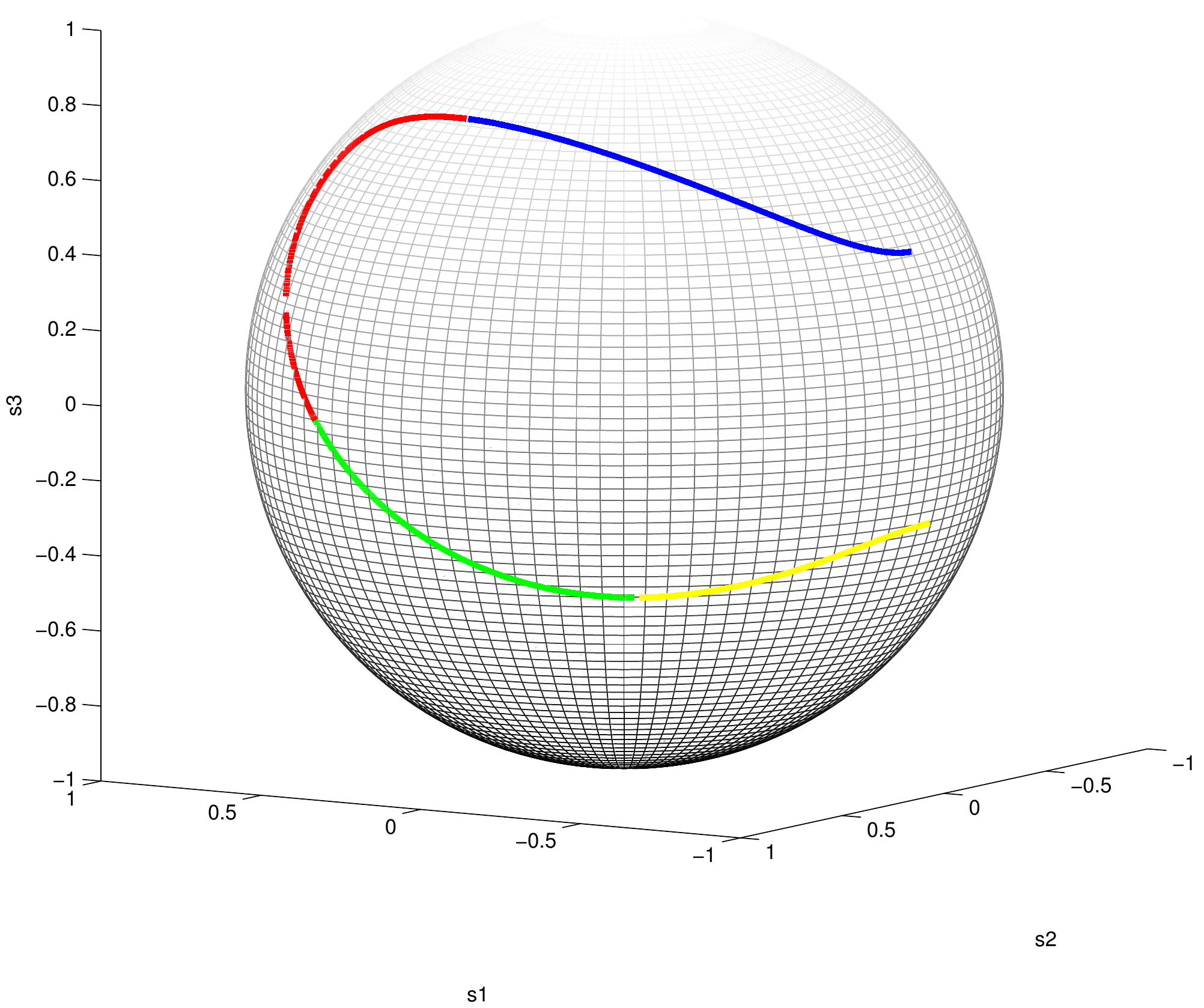}
\caption{$a_1=2$, $a_2=1$, $a_3=-4$, $L_1=5$, $H=8$, $L_2=2$.}\label{Fig26.pdf}
\end{figure}
\begin{figure}[t!]\centering
 \includegraphics[width=62mm]{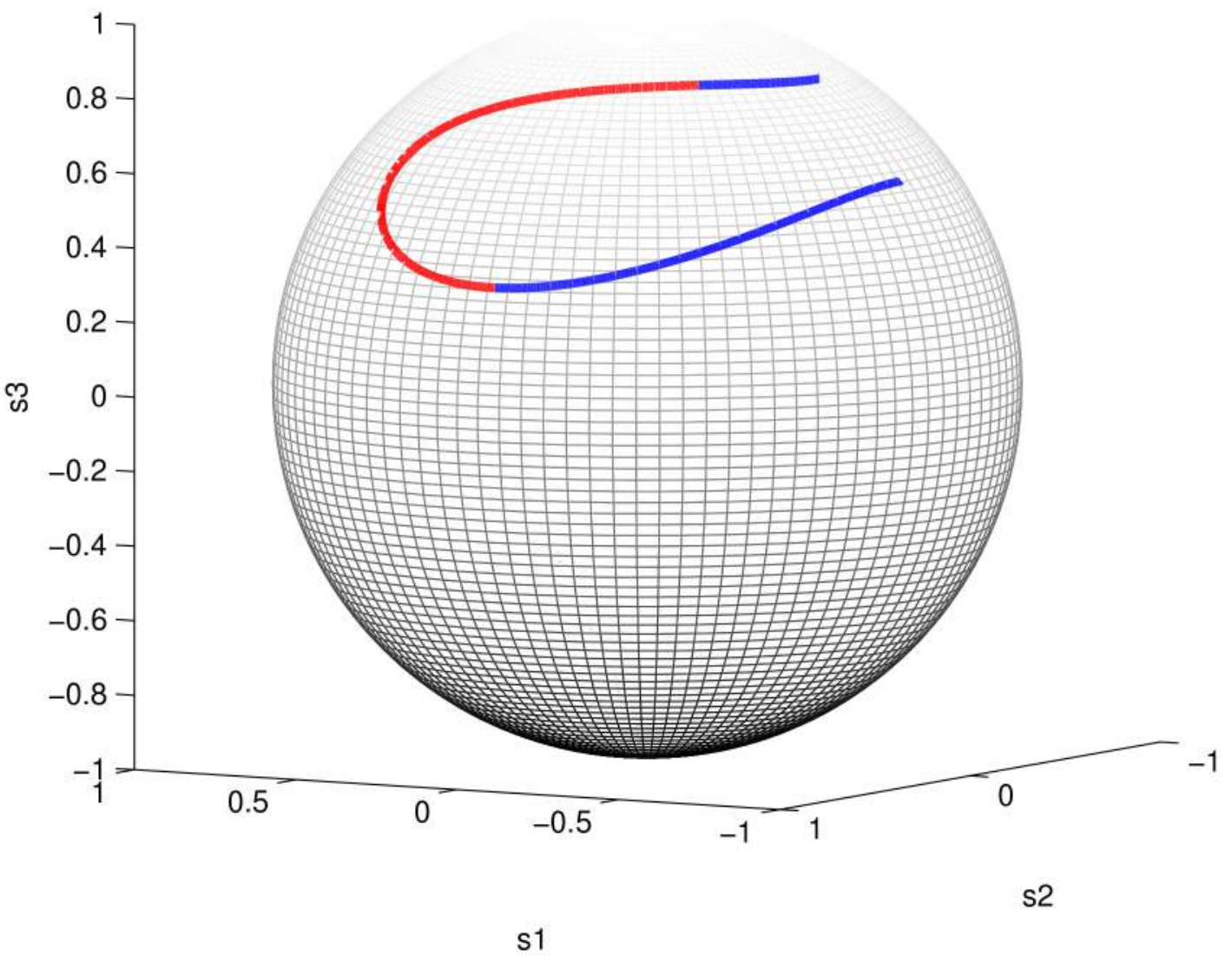}
\caption{$a_1=4$, $a_2=2$, $a_3=-10$, $L_1=5$, $H=2$, $L_2=1$.}\label{Fig27.pdf}
\end{figure}

\subsection[System $E16$ in Euclidean space as a contraction of $S7$]{System $\boldsymbol{E16}$ in Euclidean space as a contraction of $\boldsymbol{S7}$}

 $E16$ can be described as a Kepler--Coulomb system with barrier in two dimensions~\cite{KKMP, MSVW}.
The system has Hamiltonian,
\begin{gather*}
{\cal H} = p_x^2+p_y^2+\frac{a_1}{\sqrt{x^2+y^2}}+\frac{a_2y}{x^2\sqrt{x^2+y^2}}+\frac{a_3}{x^2}.
\end{gather*}
In terms of  the rotation generator ${\cal J}=xp_y-yp_x$ the def\/ining constants of the motion are $\cal H$ and
\begin{gather*}
 {\cal L}_1= {\cal J}^2+\frac{ a_2y\sqrt{x^2+y^2}}{x^2}+a_3\frac{y^2}{x^2},\\
 {\cal L}_2= p_x{\cal J }-\frac{a_1}{2}\frac{y}{\sqrt{x^2+y^2}}-a_2\left(\frac{y^2}{x^2\sqrt{x^2+y^2}}+\frac{1}{2\sqrt{x^2+y^2}}\right)-a_3\frac{y}{x^2}.
\end{gather*}
The time derivative of a function $F$ of the phase space parameters along a trajectory
is given by ${\dot F}\equiv \frac{\partial F}{\partial t}=\{{\cal H}, F\}$, so in Cartesian coordinates we
have the equations of motion
\begin{gather*}
{\dot x}=2p_x, \qquad {\dot p}_x=\frac{1}{x^3(x^2+y^2)^{3/2}}\big(a_1x^4+a_2y\big(3x^2+2y^2\big)+2a_3\big(x^2+y^2\big)^{3/2}\big), \\
{\dot y}=2p_y, \qquad {\dot p}_y=\frac{1}{(x^2+y^2)^{3/2}} ( a_1y-a_2) .
\end{gather*}
In polar coordinates $x=r\cos\theta$, $y=r\sin\theta$ we have
\begin{gather*}
 p_\theta=xp_y-yp_x,\qquad rp_r=xp_x+yp_y,\\
 p_x=-\frac{\sin\theta}{r}p_\theta+\cos\theta  p_r,\qquad
p_y=\frac{\cos\theta}{r}  p_\theta+\sin\theta   p_r,\\
 {\cal H} = p_r^2+\frac{p_\theta^2}{r^2}+\frac{a_1}{r}+\frac{a_2\sin\theta}{r^2\cos^2\theta}+\frac{a_3}{r^2\cos^2\theta}, \qquad
 {\cal L}_1= p_\theta^2+\frac{ a_2\sin\theta}{\cos^2\theta}+\frac{a_3\sin^2\theta}{\cos^2\theta},\\
 {\cal L}_2= \left(\cos\theta p_r-\frac{\sin\theta}{r}p_\theta\right)p_\theta-\frac{a_1\sin\theta}{2}-\frac{a_2}{r}\left(\frac{\sin^2\theta}{\cos^2\theta}+\frac12\right)
-\frac{a_3\sin\theta}{r\cos^2\theta}.
\end{gather*}
Setting ${\cal R}=\{{\cal L}_1,{\cal L}_2\}$ we can verify that
\begin{gather}\label{R^2}
{\cal R}^2=4({\cal L}_1+a_3)\big({\cal L}_1{\cal H}-{\cal L}_2^2\big)+a_1^2{\cal L}_1+2a_1a_2{\cal L}_2+a_2^2
 {\cal H}.
 \end{gather}
Note also that $ {\cal H}=p_r^2+\frac{a_1}{r}+\frac{{\cal L}_1+a_3}{r^2}$.
 In polar coordinates
 the equations of motion are
\begin{gather*}
 {\dot \theta}=2\frac{p_\theta}{r^2},\qquad {\dot r}=2p_r,\qquad {\dot p_r}=2\frac{{\cal L}_1+a_3}{r^3}+\frac{a_1}{r^2},\\
 {\dot p_\theta}=\frac{1}{r^3\cos^3\theta}\big(a_2\big(\cos^2\theta-2\big)-2a_3\sin\theta\big).
\end{gather*}
From these equations we see that points on a trajectory such that ${\dot \theta}=0$ are characterized by $p_\theta=0$. If such a point has coordinates $(\theta_0,r_0,p_{r_0},0)$
we f\/ind that
\begin{gather*}
  \sin\theta_0=\frac{-a_2\pm\sqrt{a_2^2+4{\cal L}_1({\cal L}_1+a_3)}}{2({\cal L}_1+a_3)},\qquad p_{\theta_0}=0 ,\\
  r_0=\frac{a_2(1+\sin^2\theta_0)+2a_3\sin\theta_0}
{(1-\sin^2\theta_0)(a_1\sin\theta_0-2{\cal L}_2)},\qquad p_{r_0}^2={\cal H}-\frac{a_1}{r_0}-\frac{{\cal L}_1+a_3}{r_0^2}.
\end{gather*}
Similarly, if ${\dot r}=0$ then $p_r=0$ and  if such a point has coordinates $(\theta_1,r_1,0,p_{\theta_1})$ we f\/ind in
particular that
$ r_1=({a_1\pm\sqrt{a_1^2+4{\cal H}({\cal L}_1+a_3)}})/({2{\cal H}})$.{\samepage
\begin{gather*}
 \text{\bf Assumption\ 1}\colon \quad  a_1<0,\qquad a_2>0,\qquad a_3>0.
 \end{gather*}
Possible values taken by the constants of the motion depend on the relative sizes of $a_2$ and $a_3$.}

{\bf Case 1}: $a_3>  a_2$.
Then ${\cal L}_1+a_3> 0$ and we must have $a_1^2+4{\cal H}({\cal L}_1+a_3)\ge 0$ for trajectories.
Unbounded trajectories correspond to ${\cal H}\ge 0$, and for ${\cal H}>0$  perigee occurs at
\begin{alignat*}{3}
  & p_{r_1}=0,\qquad && r_1=\frac{a_1+\sqrt{a_1^2+4{\cal H}({\cal L}_1+a_3)}}{2{\cal H}},&  \\
 & p^2_{\theta_1}=\frac{{\cal L}_1-a_2\sin\theta_1-({\cal L}_1+a_3)\sin^2\theta_1}{1-\sin^2\theta_1}, \qquad && \sin\theta_1=-\frac{2r_1{\cal L}_2+a_2}{2{\cal L}_1+2a_3+a_1r_1} .&
\end{alignat*}
For ${\cal H}=0$ perigee occurs at
\begin{alignat*}{3}
&p_{r_1}=0,\qquad & &r_1=-\frac{{\cal L}_1+a_3}{a_1},&  \\
&p^2_{\theta_1}=\frac{{\cal L}_1-a_2\sin\theta_1-({\cal L}_1+a_3)\sin^2\theta_1}{1-\sin^2\theta_1},
\qquad & &\sin\theta_1=-\frac{a_2}{{\cal L}_1+a_3}+2\frac{{\cal L}_2}{a_1}.&
\end{alignat*}
Bounded orbits correspond to ${\cal H}<0$, in which case perigee occurs at
\begin{alignat*}{3}
& p_{r_1}=0,\qquad & &r_1=\frac{a_1+\sqrt{a_1^2+4{\cal H}({\cal L}_1+a_3)}}{2{\cal H}},& \\
& p^2_{\theta_1}=\frac{{\cal L}_1-a_2\sin\theta_1-({\cal L}_1+a_3)\sin^2\theta_1}{1-\sin^2\theta_1},\qquad  & &\sin\theta_1=\frac{2r_1{\cal L}_2+a_2}{2{\cal L}_1+2a_3+a_1r_1}.&
\end{alignat*}
and apogee at
\begin{alignat*}{3}
&p_{r_2}=0,\qquad & &r_2=\frac{a_1-\sqrt{a_1^2+4{\cal H}({\cal L}_1+a_3)}}{2{\cal H}},&\\
&p^2_{\theta_2}=\frac{{\cal L}_1-a_2\sin\theta_2-({\cal L}_1+a_3)\sin^2\theta_2}{1-\sin^2\theta_2}, \qquad & &\sin\theta_2=\frac{2r_2{\cal L}_2+a_2}{2{\cal L}_1+2a_3+a_1r_2}.&
\end{alignat*}
We must have $a_2^2+4{\cal L}_1({\cal L}_1+a_3)\ge 0$  for trajectories. There is an equilibrium point at
\begin{gather*}
 p_{r_0}=0,\qquad p_{\theta_0}=0,\qquad r_0= -2\frac{{\cal L}_1+a_3}{a_1},\qquad \sin\theta_0=\frac{a_3-\sqrt{a_3^2-a_2^2}}{a_2}.
 \end{gather*}

{\bf Period of bounded orbit.}
Since
$\frac{dr}{dt}=2\sqrt{{\cal H}-\frac{a_1}{r}-\frac{{\cal L}_1+a_3}{r^2}}$,
we have
$ \frac{r\ dr}{2\sqrt{{\cal H}r^2-a_1r-({\cal L}_1+a_3)}}=  dt$.
Thus the time needed to go from perigee $r_1$ to apogee $r_2$ is
\begin{gather*}
 \lim_{R_1\to r_1+}  \lim_{R_2\to r_2-}\int_{R_1}^{R_2} \frac{r\ dr}{2\sqrt{{\cal H}r^2-a_1r-({\cal L}_1+a_3)}}=\frac{T}{2}.
 \end{gather*}
Clearly, this is the same time as needed to go from apogee to perigee. Hence the period $T$ is
\begin{gather*}
 T=-\frac{a_1}{2(-{\cal H})^{3/2}}\big[ \lim_{\alpha\to\infty}\arctan\alpha- \lim_{\alpha\to -\infty}\arctan\beta\big]=-\frac{a_1\pi}{2(-{\cal H})^{3/2}}.
 \end{gather*}
Thus the period depends only on the energy.

{\bf The trajectories.}
Eliminating $p_\theta$, $p_r$ in the expressions for ${\cal H}$, ${\cal L}_1$, ${\cal L}_2$ we f\/ind the orbit equations
\begin{gather}
 r=\frac{N(\sin\theta)\pm 2\sqrt{S(\sin\theta)}}{D(\sin\theta)},\qquad
N=(a_1a_2-4a_3{\cal L}_2-4{\cal L}_1{\cal L}_2)\sin\theta-2a_1{\cal L}_1-2a_2{\cal L}_2,\nonumber\\
 S= \big({-}( {\cal L}_1+a_3) \sin^2\theta+{\cal L}_1-a_2\sin\theta\big)\big(4({\cal L}_1+a_3)({\cal H}{\cal L}_1-{\cal L}_2^2)+a_1^2{\cal L}_1
 +a_2^2{\cal H}+2a_1a_2{\cal L}_2\big),\nonumber\\
 D= \big(4{\cal H}({\cal L}_1+a_3\big)+a_1^2)\sin^2\theta +(4a_1{\cal L}_2+4a_2{\cal H}) \sin\theta-4{\cal H}{\cal L}_1+4{\cal L}_2^2.\label{orbiteqns}
\end{gather}
From structure equation (\ref{R^2}) and the fact that ${\cal R}^2>0$ we see that the 2nd factor in $S$ is a nonnegative constant.
Further, from (\ref{R^2}) we see that the denominator does not vanish for ${\cal H}<0$, that it has a single root $\sin\theta=-2{\cal L}_2/a_1$ for ${\cal H}=0$,
and a double root for ${\cal H}>0$.

For the single root case with ${\cal H}=0$, assuming the plus sign in~(\ref{orbiteqns}), we set $z=  \sin\theta+2{\cal L}_2/a_1$ in~(\ref{orbiteqns})
and expand $r$  in a Laurent series  series to f\/ind  the asymptotic expression
\begin{gather*}
r= -4\frac{{\cal R}^2}{a_1^3z^2} +\frac{(2a_1a_2-8{\cal L}_1{\cal L}_2-8a_3{\cal L}_2)}{a_1^2z}  \\
\hphantom{r=}{}+\frac14\frac{\big(4{\cal R}^2({\cal L}_1+a_3)+16{\cal L}_2^2({\cal L}_1+a_3)^2-8a_1a_2{\cal L}_2({\cal L}_1+a_3)+a_1^2a_2^2\big)}{a_1R^2}+O(z).
\end{gather*}
Note that here
${\cal R}^2=-4({\cal L}_1+a_3){\cal L}_2^2+a_1^2{\cal L}_1+2a_1a_2{\cal L}_2$. If the minus sign is assumed in~(\ref{orbiteqns}), then we have the f\/inite limit
$ \lim\limits_{z\to 0}r=-\frac{a_1}{4{\cal R}^2}[4{\cal L}_1({\cal L}_1+a_3)+a_2^2]$.

 {\bf Parametrization of case 1 trajectories.} Recall that $a_1$, $a_2$, $a_3$ are f\/ixed constants with $a_1<0$, $a_3>a_2>0$.
To parametrize the trajectories we  use perigee coordinates $\ell_1$, $t_p$, $r_p$ where
$ \ell_1={\cal L}_1+a_3>0$, $-1<t_p<1$, $r_p>0$.
Then
\begin{gather*}
 {\cal H}=\frac{\ell_1}{r_p^2}+\frac{a_1}{r_p},\qquad {\cal L}_2=-\frac{t_p}{r_p}\ell_1-\frac12a_1t_p-\frac{a_2}{2r_p}, \\
 {\cal H}>0\leftrightarrow\frac{ \ell_1}{r_p}>-a_1,\qquad
{\cal H}=0\leftrightarrow\frac{\ell_1}{r_p}=-a_1,\qquad {\cal H}<0\leftrightarrow \frac{\ell_1}{r_p}<-a_1.
\end{gather*}
Note that \begin{gather}\label{NDrel}
\big(N+2\sqrt{S}\big)\big(N-2\sqrt{S}\big)=(a_2^2+4(\ell_1-a_3)\ell_1)D.
\end{gather}
For ${\cal H}<0$, apogee occurs for radial distance
$ r_a=\frac{a_1r_p^2-r_p |2\ell_1+a_1r_p |}{2(\ell_1+a_1r_p)}$.
For wedge boundaries~$S_1$,~$S_2$ (radial lines containing a point on the trajectory where~$S$ vanishes) to exist at all, we must have
$a_2^2+4(\ell_1-a_3)\ell_1\ge 0$.
Then the righthand boundary will always exist, but the left hand boundary will exist only if $2\ell_1\ge a_2$.
(However, from the ${\cal R}^2>0$ inequality below, the left hand boundary must always exist.)
The location of the wedge boundaries is def\/ined by $t=\sin\theta$, where
\begin{gather*}  S_1\colon \quad   t_1=\frac{-a_2-\sqrt{a_2^2+4\ell_1(\ell_1-a_3)}}{2\ell_1},\qquad
  S_2\colon \quad   t_2=\frac{-a_2+\sqrt{a_2^2+4\ell_1(\ell_1-a_3)}}{2\ell_1},\nonumber\\
  {\cal R}^2=\left(\frac{2\ell_1+a_1r_p}{r_p}\right)^2\big(\ell_1-a_3-a_2t_p-\ell_1t_p^2\big).
 \end{gather*}
In order for the system to correspond to a trajectory, we must have ${\cal R}^2\ge 0$. We will f\/irst require that ${\cal R}>0$, the generic case.  We see that this is violated if either
1)~$2\ell_1+a_1r_p=0$, or  2)~$t_p$ lies outside the open interval $(t_1,t_2)$.  Thus we have the conditions
\begin{gather*}
- \frac{\sqrt{a_2^2+4\ell_1(\ell_1-a_3)}}{2\ell_1}<t_p+\frac{a_2}{2\ell_1}<\frac{\sqrt{a_2^2+4\ell_1(\ell_1-a_3)}}{2\ell_1},\qquad 2\ell_1+a_1r_p\ne 0.
\end{gather*}

For ${\cal H}>0$ the zeros of the denominator occur at the values $t_\pm$ of $\sin\theta$ such that
\begin{gather*}
 t_\pm=
\frac{a_1r_pt_p-a_2\pm
2\sqrt{(\ell_1+a_1r_p)(\ell_1-a_3-a_2t_p-\ell_1t_p^2)}}{2\ell_1+a_1r_p}.
\end{gather*}
For ${\cal H}>0$, if
the numerator of the trajectory equation~(\ref{orbiteqns}) vanishes at $\sin(t)=t_n$ then $t_n=t_+$ or $t_-$ and $t_1<t_n<t_2$.
Since every zero of the denominator is a zero of one of the numerators, we must have $t_1<t_-<t_+<t_2$.  Further, if $2\ell_1<a_2$ then $D$ has at most one zero. Thus, if this case occurs then there is a single asymptote and a single wedge boundary.

{\bf Summary of necessary Case 1 conditions for trajectories with ${\cal R}^2>0$.}
\begin{gather*}
  {\rm 1)} \quad \ell_1={\cal L}_1+a_3>0,\qquad -1<t_p<1, \qquad r_p>0,\\
  {\rm 2)} \quad a_2^2+4(\ell_1-a_3)  \ell_1\ge 0,\\
  {\rm 3)} \quad 2\ell_1+a_1r_p\ne 0,\\
  {\rm 4)} \quad  -\frac{a_2+\sqrt{a_2^2+4\ell_1(\ell_1-a_3)}}{2\ell_1}<t_p<\frac{-a_2+\sqrt{a_2^2+4\ell_1(\ell_1-a_3)}}{2\ell_1}.
\end{gather*}
\begin{figure}[t!]\centering
 \includegraphics[width=58mm]{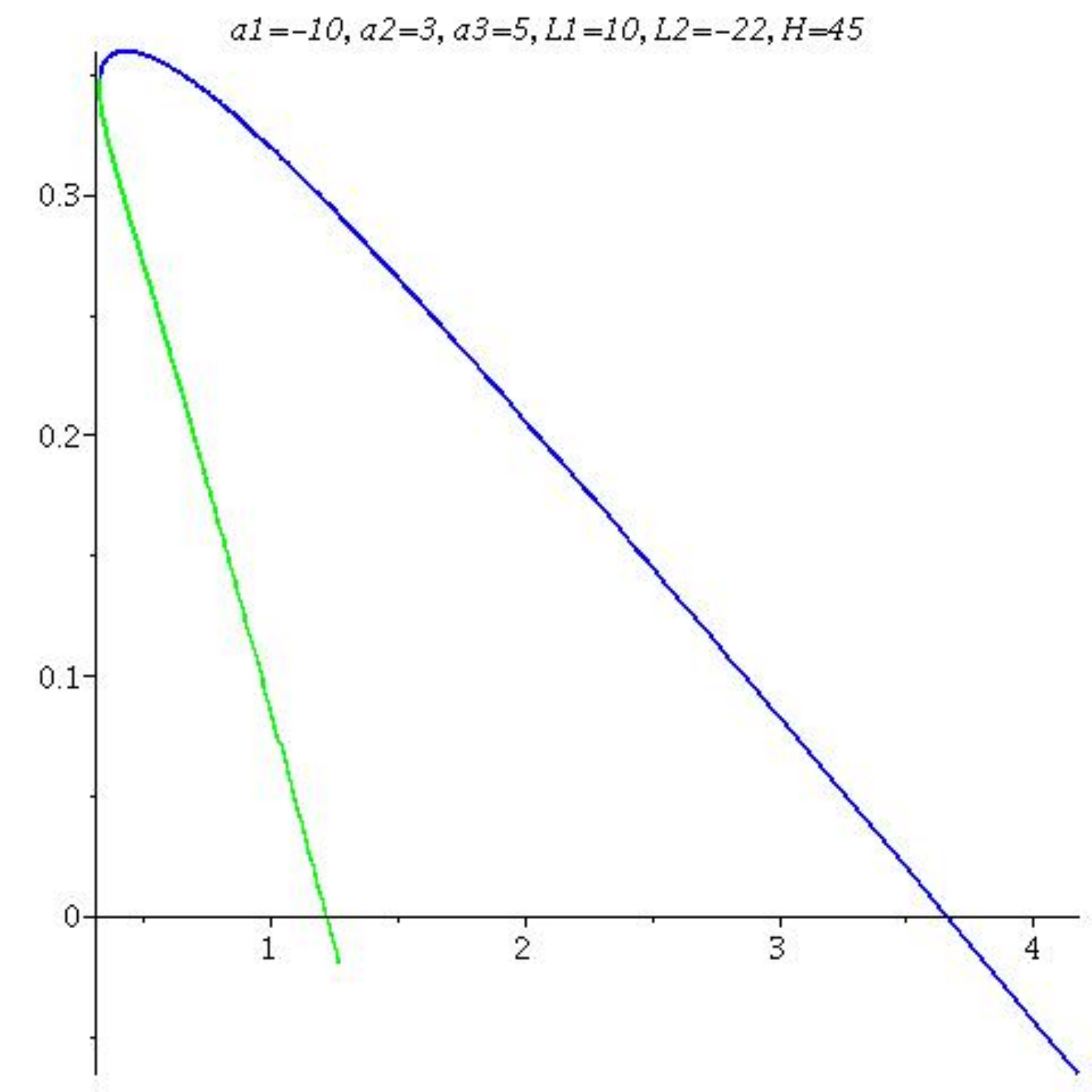}
\caption{Case 1: positive energy.}\label{Fig28.jpg}
\end{figure}
\begin{figure}[t!]\centering
 \includegraphics[width=58mm]{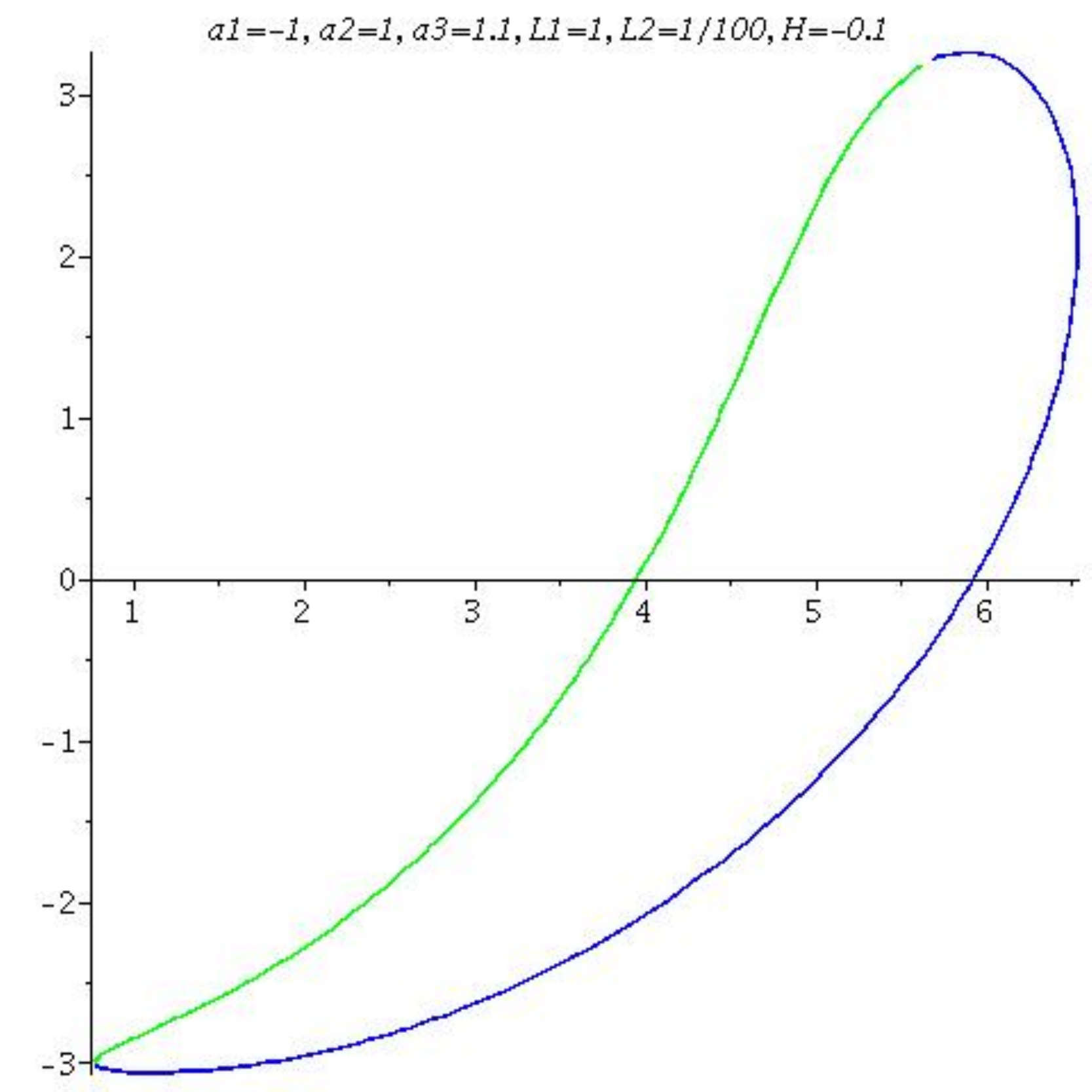}
\caption{Case 1: negative energy.}\label{Fig29.jpg}
\end{figure}
\begin{figure}[t!]\centering
 \includegraphics[width=58mm]{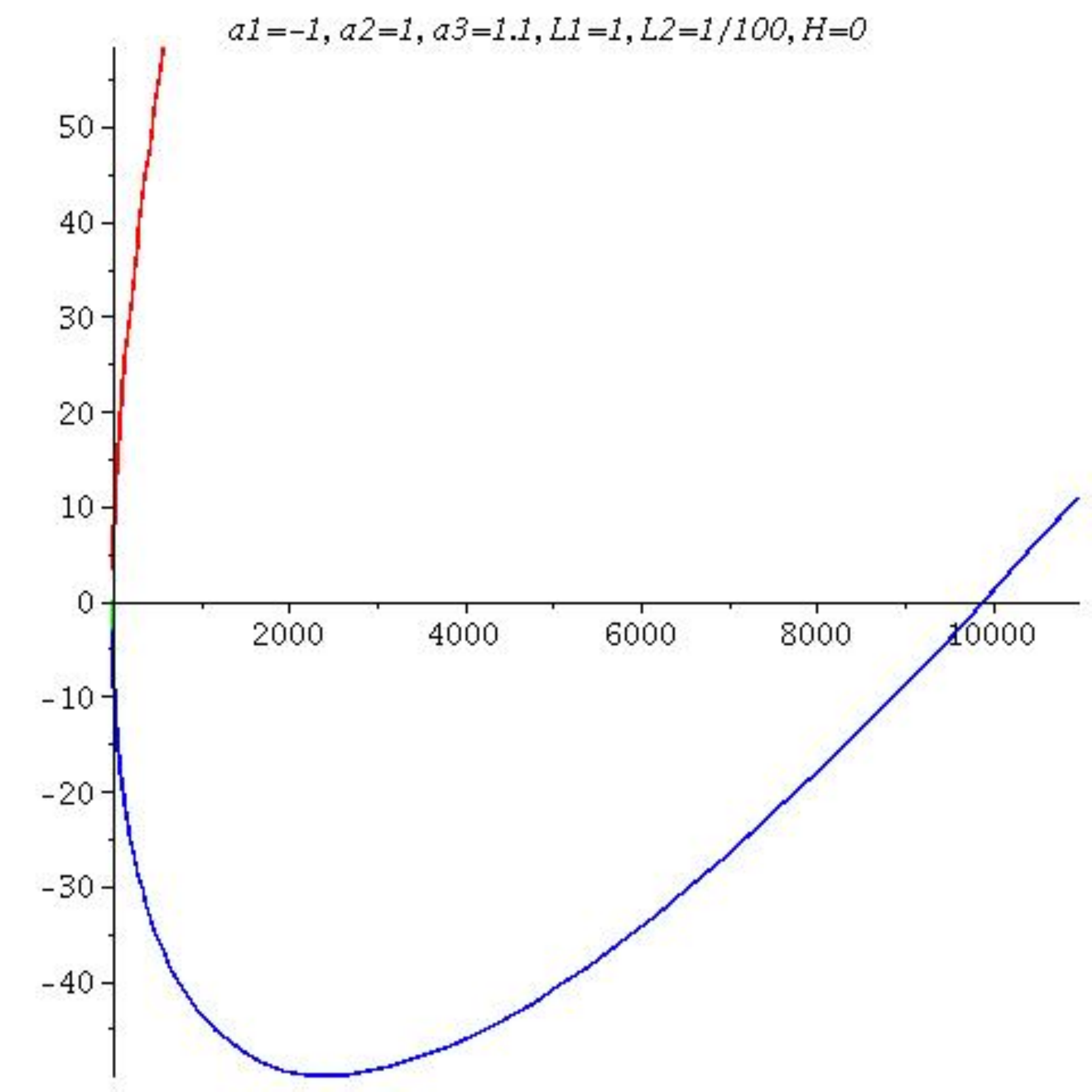}
\caption{Case 1: zero energy.}\label{Fig30.jpg}
\end{figure}
\begin{figure}[t!]\centering
\includegraphics[width=58mm]{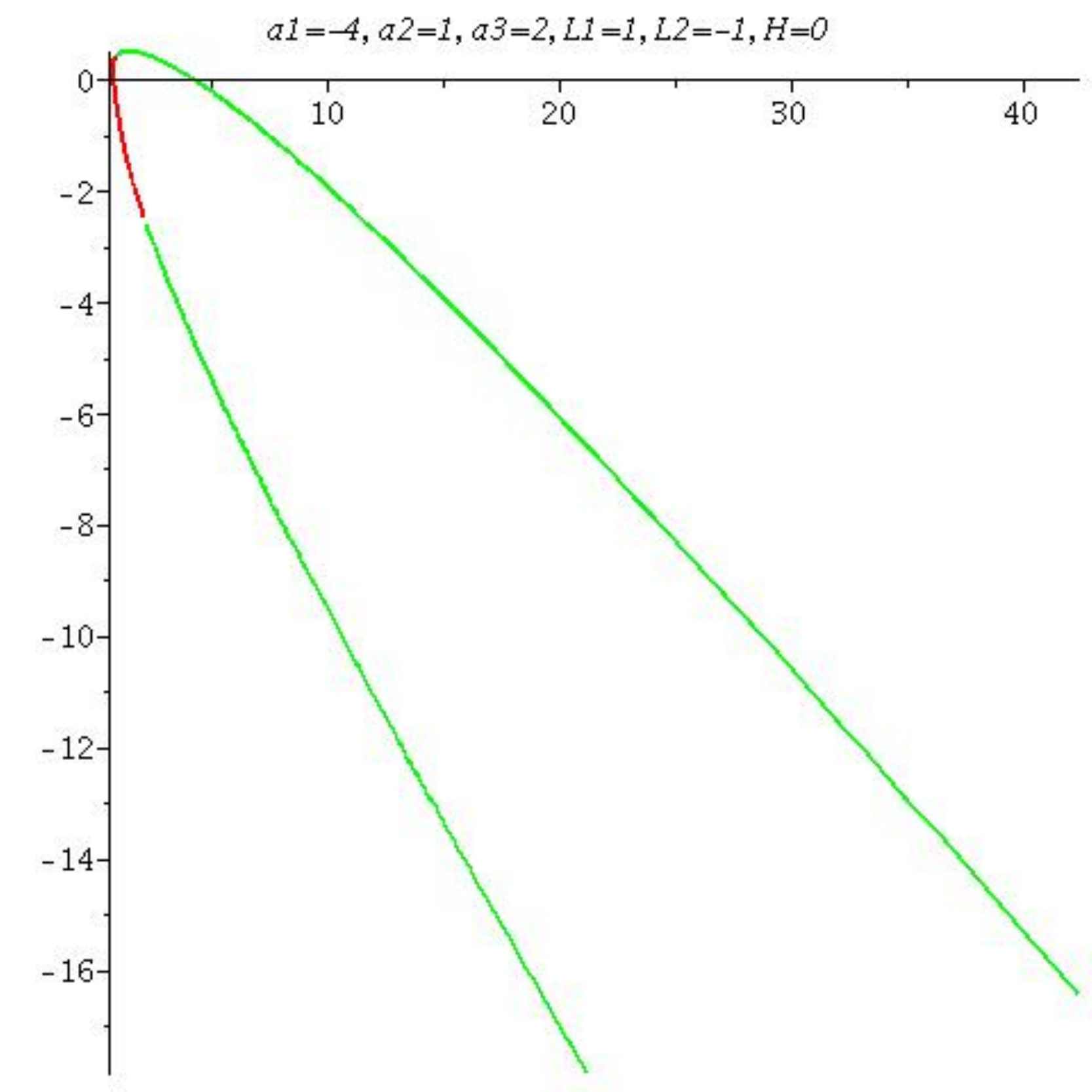}
\caption{Case 1: zero energy.}\label{Fig31.jpg}
\end{figure}
For examples, see Figs.~\ref{Fig28.jpg}, \ref{Fig29.jpg}, \ref{Fig30.jpg}, \ref{Fig31.jpg}.

{\bf Case 1a}:  ${\cal R}^2=0$. In this case equation~(\ref{NDrel}) for the trajectories simplif\/ies since  $S\equiv 0$.
We have  $ N^2=(a_2^2+4(\ell_1-a_3)\ell_1)D$ so
\begin{gather*}
r(\theta)=\frac{a_2^2+4(\ell_1-a_3)\ell_1}{N}.
\end{gather*}
Because $N$ is linear in $t=\sin\theta$, it is easily seen that the trajectories are segments of ordinary conic sections: lines, ellipses, circles, hyperbolas and parabolas.
Now  $ {\cal R}^2=(\frac{2\ell_1+a_1r_p}{r_p})^2(\ell_1-a_3-a_2t_p-\ell_1t_p^2)=0$,
so  there are two possibilities:  I)~$2\ell_1+a_1r_p=0$, and  II)~$ \ell_1-a_3-a_2t_p-\ell_1t_p^2=0$.

 {\bf Case 1aI)}: Circular arc. $2\ell_1+a_1r_p=0$. We can express the constants of the motion in terms of $r_p$, $t_p$ alone:
\begin{gather*}
 \ell_1=-\frac{a_1r_p}{2},\qquad {\cal H}=\frac{a_1}{2r_p}<0,\qquad {\cal L}_2=-\frac{a_2}{2r_p}.
 \end{gather*}
The equation for the trajectories simplif\/ies to
$ r(\theta)=\frac{a_2^2+4(\ell_1-a_3)\ell_1}{N}=r_p$,
so the trajectory is a segment of a circle centered at the origin with radius $r_p$. The wedge boundaries are
$t_1=\frac{-a_2-\sqrt{a_2^2+a_1^2r_p^2+2a_1a_3r_p}}{a_1r_p}$, $ t_2=\frac{-a_2+\sqrt{a_2^2+a_1^2r_p^2+2a_1a_3r_p}}{a_1r_p}$.
Thus for this case to occur, $r_p$ must satisfy $a_2^2+a_1^2r_p^2+2a_1a_3r_p\ge 0$.
We conclude that this case always occurs for~$r_p$ suf\/f\/iciently small or suf\/f\/iciently large,
but that there is a~forbidden intermediate band for which no circular trajectory exists.

  {\bf Case 1aII)}: Conic section arc. $ \ell_1-a_3-a_2t_p-\ell_1t_p^2=0$. Again
we can express the constants of the motion in terms of $r_p$, $t_p$ alone:
\begin{gather*}
 \ell_1=\frac{a_3+a_2t_p}{1-t_p^2},\qquad {\cal H}=\frac{a_3+a_2t_p}{r_p^2(1-t_p^2)}+\frac{a_1}{r_p},\qquad {\cal L}_2=-\frac{t_p(a_3+a_2t_p)}{r_p(1-t_p^2)}-\frac{a_1t_p}{2}-\frac{a_2}{2r_p}.
 \end{gather*}
The equation for the trajectories becomes
\begin{gather*}
 r(\theta)=\pm \frac{r_p(a_2+2t_pa_3+a_2t_p^2)}{(2a_2t_p-a_1t_p^2r_p+a_1r_p+2a_3)\sin\theta+(1-t_p^2)(a_2-a_1t_pr_p)}=\pm \frac{1}{\alpha\sin\theta+\beta},
\end{gather*}
where
\begin{gather*} \alpha=\frac{(2a_2t_p-a_1t_p^2r_p+a_1r_p+2a_3)}{r_p(a_2+2t_pa_3+a_2t_p^2)}, \qquad \beta=
\frac{(1-t_p^2)(a_2-a_1t_pr_p)}{r_p(a_2+2t_pa_3+a_2t_p^2)},
\end{gather*}
and the sign is chosen so that $r>0$. Setting $y=r\sin\theta$, $x=r\cos\theta$ we see that the trajectory is a~segment of the curve $\alpha y+\beta r=\pm 1$ or $\beta r=-\alpha y\pm 1$, so
$ \beta^2(x^2+y^2)=\alpha^2 y^2\pm 2\alpha y+1$, $\alpha^2+\beta^2>0$.
We have the following possibilities for the curve segments:
\begin{alignat*}{3}
&\beta^2>\alpha^2>0\colon \quad \text{ellipse},\qquad &  & 0<\beta^2<\alpha^2\colon \quad\text{hyperbola}, & \\
&\beta^2=\alpha^2\colon \quad\text{parabola},\qquad &   &\alpha=0\colon  \quad\text{circle}, & \\
&\beta=0\colon  \quad\text{horizontal  line}.\qquad &&&
\end{alignat*}
The wedge boundaries are
$t_p$, $-\frac{a_2+a_3t_p}{a_3+a_2t_p}$.
For the circular arc we have
$r_p=-2\frac{a_2t_p+a_3}{a_1(1-t_p^2)}$.
(This coincides with Case~1aI.)
For the horizontal line we have
$ t_p=\frac{a_2r_p}{a_1}<0$,
so trajectories are possible only for $r_p<-a_1/a_2$. Only strictly positive energy solutions are possible for the horizontal line.
For the parabolic arc we have
$ r_p={-2a_3-a_2-2a_2t_p+a_2t_p^2}/{a_1(1+t_p)(1-t_p)^2}$.
The energy is always negative.
Only strictly negative energies are possible for the elliptical arc.

In general, orbits are possible for ${\cal H}<0$ if and only if the constants of the motion satisfy the inequality
${a_3+a_2t_p}/{(1-t_p^2)(-a_1)}<r_p$.
Trajectories are possible in the positive energy case if and only if
$ r_p<{a_3+a_2t_p}/{(1-t_p^2)(-a_1)}$.

{\bf Case 2}: $a_3<  a_2$.   Now there are no restrictions on the sign of $\ell_1={\cal L}_1+a_3$.
The  force in the angular direction is always negative.
In the f\/irst quadrant the force in the $y$ direction is always negative. It follows from this that there are no periodic
 orbits contained entirely in the f\/irst quadrant.
 The positive $y$-axis is repelling, but the negative $y$-axis is attracting.

 We assume, initially, that ${\cal R}^2>0$. Equations~(\ref{orbiteqns}) for the trajectories still hold. The condition
$ a_2^2+4\ell_1(\ell_1-a_3)> 0$
 is now satisf\/ied automatically.  There is at most  a single wedge boundary
\begin{gather*} t_0=\sin\theta_0=\frac{-a_2 + \sqrt{a_2^2+4\ell_1(\ell_1-a_3)}}{2\ell_1}<0.
\end{gather*}
The following general types of behavior occur:

1.~If $ \ell_1<0$ the force in the radial direction is always negative.
 The trajectories all impact the lower $y$-axis and are perpendicular to the axis at impact.
Since~${\dot p}_r $ is always strictly negative along a trajectory,  perigee must occur at $y_p=-r_p$ on the negative $y$-axis.
Note that ${\cal R}^2$ is linear in $\cal H$, always  with positive linear coef\/f\/icient. The critical value of $\cal H$ such that ${\cal R}^2=0$ is
$ E_0=[{4a_3{\cal L}_2^2+4{\cal L}_1{\cal L}_2^2-2a_1a_2{\cal L}_2-a_1^2{\cal L}_1}]/[{4{\cal L}_1^2+a_2^2+4a_3{\cal L}_1}]$.
The qualitative behavior of the trajectories divides into two basic classes, depending on the sign of $N(\sin\theta)$ for $\theta=-\frac{\pi}{2}$.
Note that
\begin{gather*}
 N(-1)=(2{\cal L}_2-a_1)(2\ell_1-a_2)+2a_1(a_3-a_2),\qquad S(-1)=(a_2-a_3){\cal R}^2.
 \end{gather*}

  {\bf Case 2a}: $N(-1)>0$.  Then if we increase ${\cal H}$ from $E_0$ to $E_0+\epsilon$ for arbitrarily small $\epsilon>0$  there will be two
points of intersection of the trajectory with the negative $y$-axis,  one on the curve $r =(N+2\sqrt{S})/D$ and one on the curve
$r =(N-2\sqrt{S})/D$. The trajectory will remain bounded and there will be a wedge boundary.  If we further increase~${\cal H}$,
leaving ${\cal L}_1,{\cal L}_2$ unchanged,
this behavior will persist until the critical value
$E_1=\frac14 [(2{\cal L}_2-a_1)^2]/[{(a_2-a_3)}]\ge E_0$,
where $D(-1)=0$. Here the intersection point with the negative $y$-axis and the curve $r =(N-2\sqrt{S})/D$ has moved to $y=-\infty$
whereas the intersection point with the curve $r =(N+2\sqrt{S})/D$ remains f\/inite, As ${\cal H}$ is further increased the behavior
of the trajectory is that it is unbounded and asymptotic to a radial line in the 2nd quadrant, that there is a wedge boundary and
a single intersection point with the negative $y$-axis.

  {\bf Case 2b}: $N(-1)<0$. Again, we increase ${\cal H}$ from $E_0$ to $E_0+\epsilon$ for arbitrarily small $\epsilon>0$. Now there
are no points of intersection with the negative $y$-axis and, in fact, the trajectory is non-physical until  $\cal H$
is increased to $\epsilon$ beyond the critical value $E_1$. The curve is entirely described by the function $r =(N+2\sqrt{S})/D$.
The trajectory  is unbounded, asymptotic to a~radial line in the 2nd quadrant and intersects the negative $y$-axis a single time.
There is no wedge boundary. Note that for this case, the region $\{ N(-1)< 0, {\cal H}\le E_1\}$ is nonphysical.
 (We do not have a complete proof of this but strong numerical evidence.)

2.~If $\ell_1>0$ the behavior is dif\/ferent. The trajectories do not necessarily intersect the negative $y$-axis. Some trajectories are disjoint.
The qualitative behavior of the trajectories divides into two basic classes, depending on the sign of $N(\sin\theta)$ for $\theta=-\frac{\pi}{2}$.
Note again that
\begin{gather*}
 N(-1)=(2{\cal L}_2-a_1)(2\ell_1-a_2)+2a_1(a_3-a_2),\qquad S(-1)=(a_2-a_3){\cal R}^2.
 \end{gather*}

 {\bf Case 2c}: $N(-1)>0$.  If we increase ${\cal H}$ from $E_0$ to $E_0+\epsilon$ for arbitrarily small $\epsilon>0$  there will be two
points of intersection of the trajectory with the negative $y$-axis,  one on the curve $r =(N+2\sqrt{S})/D$ and one on the curve
$r =(N-2\sqrt{S})/D$. The trajectory will remain bounded if ${\cal H}<0$ (with a wedge boundary)  and be unbounded if ${\cal H}>0$.
The unbounded trajectory will divide into two disconnected parts.  If we further
 increase ${\cal H}$, leaving ${\cal L}_1,{\cal L}_2$ unchanged,
this behavior will persist until the critical value
$E_1=\frac14[{(2{\cal L}_2-a_1)^2}][{(a_2-a_3)}]\ge E_0$,
where $D(-1)=0$. Here the intersection point with the negative $y$-axis and the curve $r =(N-2\sqrt{S})/D$ has moved to $y=-\infty$
whereas the intersection point with the curve $r =(N+2\sqrt{S})/D$ remains f\/inite, As ${\cal H}$ is further increased the behavior
of the trajectory is that it is unbounded and asymptotic to a radial line in the 2nd quadrant,  and
a single intersection point with the negative $y$-axis. There is a wedge boundary.

 {\bf Case 2d}: $N(-1)<0$.  Again, we increase ${\cal H}$ from $E_0$ to $E_0+\epsilon$ for arbitrarily small $\epsilon>0$.
There
are no points of intersection with the negative $y$-axis and, the trajectory is  non-physical for ${\cal H}<0$,
but physical for the  ${\cal H}>0$ interval until~$\cal H$
is increased to $\epsilon$ beyond the critical va\-lue~$E_1$.
Then the  trajectory  is unbounded and asymptotic to a radial line in the 2nd quadrant. Physical trajectories have a wedge boundary.
See Fig.~\ref{Fig32.jpg}.
\begin{figure}[t!]\centering
 \includegraphics[width=60mm]{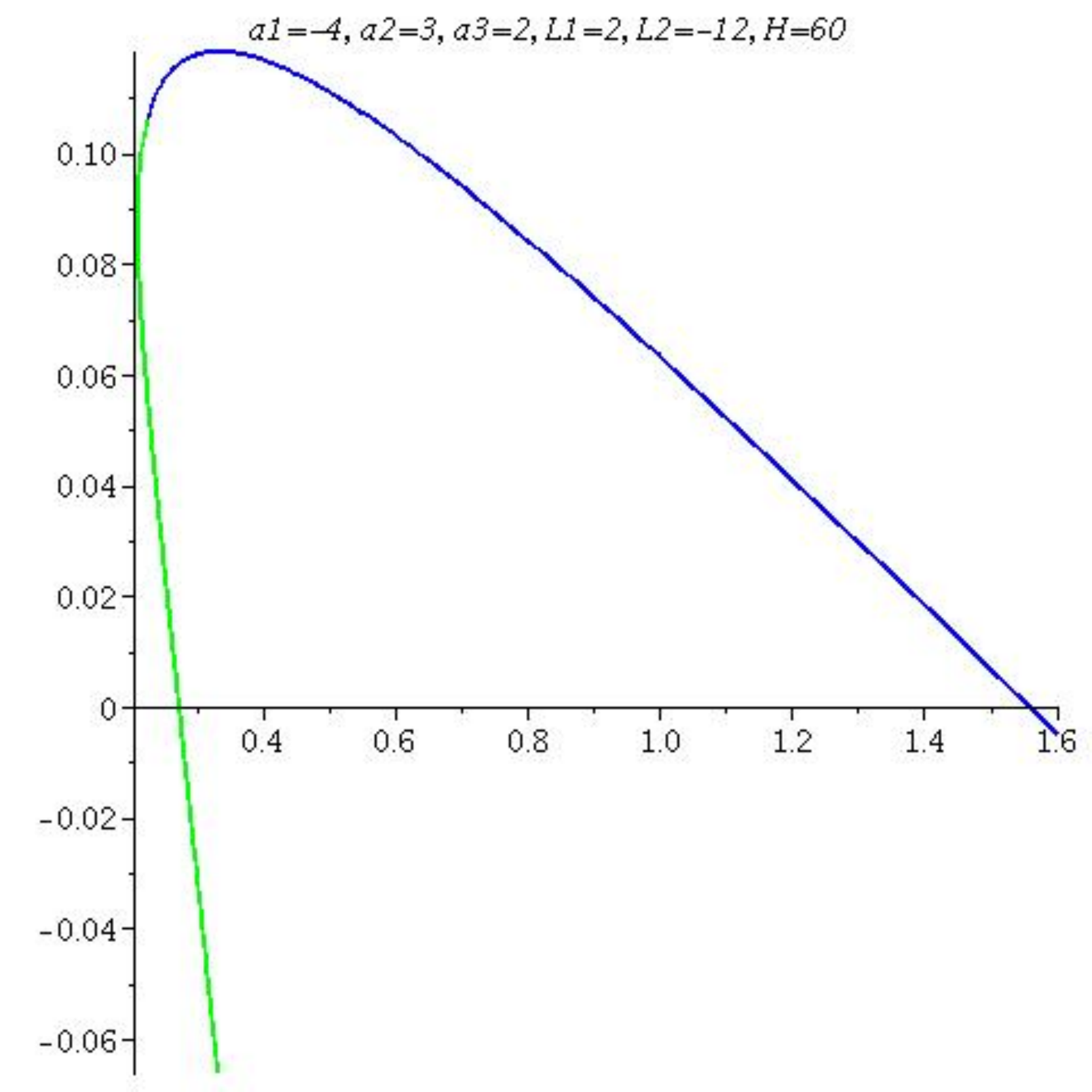}
\caption{Case 2d: $N(-1)<0$.}\label{Fig32.jpg}
\end{figure}

3.~${\cal R}^2=0$.
In this case equation (\ref{NDrel}) for the trajectories simplif\/ies since  $S\equiv 0$. We have  $ N^2=(a_2^2+4(\ell_1-a_3)\ell_1)D$ so
\begin{gather*}
\label{rR1=0} r(\theta)=\pm\frac{a_2^2+4(\ell_1-a_3)\ell_1}{N}
 =\pm\frac{a_2^2+4(\ell_1-a_3)\ell_1}{(a_1a_2-4{\cal L}_2\ell_1)t+2(a_1a_3-a_2{\cal L}_2-a_1\ell_1)}=\pm \frac{1}{\alpha\sin\theta+\beta},
 \end{gather*}
where
$\alpha=\frac{a_1a_2-4{\cal L}_2\ell_1}{a_2^2+4(\ell_1-a_3)\ell_1}$, $\beta=
\frac{2(a_1a_3-a_2{\cal L}_2-a_1\ell_1)}{a_2^2+4(\ell_1-a_3)\ell_1}$,
and the sign is chosen so that $r>0$. Setting $y=r\sin\theta, x=r\cos\theta$ we see that the
trajectory is a segment of the curve $\alpha y+\beta r=\pm 1$ or $\beta r=-\alpha y\pm 1$, so
$ \beta^2(x^2+y^2)=\alpha^2 y^2\pm 2\alpha y+1$,$\alpha^2+\beta^2>0$.
We have the following possibilities for the curve segments:
\begin{alignat*}{3}
&\beta^2>\alpha^2>0\colon \quad\text{ellipse},\qquad &  &0<\beta^2<\alpha^2\colon \quad\text{hyperbola},&\\
&\beta^2=\alpha^2\colon \quad\text{parabola},\qquad & &\alpha=0\colon \quad\text{circle},&\\
&\beta=0\colon \quad\text{horizontal  line}.&&&
\end{alignat*}
Since ${\cal R}^2=0$ we can express the constants of the motion in terms of $\alpha$, $\beta$ alone:
\begin{gather*}   {\cal H}=\frac{a_2\alpha^3-a_1\alpha^2-2a_3\alpha^2\beta+a_2\alpha\beta^2+a_1\beta^2}{2\beta},\qquad
\ell_1=\frac{a_2\alpha-a_1}{2\beta},\\ 
 {\cal L}_2=-\frac{a_2\alpha^2-a_1\alpha-2a_3\alpha\beta+a_2\beta^2}{2\beta}.
\end{gather*}

 {\bf Case 2e}: circular arc, $a_1a_2-4{\cal L}_2\ell_1=0$.
Then ${\cal H}=-\frac{a_1^2}{4\ell_1}$. The radius is $r_0=\pm \frac{2\ell_1}{a_1}$. There is a wedge boundary at
$ t_0=\sin\theta_0=\frac{-a_2 + \sqrt{a_2^2+4\ell_1(\ell_1-a_3)}}{2\ell_1}$.
The trajectory impacts the negative $y$-axis.
\begin{figure}[t!]\centering
 \includegraphics[width=60mm]{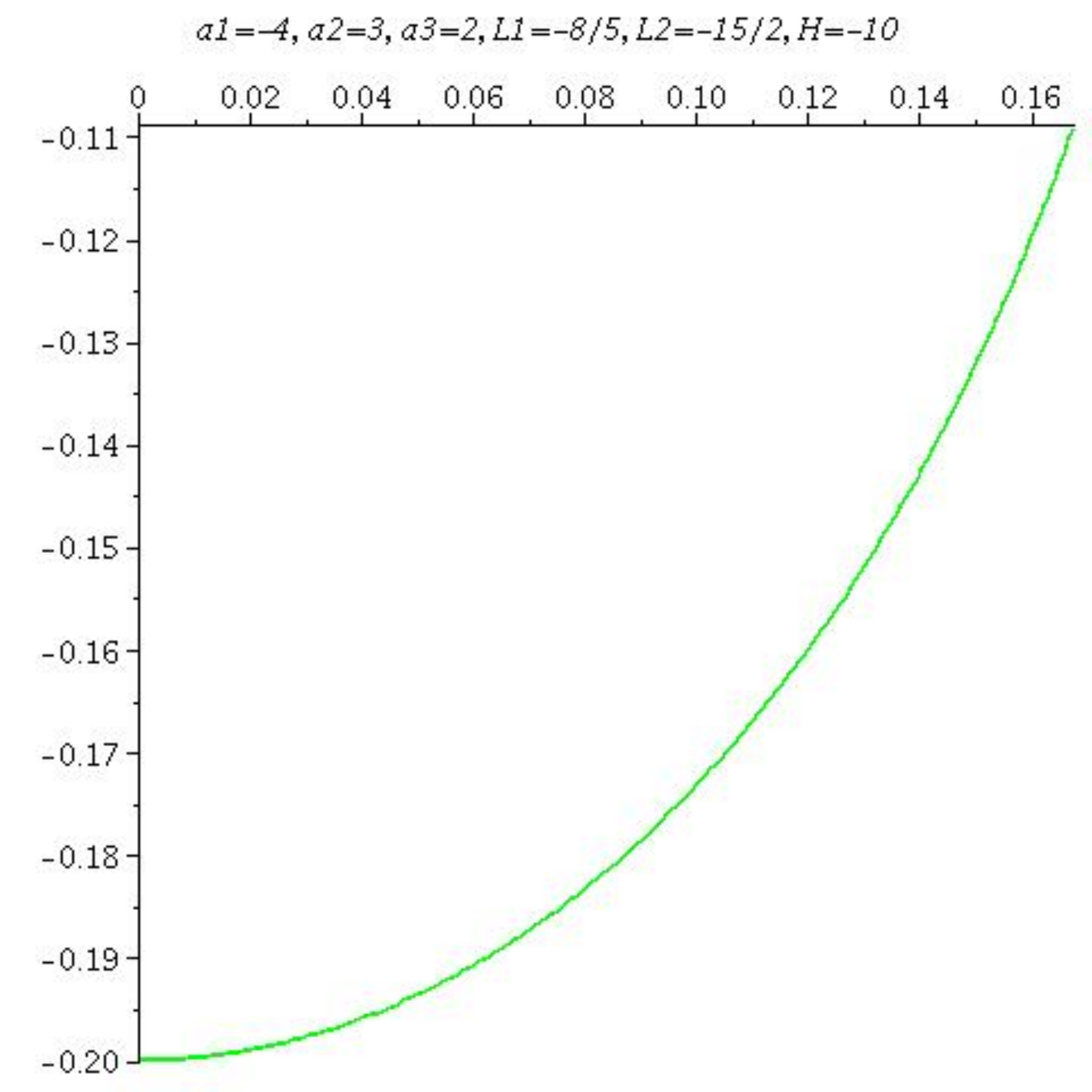}
\caption{Case 2e: circular arc.}\label{Fig33.jpg}
\end{figure}
See Fig.~\ref{Fig33.jpg}.

 {\bf Case 2f}: horizontal line segment, $a_2{\cal L}_2+a_1\ell_1=a_1a_3$,
$y_0=\frac{a_2}{a_1}$. The energy is ${\cal H}=\frac{a_1^2(\ell_1-a_3)}{a_2^2}$. If the energy is negative there is a wedge boundary at
$t_0=\sin\theta_0=\frac{-a_2 + \sqrt{a_2^2+4\ell_1(\ell_1-a_3)}}{2\ell_1}<0$,
and the trajectory impacts the negative $y$-axis. If the energy is positive, the trajectory is unbounded and there is no wedge boundary,
 but the trajectory still impacts the negative $y$-axis.

  {\bf Case 2g}: parabolic arc.   If $\alpha=\beta$ we have
\begin{gather*}
 {\cal L}_2=\frac{a_1(a_2-2a_3+2\ell_1)}{2(2\ell_1-a_2)},\qquad {\cal H}=\frac{(a_2-a_3)a_1^2}{(2\ell_1-a_2)^2}
 \end{gather*}
and the energy is positive. If $\alpha=-\beta$ we have
\begin{gather*}
 {\cal L}_2=\frac{a_1(a_2+2a_3-2\ell_1)}{2(2\ell_1+a_2)},\qquad {\cal H}=-\frac{(a_2+a_3)a_1^2}{(2\ell_1+a_2)^2}
 \end{gather*}
and the energy is negative. There is a wedge boundary at
$t_0=\sin\theta_0=\frac{-a_2 + \sqrt{a_2^2+4\ell_1(\ell_1-a_3)}}{2\ell_1}$,
and the trajectory impacts the negative $y$-axis.

  {\bf Case 2h}: elliptic arc.
\begin{gather*}
{\cal H}= \frac{4{\ell}_1{\cal L}_2^2-2a_1a_2{\cal L}_2-a_1^2({\ell }_1-a_3)}{a_2^2+4\ell_1(\ell_1-a_3)},
\end{gather*}
the general expression. There is a wedge boundary at
$t_0=\sin\theta_0=\frac{-a_2 + \sqrt{a_2^2+4\ell_1(\ell_1-a_3)}}{2\ell_1}$,
and the trajectory impacts the negative $y$-axis.

  {\bf Case 2i}: hyperbolic arc.
 Here again,
\begin{gather*}
{\cal H}= \frac{4{\ell}_1{\cal L}_2^2-2a_1a_2{\cal L}_2-a_1^2({\ell }_1-a_3)}{a_2^2+4\ell_1(\ell_1-a_3)},
\end{gather*}
the general expression. There is a wedge boundary at
$t_0=\sin\theta_0=\frac{-a_2 + \sqrt{a_2^2+4\ell_1(\ell_1-a_3)}}{2\ell_1}$,
and the trajectory impacts the negative $y$-axis.

There are no radial line trajectories in this case.

\section[The Post--Winternitz superintegrable systems]{The Post--Winternitz superintegrable systems}\label{section4}

These systems are especially interesting as a pair of classical and quantum superintegrable systems that do not admit separation of variables,
so they cannot be solved by traditional methods.
We treat only the classical case here and show that  its superintegrability allows us  to determine the trajectories exactly.
The quantum case can also be solved.

\subsection{The classical system}

The classical Post--Winternitz system is def\/ined by the Hamiltonian \cite{PW2011}
\begin{gather*}
{\cal H} = \frac12 \big(p_x^2+p_y^2\big)+\frac{cy}{x^{2/3}}.
\end{gather*}
The generating constants of the motion are of 3rd and 4th degree:
\begin{gather*}
 {\cal L}_1=3p_x^2p_y+2p_y^3+9cx^{1/3}p_x+6\frac{cyp_y}{x^{2/3}},\\
 {\cal L}_2=p_x^4+4cp_x^2\frac{y}{x^{2/3}}-12cx^{1/3}p_xp_y-2c^2\frac{(9x^2-2y^2)}{x^{4/3}}.
 \end{gather*}
Here $\{L_1,L_2\}=-188c^3$, so the constants of the motion generate a Heisenberg algebra.
Since there is no 2nd degree constant other than~$\cal H$, separation of the Hamilton--Jacobi equations is not
possible in any coordinate system, so separation of variables methods cannot be used to compute the trajectories. However, we can f\/ind a parametric description of the trajectories.
Choosing constants  ${\cal H}=E$, ${\cal L}_1=\ell_1$ and ${\cal L}_2=\ell_2$, we can verify the two identities
\begin{gather*}
 -p_y^3+9cx^{1/3}p_x-\ell_1+6Ep_y=0,\\
 \left(-\frac13p_y^4+4py^2E+4E^2-\ell_2-\frac43\ell_1 p_y\right)p_x^2-\frac29\big({-}p_y^3-\ell_1+6Ep_y\big)^2=0.
 \end{gather*}
From these results we can solve for $p_x$ and $x$, $y$ as functions of the parameter $p_y$:
\begin{gather}
  p_x = \frac{2\big({-}p_y^3-\ell_1+6Ep_y\big)}{\sqrt{-6p_y^4+72Ep_y^2-24\ell_1 p_y+72E^2-18\ell_2}},\nonumber\\
  x = -\frac{1}{5832c^3}\big({-}6p_y^4+72Ep_y^2-24\ell_1p_y+72E^2-18\ell_2\big)^{3/2},\nonumber\\
  y=\frac{1}{324c^3}\big({-}2\ell_1^2-36E^2p_y^2+p_y^6+8\ell_1p_y^3-18Ep_y^4+9\ell_2p_y^2+72E^3-18\ell_2 E\big),\label{xypx}
\end{gather}
and equations (\ref{xypx}) enable us to plot the trajectories. See Figs.~\ref{Fig34.pdf},~\ref{Fig35.pdf} and~\ref{Fig36.pdf}. There are no closed orbits.
The trajectories always impact the $y$-axis from the left or right, depending on the signs of the parameters.
\begin{figure}[t!]\centering
 \includegraphics[width=60mm]{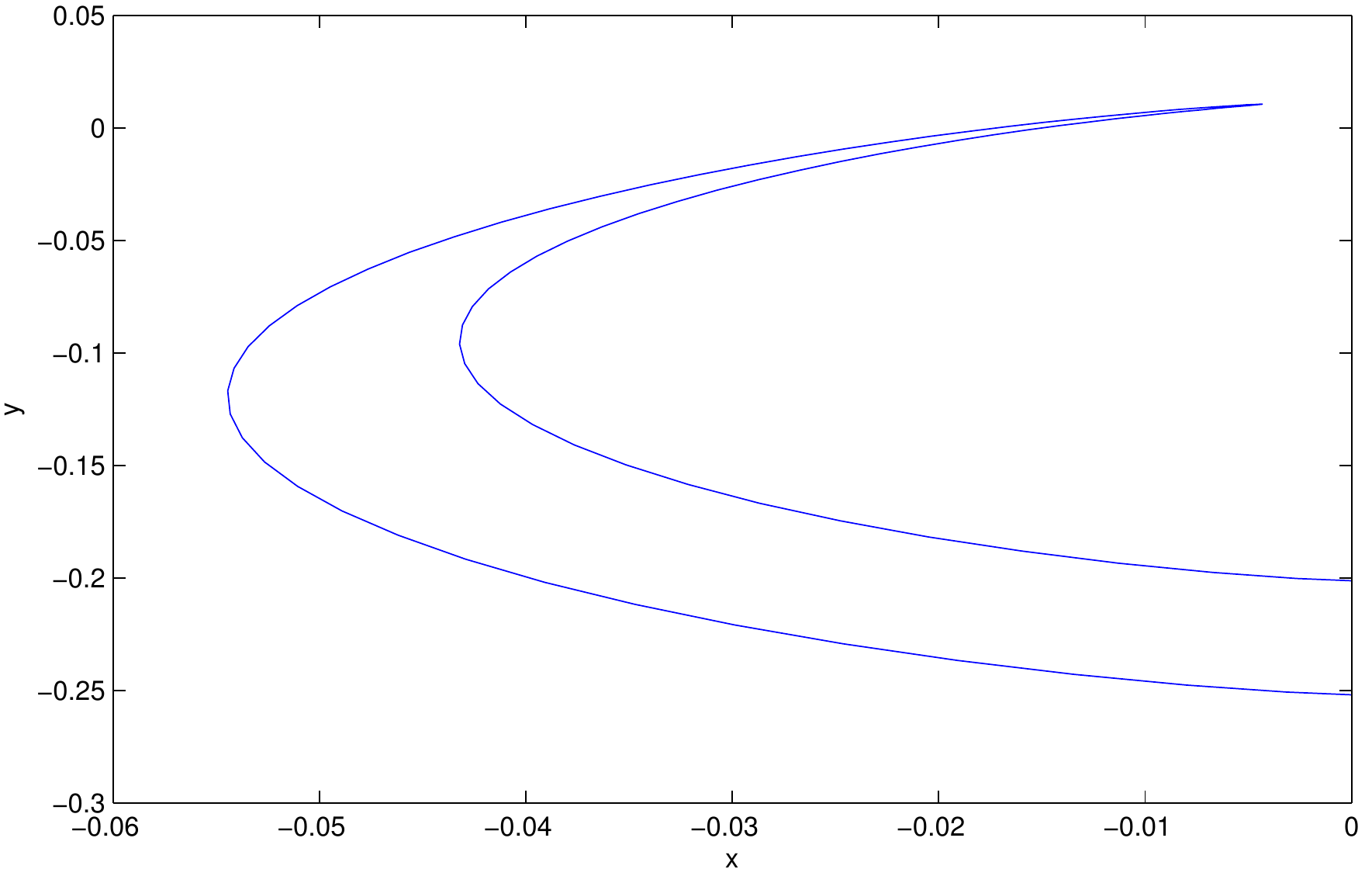}
\caption{$c=5$, $E=2$, $L_1=1$, $L_2=4$.}\label{Fig34.pdf}
\end{figure}
\begin{figure}[t!]\centering
 \includegraphics[width=60mm]{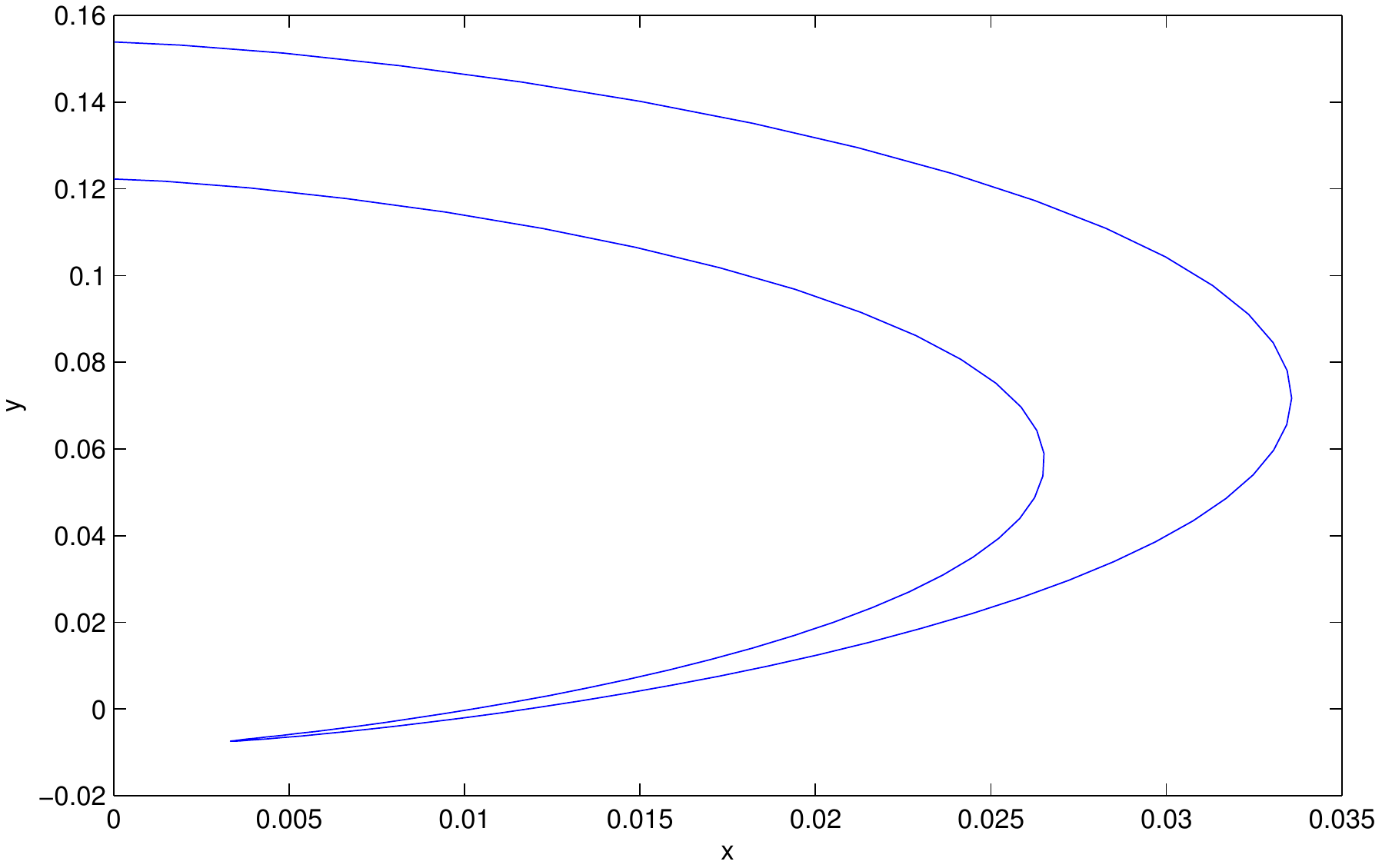}
\caption{$c=-12$, $E=4$, $L_1=3$, $L_2=6$.}\label{Fig35.pdf}
\end{figure}
\begin{figure}[t!]\centering
 \includegraphics[width=60mm]{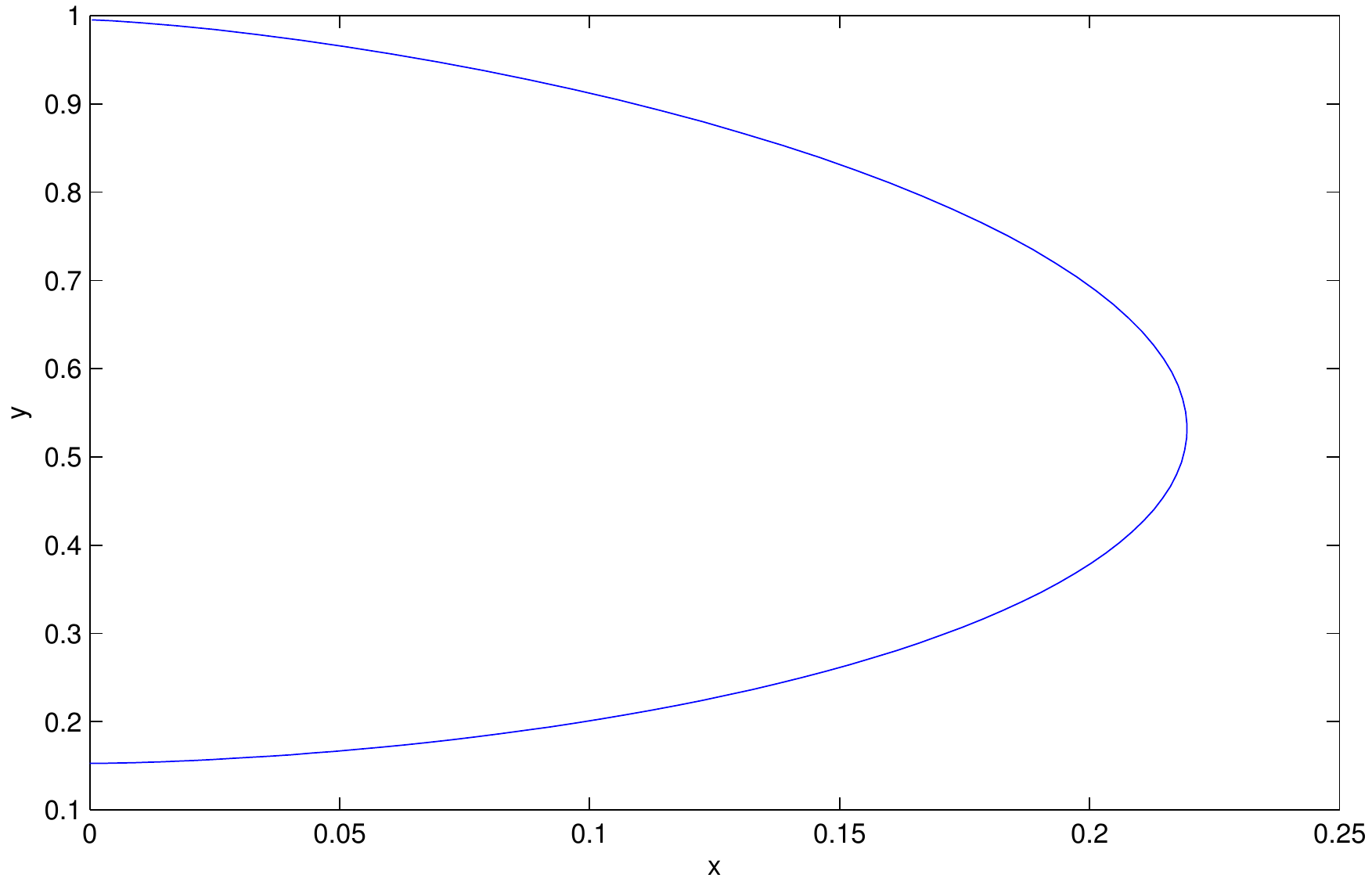}
\caption{$c=-1$, $E=0$, $L_1=5$, $L_2=5$.}\label{Fig36.pdf}
\end{figure}

\section{Classical elliptic superintegrable systems}\label{section5}

\subsection{Elliptic superintegrability }

Here we  extend the ideas in \cite{TTW1,TTW2}  to a system that separates in Jacobi elliptic coordina\-tes~\cite{DLMF}. We start with a simple example,
the Higgs oscillator $S3$~(\ref{Hamiltonian})~\cite{KKMP}, expressed in elliptic coordinates.
\begin{gather*}
 H= \frac{1}{\mbox{sn}^2(\alpha,k)-\mbox{sn}^2(\beta,k)}\left(p_\alpha^2-p_\beta^2 +\frac{A}{\mbox{sn}^2(\alpha,k)}-\frac{A}{\mbox{sn}^2(\beta,k)}\right).
 \end{gather*}
The Hamilton--Jacobi equation can be taken as $H=E$ for $p_\alpha=\frac{\partial W}{\partial \alpha}$, $p_\beta=\frac{\partial W}{\partial \beta}$.
The separation equations in the coordinates $\alpha$, $\beta$ are
\begin{gather}\label{papb}
p_\alpha^2-E\, \mbox{sn}^2(\alpha,k) +\frac{A}{\mbox{sn}^2(\alpha,k)}+\Lambda_1=0,\qquad p_\beta^2-E\,
\mbox{sn}^2(\beta,k) +\frac{A}{\mbox{sn}^2(\beta,k)}+\Lambda_1=0.
\end{gather}
Here, $H$ and $L_1= - p_\alpha^2+ \mbox{sn}^2(\alpha,k)H -\frac{A}{\mbox{sn}^2(\alpha,k)}$ are constants of the motion, and $L_1=\Lambda_1$ on a~coordinate hypersurface.
To f\/ind another constant of the motion via the action-angle variable construction we introduce  $M$ and $N$ which
are computed according to
\begin{gather*}
 M=\frac12\int\frac{d\alpha}{p_\alpha},\qquad N=\frac12\int\frac{d\beta}{p_\beta},
 \end{gather*}
where $p_\alpha,p_\beta$ are expressed in terms of coordinates $\alpha$, $\beta$ by (\ref{papb}).
Applying  action-angle theory, see, e.g.,~\cite{KKM10,KKM10b}, we see that  $M-N$ is a  constant of the motion. The integral  $M$ can be written as
\begin{gather*}
 M=\frac12\int\frac{\mbox{sn}(\alpha,k)  d\alpha}{\sqrt{E\, \mbox{sn}^4(\alpha,k)-\Lambda_1\, \mbox{sn}^2(\alpha,k)-A}},
 \end{gather*}
with  a similar integral for~$N$. This can be calculated as an elliptic
integral by choosing the new variable $\rho =\mbox{sn} ^2 ( \alpha ,k)$. We then have the relations
\begin{gather*}
 \frac{d\rho}{d\alpha}= 2\mbox{sn}( \alpha ,k)\mbox{cn}( \alpha ,k)\mbox{dn}( \alpha ,k),\qquad
 \mbox{cn}( \alpha ,k)=\sqrt{1- \rho },\qquad  \mbox{dn}( \alpha ,k)=\sqrt{1-k ^2 \rho }.
\end{gather*}
 The change of variables is then tantamount to the
computation of
\begin{gather*}
 M=\frac12\int\frac{\mbox{cn}( \alpha ,k)\mbox{dn}( \alpha ,k)\mbox{sn}( \alpha ,k)\, d\alpha}{\mbox{cn}( \alpha ,k)\mbox{dn}( \alpha ,k)
 \sqrt{E\, \mbox{sn}^4(\alpha,k)-\Lambda_1\, \mbox{sn}^2(\alpha,k)-A}}\\
 \hphantom{M}{}
 =\frac{1}{4k\sqrt{E}}\int\frac{d\rho}{\sqrt{(\rho-a)(\rho-b)(\rho-c)(\rho-d)}},
\end{gather*}
 where $a=1$, $b=\frac{ 1 }{k^2}$ and $c,d$ are roots of  $\rho^2-\frac{\Lambda_1}{E }\rho-\frac{A}{E}$.
A~similar calculation applies to the
substitution $\mu =\mbox{sn}^ 2 (\beta ,k)$ with exactly the same $a$, $b$, $c $ and $d$. A convenient choice of integral is
 \begin{gather*}
  \int\frac{dt}{\sqrt{(t-a)(t-b)(t-c)(t-d)}}\\
  \qquad{} =\frac{2}{\sqrt{(a-c)(b-d)}}\mbox{sn}^{-1}\left(\sqrt{\frac{(b-d)(t-a)}{(a-d)(t-b)}},
 \sqrt{\frac{(b-c)(a-d)}{(a-c)(b-d)}}\right).
\end{gather*}
This integral can be used for $t= \rho$ or $\mu$. If we  form $\mbox{sn}(2k \sqrt{E(a-c)(b-d)}(M-N),\kappa)$ for $\kappa=\sqrt{(b-c)(a-d)/(a-c)/(b-d)}$
and use the
addition formula~\cite{DLMF}
\begin{gather*}
 \mbox{sn}(u+v, \kappa )=\frac{ \mbox{sn}(u,\kappa)\mbox{cn}(v,\kappa)\mbox{dn} (v,\kappa)+\mbox{sn}(v,\kappa) \mbox{cn}(u,\kappa)\mbox{dn}(u,\kappa)}{1-\kappa^2\mbox{sn}^2(u,\kappa)
 sn^2(v,\kappa)},
\end{gather*}
  we obtain a constant of the motion. In terms of the coordinates it is helpful to write
\begin{gather*}
 p _\alpha =\sqrt{E}\sqrt{\frac{( \rho -c)( \rho -d)}{ \rho} }, \qquad p _\beta =\sqrt{E}\sqrt{\frac{( \mu -c)( \mu -d)}{ \mu} }.
 \end{gather*}
Then  for $Z=\sqrt{\frac{ b-d}{a-c}} \frac{1}{a-d}$  we have
\begin{gather*}
 \mbox{sn}(u, \kappa)\mbox{cn}(v, \kappa )\mbox{dn}(v, \kappa )=Z\sqrt{\frac{\rho-a}{\rho-b}}\frac{b-a}{\mu-b}\sqrt{(\mu-c)(\mu-d)}
 \end{gather*}
with a similar expression for $\mbox{sn}( v, \kappa )\mbox{cn}( u , \kappa )\mbox{dn}( u, \kappa )$. We also note that
\begin{gather*}
 \kappa^2\mbox{sn}(\rho,\kappa)^2\mbox{sn}(\mu,\kappa)^2=\frac{(\rho-a)(\mu-a)(b-c)(b-d)}{(\rho-b)(\mu-b)(a-c)(a-d)},
 \end{gather*}
as well as
\begin{gather*}
  \sqrt{(\mu-c)(\mu-d)}=2 \mu\sqrt{(1-\mu)(1-k^2 \mu)}\frac{p_\mu}{\sqrt{E}},\\
 \sqrt{(\rho-c)(\rho-d)}=2 \rho\sqrt{(1-\rho)(1-k^2 \rho)}\frac{p_\rho}{\sqrt{E}},\\
  \mbox{sn}(u, \kappa )=\sqrt{ \frac{(\rho-1)(d\, k^2-1)} {(d-1)(\rho\, k^2-1)}},\qquad \mbox{cn}(u, \kappa )=\sqrt{ \frac{(\rho-d)( k^2-1)} {(d-1)(1-\rho\, k^2)}},\\
  \mbox{dn}(u, \kappa )=\sqrt{ \frac{(\rho-c)( k^2-1)} {(c-1)(1-\rho\, k^2)}},\qquad \kappa=\sqrt{\frac{(d-1)(c\, k^2-1)}{(c-1)(d\, k^2-1)}},
\end{gather*}
with similar expressions for $\mbox{sn}(v, \kappa )$, $\mbox{cn}(v, \kappa )$ and $\mbox{dn}(v, \kappa )$ with $\rho$ replaced by~$\mu$.  From this we
deduce that
$\frac{\sqrt{E}}{Z}\mbox{sn}(2k\sqrt{E(a-c)(b-d)}(M-N),\kappa)=\frac{1}{X}$
where $X$ is a constant of the motion, linear in the momenta,  viz
$X= \frac{\sqrt{ (1-k^2\rho) (1-k^2\mu)(1-\rho)(1-\mu)}}{(\rho-\mu)}(\rho p_\rho-\mu p_\mu)$,
 in involution with the classical Hamiltonian
\begin{gather*}
 H=\frac{1}{\rho-\mu}\big( \rho(1-\rho)\big(1-k^2\rho\big)p_\rho^2-\mu(1-\mu)\big(1-k^2\mu\big)p_\mu^2\big)-\frac{A}{\rho\mu}.
 \end{gather*}
This is just the f\/irst degree symmetry that we expect for this Hamiltonian. In this case we choose the
variables $\rho$ and $\mu$ to obtain the constant.  In particular,
$cd=- A/E$, $c+d=-\Lambda_1/E$.

This seems to be  just a complicated way of deriving a straightforward result. However, we can use the TTW trick~\cite{TTW1,TTW2} by looking at the Hamiltonian
\begin{gather*}
 H= \frac{1}{\mbox{sn}^2(p\alpha,k)-\mbox{sn}^2(q\beta,k)}\left(p_\alpha^2-p_\beta^2 +\frac{A}{\mbox{sn}^2(p\alpha,k)}-\frac{A}{\mbox{sn}^2(q\beta,k)}\right),
\end{gather*}
where $p$ and $q$ are arbitrary nonzero numbers.
The separation equations are
\begin{gather}  p_\alpha^2-E\, \mbox{sn}^2(p\alpha,k) +\frac{A}{\mbox{sn}^2(p\alpha,k)}+\Lambda_1=0,\nonumber\\ p_\beta^2-E\,
\mbox{sn}^2(q\beta,k) +\frac{A}{\mbox{sn}^2(q\beta,k)}+\Lambda_1=0.\label{papb1}
\end{gather}
Here, $H$ and $L_1=  p_\alpha^2- \mbox{sn}^2(p\alpha,k)H +\frac{A}{\mbox{sn}^2(p\alpha,k)}$ are constants of the motion.
To f\/ind another constant of the motion via the action-angle variable construction we introduce~$M$ and~$N$ which
are computed according to
$ M=\frac12\int\frac{d\alpha}{p_\alpha}$, $N=\frac12\int\frac{d\beta}{p_\beta}$,
where $p_\alpha$, $p_\beta$ are expressed in terms of coordinates~$\alpha$,~$\beta$ by~(\ref{papb1}).
Again $M-N$ is a  constant of the motion. The integral  $M$ can be written as
\begin{gather*}
 M=\frac12\int\frac{\mbox{sn}(p\alpha,k)\, d\alpha}{\sqrt{E\, \mbox{sn}^4(p\alpha,k)-\Lambda_1\, \mbox{sn}^2(p\alpha,k)-A}},
 \end{gather*}
with  a similar integral for $ N$.

This can be calculated as an elliptic
integral by choosing the new variable $\rho =\mbox{sn} ^2 ( p\alpha ,k)$. We then have the relations
 \begin{gather*}
  \frac{d\rho}{d\alpha}= 2p\ \mbox{sn}( p\alpha ,k)\mbox{cn}( p\alpha ,k)\mbox{dn}( p\alpha ,k),\\
 \mbox{cn}( p\alpha ,k)=\sqrt{1- \rho },\qquad  \mbox{dn}( p\alpha ,k)=\sqrt{1-k ^2 \rho }.
\end{gather*}
 The change of variables is then tantamount to the
computation of
\begin{gather*}
 M=\frac{1}{2p}\int\frac{\mbox{cn}( {\hat \alpha} ,k)\mbox{dn}( {\hat \alpha} ,k)\mbox{sn}( {\hat \alpha} ,k)\, d{\hat \alpha}}{\mbox{cn}({\hat  \alpha} ,k)\mbox{dn}(
{\hat \alpha} ,k)
 \sqrt{E\ \mbox{sn}^4({\hat \alpha},k)-\Lambda_1\, \mbox{sn}^2({\hat \alpha},k)-A}}\\
 \hphantom{M}
=\frac{1}{4k\sqrt{E}p}\int\frac{d\rho}{\sqrt{(\rho-a)(\rho-b)(\rho-c)(\rho-d)}},
\end{gather*}
where ${\hat \alpha}=p\alpha$, $a=1$, $b=\frac{1}{k^2}$ and $c,d$ are roots of  $\rho^2-\frac{\Lambda_1}{E }\rho-\frac{A}{E}$.
A similar calculation applies to the
substitution $\mu =\mbox{sn}^ 2 (q\beta ,k)$ with exactly the same~$a$,~$b$,~$c $ and~$d$. A convenient choice of integral is
 \begin{gather*}
 \int\frac{dt}{\sqrt{(t-a)(t-b)(t-c)(t-d)}}\\
 \qquad {} =\frac{2}{\sqrt{(a-c)(b-d)}}\mbox{sn}^{-1}\left(\sqrt{\frac{(b-d)(t-a)}{(a-d)(t-b)}},
 \sqrt{\frac{(b-c)(a-d)}{(a-c)(b-d)}}\right).
\end{gather*}
This integral can be used for $t= \rho$ or $\mu$. If we now form  $\mbox{sn}(2k \sqrt{E(a-c)(b-d)} pq(M-N))$ and use the
addition formula
\begin{gather*} \mbox{sn}(qu+pv, \kappa )=\frac{ \mbox{sn}(qu,\kappa)\mbox{cn}(pv,\kappa)\mbox{dn} (pv,\kappa)+\mbox{sn}(pv,\kappa) \mbox{cn}(qu,\kappa)\mbox{dn}(qu,\kappa)}
{1-\kappa^2\mbox{sn}^2(qu,\kappa)
 \mbox{sn}^2(pv,\kappa)}
\end{gather*}
we obtain a constant of the motion.  Thus computation yields a rational constant of the motion whenever~$p/q$ is a rational number.

Let us restrict ourselves to $p=1$, $q=2$. Then $u= \alpha/2$ and
$v= \beta$ where
\begin{gather*}
 p _\alpha =\sqrt{E}\sqrt{\frac{( \rho -c)( \rho -d)}{ \rho} }, \qquad p _\beta =\sqrt{E}\sqrt{\frac{( \mu -c)( \mu -d)}{ \mu} },\qquad
Y=\sqrt{\frac{ b-d}{a-d}}.
\end{gather*}
We form
\begin{gather*}
 T_1= \mbox{sn}(2u, \kappa)\mbox{cn}(v, \kappa )\mbox{dn}(v, \kappa )\\
 \hphantom{T_1}{}
 =-2Y\frac{\sqrt{(\rho-a)(\rho-b)(\rho-c)(\rho-d)(\mu-c)(\mu-d)}}
{(\mu-b)(\rho^2(a+b-c-d)+2\rho(cd-ab)-cd(a+b)+ab(c+d))}, \\
 T_2= \mbox{sn}(u, \kappa)\mbox{cn}(2v, \kappa )\mbox{dn}(2v, \kappa )\\
 \hphantom{T_2}{}
 =-Y\sqrt{\frac{\mu-a}{\mu-b}}
 \left((c+b-a-d)\rho^2+2(ad-bc)\rho+cb(a+d)-ad(b+c)\right)\\
\hphantom{T_2=}{} \times \left((d+b-a-c)\rho^2+2(ac-bd)\rho+cb(d-a)+ad(b-c)\right) \\
\hphantom{T_2=}{} \times
 \left((a+b-c-d)\rho^2+2(cd-ab)\rho+ab(c+d)-cd(a+b)\right)^{-2}.
\end{gather*}
We also note that
 \begin{gather*}
 T_3= \kappa^2\mbox{sn}(2u,\kappa)^2\mbox{sn}(\mu,\kappa)^2\\
 \hphantom{T_3}{}
 =4\frac{(\mu-a)(\rho-a)(\rho-b)(\rho-c)(\rho-d)(b-d)(b-c)}
{(\mu-b)\left((a+b-c-d)\rho^2+2(cd-ab)\rho+ab(c+d)-cd(a+b)\right)^2},
\\
 \sqrt{( \mu -c)( \mu -d)}=4 \mu \sqrt{(1- \mu )(1-k ^2 \mu)}\frac{p_\mu}{\sqrt{E}},\\
\sqrt{( \rho -c)( \rho -d)}=2\rho \sqrt{(1- \rho )(1-k ^2 \rho)}\frac{p_\rho}{\sqrt{E}},
\end{gather*}
and $\rho =\mbox{sn}^ 2 ( \alpha,k)$ and $\mu =\mbox{sn} ^2 (2 \beta ,k)$.
From this we deduce that
$\frac{\sqrt{E}}{Y}\mbox{sn}(4k\sqrt{E(a-c)(b-d)}(M-N),\kappa)=P$,
where $P$ is a constant of the motion rational in the momenta. We can  calculate the
components of this constant as follows: $T _1=n_1/d_1$, where
\begin{gather*}
  n_ 1 =-16(1- \rho )\big(1-k^ 2 \rho \big) \rho\mu  {k'}^2 \sqrt{1- \mu }p _\rho p _\mu,\\
  d_1 =(k ^2 \mu-1)\big(\big(k ^2 \Lambda_ 1 -E(1+k 2 )\big) \rho^2 +2\big(Ak^ 2 +E\big) \rho -\Lambda_1 -A(k 2 +1)\big),
\end{gather*}
$T _2=n_2/d_2$, where
\begin{gather*}
 n_ 2 = \sqrt{1- \mu}\big[{-}\big(E^ 2 {k'}^ 4 +k^ 4 \big(\Lambda_ 1 ^2 +4AE\big)\big) \rho^4 +
 \big({-}k ^4 \big({}-8AE+2\Lambda_1 E-2\Lambda_ 1 ^2 \big)\\
 \hphantom{n_ 2 =}{}
  +k ^2 \big({-}8AE-4\Lambda_ 1 E-\Lambda^ 2_ 1 \big)+2\Lambda_ 1 E\big) \rho^ 3 +6\big(2k ^2 AE+AE+k ^2 (\Lambda^ 2_ 1 +AE)\big) \rho^ 2 \\
 \hphantom{n_ 2 =}{}
  +\big({-}2k ^4 \Lambda_ 1 A   +2k^ 2 \big({-}\Lambda_ 1^ 2 +2\Lambda_1 A-4AE\big)-2\Lambda_1^2 -2\Lambda_1 A-8AE\big) \rho\\
 \hphantom{n_ 2 =}{}
  +\big({-}k ^4 A 2 +4AE+\big(2k ^2 -1\big)A^ 2 -\Lambda^ 2 _1 \big)\big],\\
d_ 2 = \sqrt{k ^2 \mu -1}\big[\big(k^ 2 \Lambda_ 1 -E\big(1+k^ 2 \big)\big) \rho^ 2 +2\big(Ak^ 2 +E\big) \rho -\Lambda_ 1 -A\big(k ^2 +1\big)\big]^ 2,
\end{gather*}
$T_3=n_3/d_3$, where
\begin{gather*}
 n_ 3 =4( \rho -1)\big( \rho k ^2 -1\big)( \mu -1)\big({-}E \rho^ 2 +\Lambda_ 1 \rho +A\big)\big({-}E+Ak^ 4 +\Lambda_1 k^ 2 \big),\\
 d _3 =(k ^2 \mu -1)\big[\big(k^ 2 \Lambda_ 1 -E\big(1+k^ 2 \big)\big) \rho^ 2 +2\big(Ak ^2 +E\big) \rho -\Lambda_ 1 -A\big(k ^2 +1\big)\big]^ 2.
\end{gather*}
We could, of course, use the expressions above given in terms of $c+d $ and $cd$.

If we set
$x={\rm sn}^2 ( \alpha,k)=\rho$ and $y=\mbox{sn} ^2 (2 \beta,k)=\mu$ then $H$ and $L_1$ become
\begin{gather}\label{Hamelliptic1}
H=\frac{1}{x-y}\big( x(1-x)\big(1-k ^2 x\big)p_ x ^2 -4y(1-y)\big(1-k^ 2 y\big)p^ 2 _y \big)-\frac{A}{xy},\\
L_1 =\frac{xy}{x-y}\big( (1-x)\big(1-k^ 2 x\big)p ^2 _x -4(1-y)\big(1-yk^ 2 \big)p ^2 _y \big)-A\left(\frac{1}{x}+\frac{1}{y}\right).\nonumber
\end{gather}

For $ A=0$   the extra constant of the motion is
$G= N/D$ where
\begin{gather*}
 N=\sqrt{y} \big(\big({-}yxk ^2 -x^ 2 k-x ^2 +2x-y\big)p^ 2_ x -4y(y-1)\big({-}1+k^ 2 y\big)p ^2 _y \\
 \hphantom{N=}{} +4x(y-1)\big({-}1+k ^2 y\big)p_ x p _y \big),\\
 D=\sqrt{y-1} \big(\big(x^ 2 k^ 2 +x^2 yk^ 2 -x ^2 -2xyk ^2 +y\big)p^ 2_ x +4y(y-1)\big({-}1+k^ 2 y\big)p^ 2_ y\\
 \hphantom{D=}{} -4y\big({-}1+k ^2 y\big)(x-1)p_ x p _y \big).
\end{gather*}
Setting $ R=\{L_ 1 ,G\}$  we obtain the polynomial structure relations
\begin{gather*}
 R ^2 -4\big(1-{k'}^ 2 G^ 2 \big)\big(L_ 1 +G ^2 H-L_ 1 G ^2 \big)=0,\\
 \{G,R\}=2\big(G ^2 -1\big)\big(G^ 2 {k'}^ 2 -1\big),\\
  \{L_1 ,R\}=-4G\big(\big(2-k ^2 \big)L_ 1 +2{k'}^ 2 HG^ 2 -2{k'}^ 2 L_ 1 G^ 2 -H\big).
\end{gather*}
Turning to the case when $A\ne 0$ we have $G=N/D$,
where
\begin{gather*}
  N=\sqrt{y-1}\big[x^ 2 y\big(x ^2 yk^ 2 +x^ 2 k^ 2 -x^ 2 -2xyk ^2 +y\big)p^ 2 _x +4x ^2 y^ 2 (y-1)\big({-}1+k ^2 y\big)p^ 2_ y \\
  \hphantom{N=}{}-
4x ^2 y^ 2 \big({-}1+k^ 2 y\big)p_ x p_ y +A(x-y)^ 2 \big],\\
  D=\sqrt{-1+k ^2y}\big[x ^2 y\big({-}2xy+y-k^ 2 x ^2 +x^ 2 yk^ 2 +x^ 2 \big)p ^2 _x +4x^ 2 y ^2 (y-1)\big({-}1+k ^2 y\big)p^ 2 _y\\
  \hphantom{D=}{}
-4x^ 2 y^ 2 (y-1)\big({-}1+k ^2 x\big)p_ x p_ y +A(x-y)^2\big].
\end{gather*}
The structure  relations are
\begin{gather}
 \{G,R\}=-2\big(G^ 2 -1\big)\big({k'}^ 2 G^ 2 -1\big),\nonumber\\
  \{L_ 1 ,R\}=-4G\big[2L_ 1 k ^2 G ^2 +2Ak^ 4 G^ 2 -2HG^ 2 -L_1 -L_ 1 k^ 2 -2Ak ^2 +2H\big],\label{structureelliptic2}\\
  R^ 2 +4\big[Ak^ 4 G^ 2 -L_1 k^ 2 G^ 2 +L_1 G^ 4 -2Ak^ 2 G+A-L_ 1 G ^2 +2HG ^2 +L_ 1 -H-H ^4 \big]=0.\nonumber
\end{gather}
All of the examples that we have constructed using the addition theorem for elliptic functions obey polynomial structure equations,
but we have no proof that this is true in general.

 \subsection{Example}
\begin{gather*}
 H= \frac{1}{y^2-x^2}\big(\big(1-x^2\big)p_x^2-4\big(1-y^2\big)p_y^2\big)-\frac{A}{(1-x^2)(1-y^2)},\\
 L_1 = \frac{(1-x^2)(1-y^2)}{(y^2-x^2)}\big(p_x^2-4p_y^2\big)
+\frac{A(x^2+y^2-2)}{(1-x^2)(1-y^2)},\\
 L_2 = \frac{(yx^4-2x^2y+y^3)p_x^2+(-4x^2y+4x^2
y^3)p_y^2+(-4x^3y^2+4x^3)p_xp_y-A\frac{x^2y(x^2-y^2)^2}{(1-y^2)(1-x^2)^2}}{(x^4-2y^2x^2+y^2)p_x^2+(4x^2-4x^2y^2)p_y^2+(-4xy+4xy^3
)p_xp_y+A\frac{x^2(x^2-y^2)
^2}{(1-y^2)(1-x^2)^2}}.
\end{gather*}
This is the special case of (\ref{Hamelliptic1}), (\ref{structureelliptic2}) where we have replaced $x,y$ by $1-x^2,1-y^2$, respectively, set $k=0$ and rescaled.
Note that
$ H=4p_y^2+\frac{L_1}{(1-y^2)}+\frac{A}{(1-y^2)^2}$,
and
$ \frac{(1-x^2)}{(y^2-x^2)}\left(p_x^2-4p_y^2\right)=\frac{L_1}{(1-y^2)} -\frac{A(x^2+y^2-2)}{(1-x^2)(1-y^2)^2}$.
We have
 $\{H,L_1\}=\{H,L_2\}=0,$
and with $R=\{L_1,L_2\}$, the polynomial structure relations are
\begin{gather*}
 \{L_1,R\}-16L_2L_1+32HL_2-32L_2^3H=0,\qquad
\{L_2,R\}-8\big(1-L_2^2\big)=0,\\
 R^2+16L_1-16L_2^2L_1+32H L_2^2-16L_2^4H-16H+16A=0.
 \end{gather*}
Using the three constants of the motion to eliminate the momenta terms we f\/ind
\begin{gather*}
 p_x^2 = \frac{H(x^2-1)^2+L_1(x^2-1)-A}{(1-x^2)^2},\qquad  p_y^2 = \frac{H(y^2-1)^2+L_1(y^2-1)-A}{4(1-y^2)^2},
 \end{gather*}
and  the trajectories must satisfy the implicit equation
\begin{gather*} 
0=\left(Hx^2(A+L_1-L_1x^2-Hx^4+2Hx^2-H)y^2+\frac14\big({-}H+L_1+Hx^4+A\big)^2\right)L_2^2 \\
\hphantom{0=}{}
 +\left(3Hx^4L_1-\frac12 H^2x^8-\frac12 A^2+AH -3H^2x^4+L_1x^2A+2H^2x^6+3Hx^4A\right.\\
\left. \hphantom{0=}{}
-AL_1-L_1Hx^6+L_1H+2H^2x^2
-\frac12L_1^2+L_1^2x^2
-\frac12H^2-3L_1x^2H-2Hx^2A\right)yL_2\\
\hphantom{0=}{}+\frac14\big({-}H+L_1+Hx^4+A\big)^2y^2\\
\hphantom{0=}{}-x^2(A+L_1-H)\big(A+L_1-L_1x^2-Hx^4+2Hx^2-H\big).
\end{gather*}
The left-hand side of this equation is quadratic in~$y$, so we can solve for~$y$ as a function of~$x$.  We obtain the solutions
$ y=\frac{N\pm 2\sqrt{S}}{D}$,
where
\begin{gather*}
N= L_2(-1+x)^4(x+1)^4H^2-2L_2\big({-}2Ax^2-3L_1x^2-L_1x^6+3L_1x^4+3Ax^4+A+L_1\big)H\\
\hphantom{N=}{} +L_2(A+L_1)\big(A+L_1-2L_1x^2\big),\\
S= -x^2\big(A+L_1-L_1x^2-Hx^4+2Hx^2-H\big)\big({-}H+L_1+Hx^4+A\big)^2\frac{R^2}{16},\\
D= -(-1+x)^2(x+1)^2\big(x^2+2L_2x+1\big)\big({-}x^2+2L_2x-1\big)H^2\\
\hphantom{D=}{} + \big(4L_2^2Ax^2+2L_1x^4-2L_1+4L_2^2L_1x^2-2A-4L_2^2L_1x^4+2Ax^4\big)H+(A+L_1)^2.
\end{gather*}
Let $(x_0,y_0)$ be a point on the trajectory such that $p_{x_0}=0$. A straightforward calculation gives
\begin{gather*}
x_0^2=\frac{2H-L_1\pm \sqrt{L_1^2+4AH}}{2H},\qquad y_0=-L_2,\\
p_{y_0}^2=-\frac{A-H-HL_2^4+2HL_2^2+L_1-L_1L_2^2}{4(L_2-1)^2(L_2+1)^2}.
\end{gather*}

{\bf Assumption: $A>0$}.
Note that
\begin{gather}\label{paragee1}
x_0^2=1+\left(\frac{-L_1\pm \sqrt{L_1^2+4AH}}{2H}\right).
\end{gather}
It follows that if we are in the region $x>1$ and $H>0$, then since $\sqrt{L_1^2+4AH}>|L_1|$, only the plus
sign is possible in (\ref{paragee1}). We conclude that the trajectory must be unbounded to the right.
Examples are Figs.~\ref{Fig37.pdf} and~\ref{Fig38.pdf}. (The dif\/ferent colors in the f\/igures correspond to dif\/ferent functions describing the trajectories.)
\begin{figure}[t!]\centering
 \includegraphics[width=60mm]{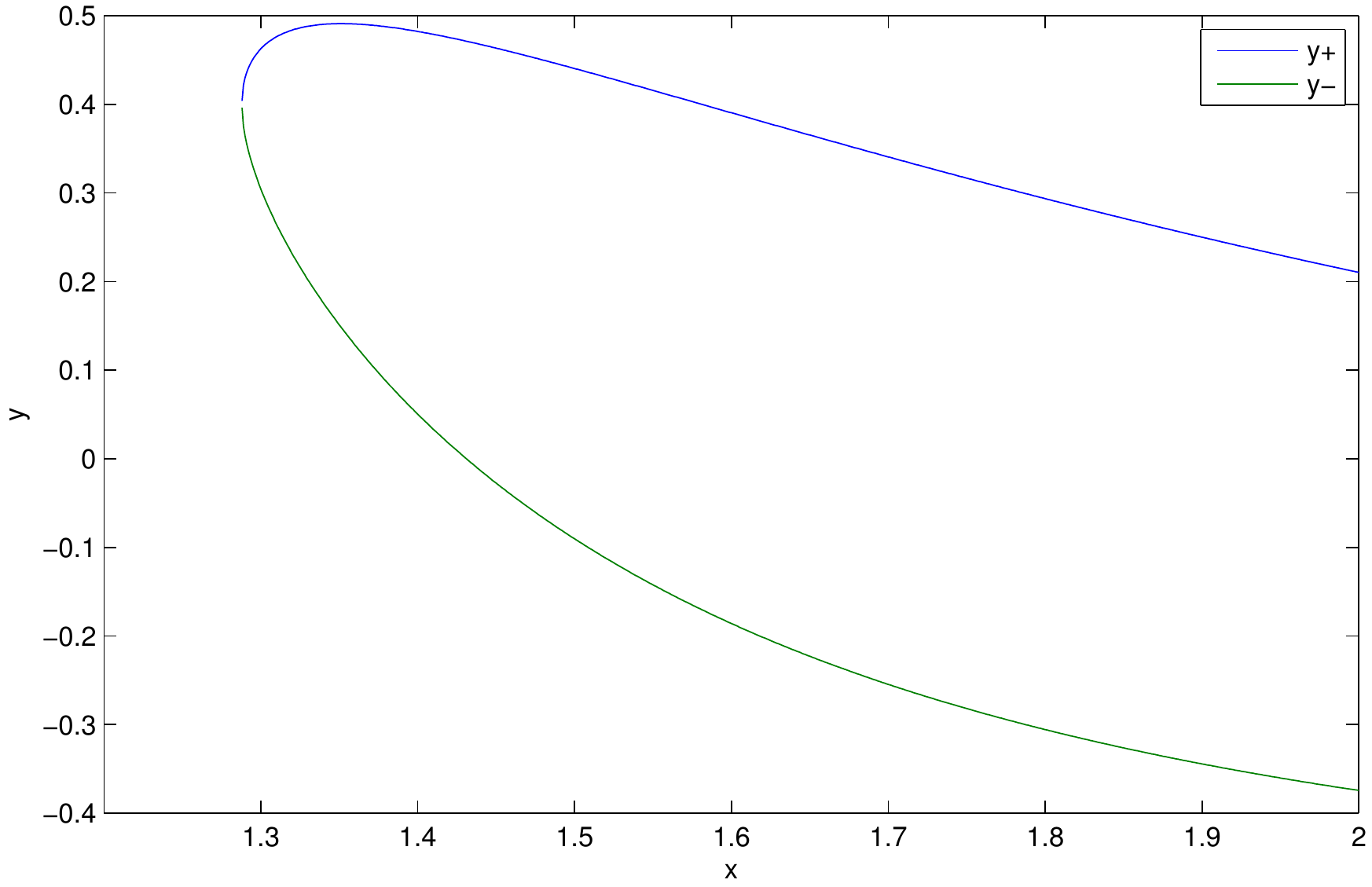}
\caption{Case: $A=0.5$, $L_1=0.1$, $L_2=-0.4$, $H=1$.}\label{Fig37.pdf}
\end{figure}
\begin{figure}[t!]\centering
 \includegraphics[width=60mm]{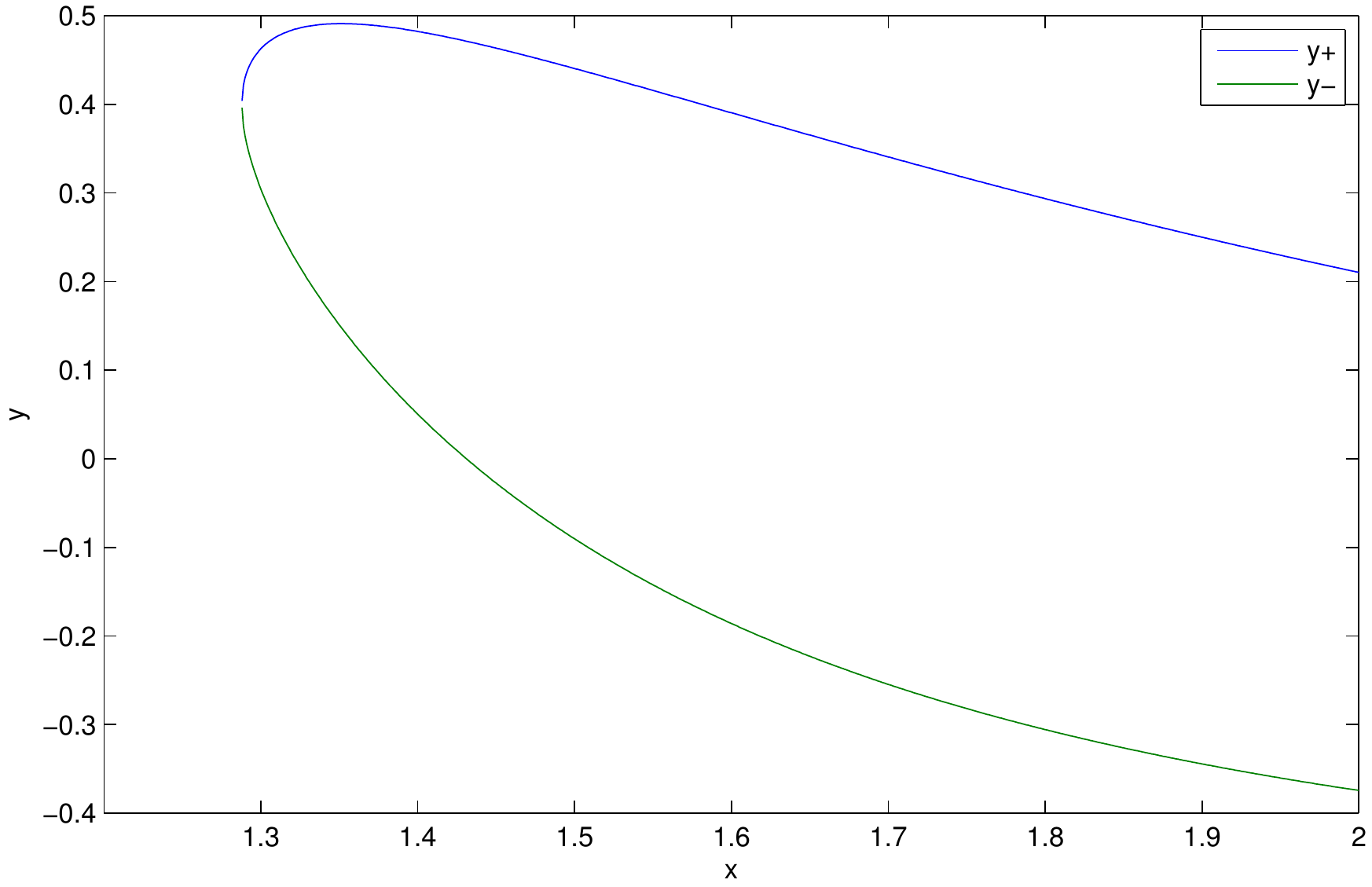}
\caption{Case: $A=1$, $L_1=0.2$, $L_2=0$, $H=2$.}\label{Fig38.pdf}
\end{figure}

{\it We investigate what happens if the trajectory intersects the lines $y=\pm x$.} If the intersection occurs for $y=x$ then we
f\/ind that $x=-L_2$ or $p_x=2p_y$, or both.
If the intersection occurs for $y=-x$ then we f\/ind that $x=L_2$, or $p_x=-2p_y$, or both. In both cases, we have
$ p_x^2=4p_y^2=\frac{H(L_2^2-1)^2+L_1(L_2^2-1)-A}{(L_2^2-1)^2}$.
Thus, in both cases if $|L_2|\ne 1$, the trajectory intersects the lines $y=\pm x$ at f\/inite momentum. However, the velocity is inf\/inite.
Indeed
${\dot x} = \{H,x\} = 2p_x(1-x^2)/(y^2-x^2)$, which blows up unless $ p_x=0$. This suggests that there are no bounded orbits in the
square $-1<x,y<1$. However we f\/ind some bounded orbits in the region $x>1$ for  negative energies.
Examples of these bounded orbits are Figs.~\ref{Fig39.pdf} and~\ref{Fig40.pdf}.
\begin{figure}[t!]\centering
 \includegraphics[width=60mm]{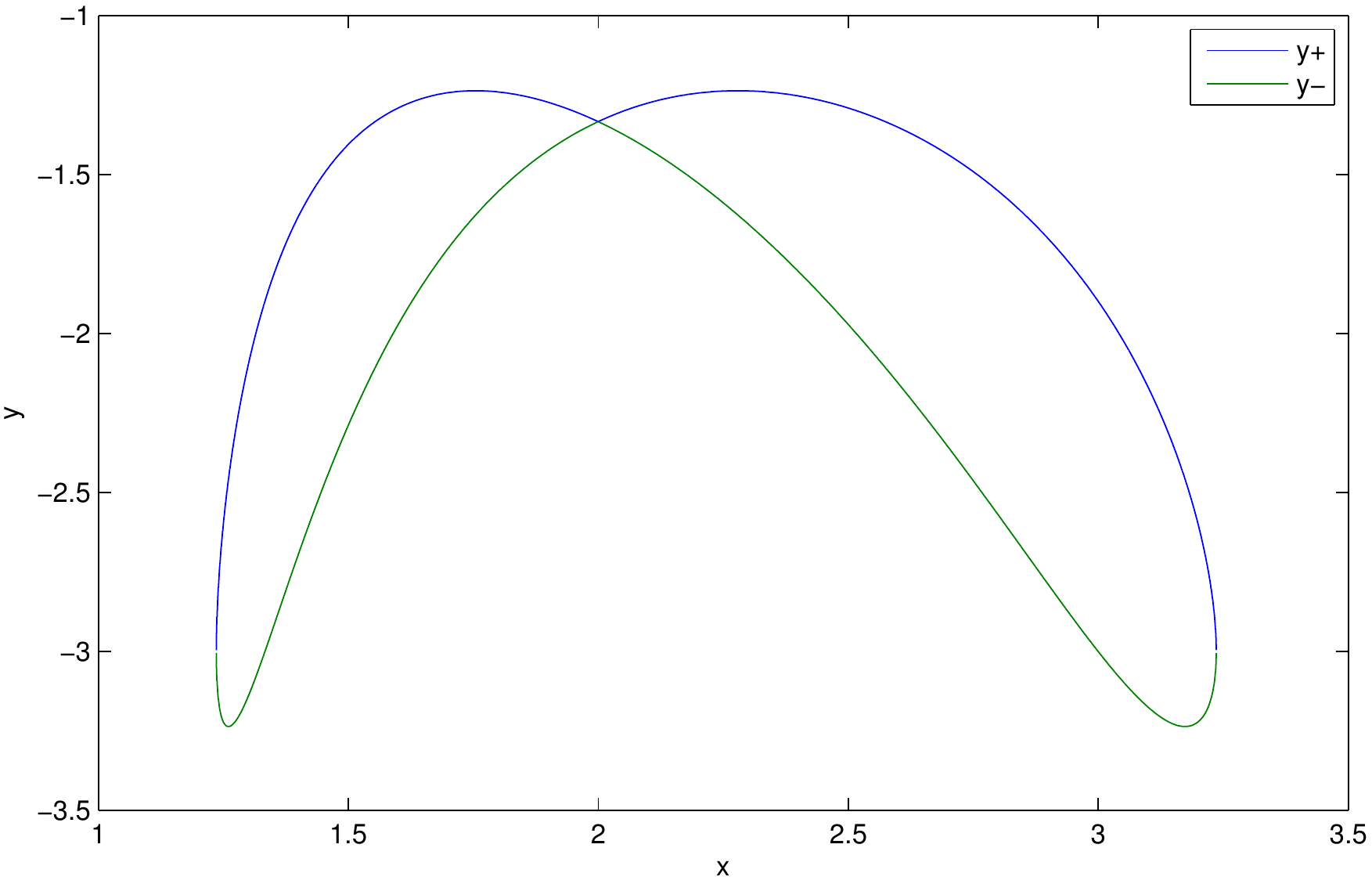}
\caption{Case: $A=5$, $L_1=10$, $L_2=3$, $H=-1$.}\label{Fig39.pdf}
\end{figure}
\begin{figure}[t!]\centering
\includegraphics[width=60mm]{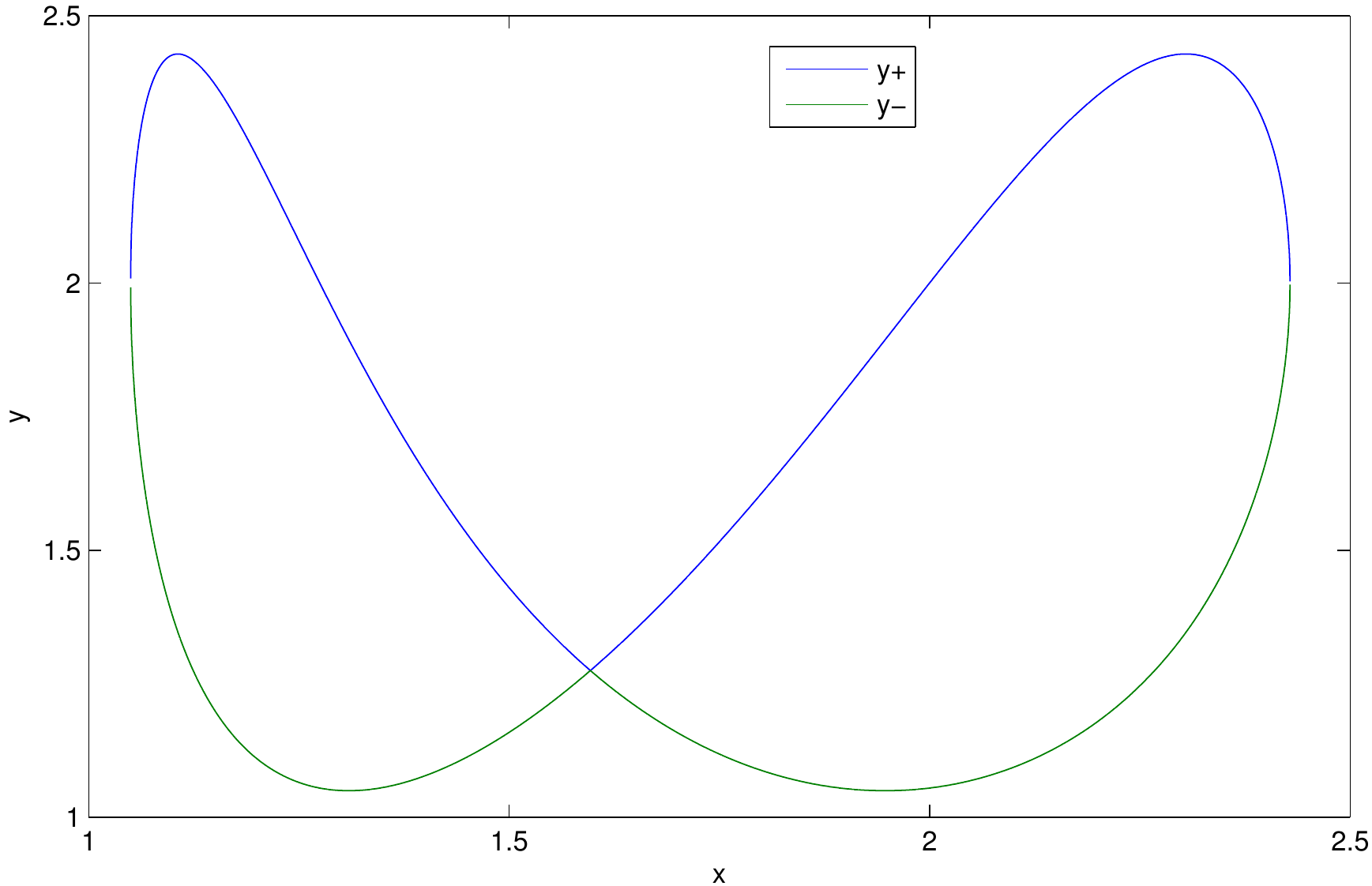}
\caption{Case: $A=1$, $L_1=10$, $L_2=-2$, $H=-2$.}\label{Fig40.pdf}
\end{figure}

\subsection{Another version of the example}
We set $y=iz$ where $z$ is real, multiply $H$ by $-1$ and $L_2$ by $-i$ to get the system
\begin{gather*}
 H=\frac{(1-x^2)p_x^2+4(1+z^2)p_z^2}{x^2+z^2}+\frac{A}{(1-x^2)(1+z^2)}, \\
  L_1=-\frac{(1+z^2)(1-x^2)(p_x^2+4p_z^2)}{x^2+z^2}+\frac{A(x^2-z^2-2)}{(1+z^2)(1-x^2)},\\
  L_2=\frac{zx^4p_x^2-4x^3p_zp_x z^2-4x^3 p_zp_x-2x^2zp_x^2+4x^2p_z^2z+4x^2p_z^2z^3-p_x^2z^3-\frac{ Ax^2z(z^2+x^2)^2}{(1+z^2)(1-x^2)^2}}
{x^4p_x^2-4p_z^2x^2-4 p_xzp_zx-4p_z^2x^2z^2+2p_x^2z^2x^2-p_x^2z^2-4xz^3p_zp_x+\frac{Ax^2(x^2+z^2)^2}{(1+z^2)(1-x^2)^2)}}.
\end{gather*}
 We will consider this system for $A>0$ and resticted to the vertical strip of points $(x,z)$ where $-1<x<1$. Then $H$ must be  $>0$. The motion is unbounded, there are no bounded orbits, and the origin~$(0,0)$ and the boundary lines $x=\pm 1$ are repulsive.

The polynomial structure relations are $R=\{L_1,L_2\}$ and
\begin{gather*}
 \{L_2,R\}+8\big(1+L_2^2\big)=0,\qquad   \{L_1,R\}-16L_2L_1-32HL_2-32L_2^3H=0,\\
  -R^2+16A+16L_1+16L_2^2L_1+32HL_2^2+16L_2^4H+16H=0.
\end{gather*}
 Expressing the momenta in terms of the coordinates we have
 \begin{gather*}
  p_x^2 = -\frac{(A+H+x^4H+L_1-x^2L_1-2Hx^2}{(1-x^2)^2}, \\
  p_z^2 = \frac14\frac{(2Hz^2+Hz^4+L_1z^2+A+H+L_1)}{(1+z^2)^2}.
 \end{gather*}
 The equations for the trajectories turn out to be
 $ y=\frac{N\pm 2\sqrt{S}}{D}$,
 where
\begin{gather*}
 N = L_2\big(1-x^2\big)^4H^2+2L_2\big(L_1+3Ax^4+A-L_1x^6+3L_1x^4-2Ax^2-3L_1x^2\big)H\\
 \hphantom{N=}{} +L_2(L_1+A)\big(L_1+A-2L_1x^2\big),\\
 S= -\frac{x^2}{16}\big(A+H+Hx^4+L_1-L_1x^2-2Hx^2\big)\big(A+H+L_1-Hx^4\big)^2R^2,\\
 D=\big(1-x^2\big)^2  \big(x^4+4 L_2^2 x^2+2 x^2+1\big) H^2\\
 \hphantom{D=}{} +\big({-}2 A x^4+2 L_1-4 L_2^2 L_1 x^4+2 A-2 L_1 x^4+4 L_2^2 L_1 x^2+4 L_2^2 A x^2\big) H+(L_1+A)^2.
\end{gather*}
 The points $(x_0,z_0, p_{x_0},p_{z_0})$ on the trajectory where $p_{x_0}=0$ satisfy
 \begin{gather*}
  z_0 = -L_2,\qquad  p_{z_0}^2 = \frac14 \frac{(A+H(1+L_2^2)^2+L_1+L_1L_2^2)}{(1+L_2^2)^2},\\
   x_0^2 = 1+\frac{L_1\pm \sqrt{L_1^2-4AH}}{2H}.
   \end{gather*}
 Requiring f\/irst that $R^2>0$, we see that $L_1<0$ is a necessary condition for trajectories (otherwise $S<0$ in the strip $-1<x<1$. If
 $|L_1|/2H>1$ then a further necessary condition for $S>0$ in a~strip is $L_1+H+A<0$. If $|L_1|/2H\le 1$ then a necessary condition is $L_1^2-4AH>0$.
Examples here are Figs.~\ref{Fig41.pdf} and~\ref{Fig42.pdf}. Note that here there are 4~dif\/ferent trajectories corresponding to a~single choice of constants of the motion.

\begin{figure}[t!]\centering
 \includegraphics[width=65mm]{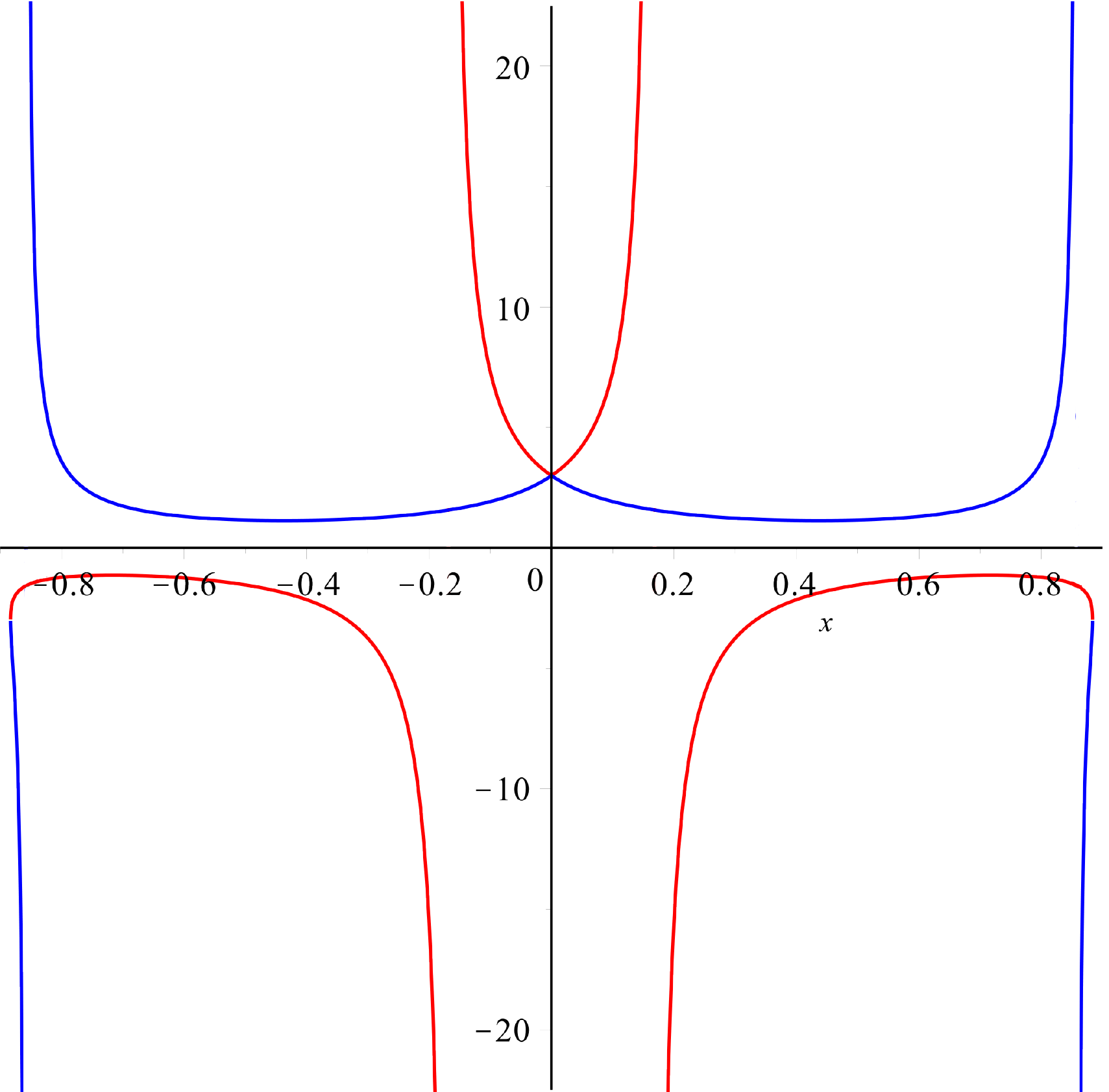}
\caption{Case: $A=1$, $L_1=-5$, $L_2=3$, $H=2$.}\label{Fig41.pdf}
\end{figure}
\begin{figure}[t!]\centering
 \includegraphics[width=65mm]{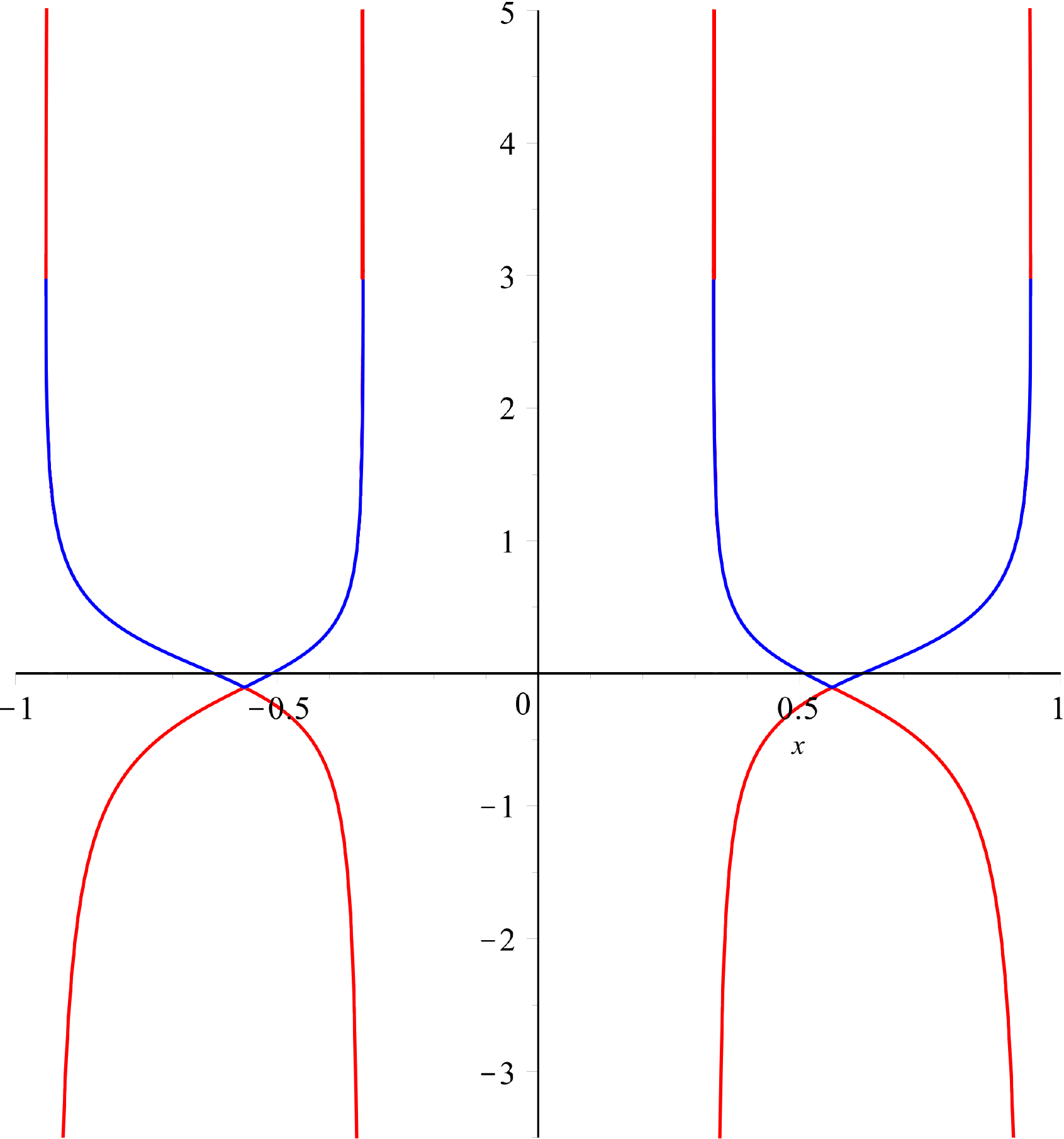}
\caption{Case: $A=1$, $L_1=-10$, $L_2=-3$, $H=10$.}\label{Fig42.pdf}
\end{figure}

{\bf Case when} $R^2=0$. When $L_1=-9$, $L_2=0$, $A=1$, $H=8$, we have $R^2=0$ such that $S=0$. The trajectory in this case is a straight line on the x-axis. The two boundary points where $p_x=0$ for $x>0$ are $0$ and $0.9354$.
The velocity of the particle is $\dot{x}=\{H,x\}=\frac{p_x(1-x^2)}{x^2+z^2}$. Substituting $z=0$ and the expression for $p_x$ in terms of $x$ into the equation gives,$\dot{x}=\sqrt{\frac{L_1+2H}{x^2}-H} $
Thus $\displaystyle \lim_{x \to 0} \dot{x}=\infty$, so the particle will go through the origin instead of bouncing back.

\appendix
\section{Review of some  basic concepts in Hamiltonian mechanics} \label{appendixA}

Hamiltonian mechanics is a reformulation of classical Newtonian mechanics alternate to  Lagrangian mechanics. In the Hamiltonian formalism a  physical system describing the motion of
a~particle at time~$t$ involves $n$ generalized coordinates~$q_j(t)$, and $n$ generalized momenta~$p_j(t)$. The phase space of  the  system is
described by points $(p_j, q_j)\in  {\mathbb R}^{2n}$. The generalized coordina\-tes~$q_j's$ and generalized momenta~$p_j's$ are derived from the Lagrangian formulation.
The dynamics of the system are given by
Hamilton's equations~\cite{Arnold, Goldstein}
\begin{gather}\label{hameqns}   \frac{d q_j}{dt}  =+\frac{\partial {\cal H}}{\partial p_j}, \qquad
\frac{dp_j}{dt}=-\frac{\partial {\cal H}}{\partial q_j},
 \end{gather}
where ${\cal H}={\cal H}(q_1,\dots, q_n, p_1,\dots, p_n, t)$ is the Hamiltonian, which in this paper corresponds to the total (time-independent) energy of the system:
$ {\cal H} = T + V$,
where $T$ and $V$ are kinetic and potential energy, respectively. Solutions of these  equations give the trajectories of the  system. Here~$T$ is a function of ${\bf q}$ alone while $V$ is a function of ${\bf q}$ alone. Explicitly,
\begin{gather}\label{relham} {\cal H} = \frac{1}{2m}\sum_{j,k}g^{jk}({\bf q}) p_jp_k+ V ({\bf q}),
 \end{gather}
where $g^{jk}$ is a contravariant metric tensor on some real or complex Riemannian  manifold.
That is $g^{-1}=\det(g^{jk})\ne 0$, $g^{jk}=g^{kj}$  and the metric on the manifold is given by $ds^2=\sum\limits_{j,k=1}^ng_{jk}dq^j  dq^k$, where $(g_{jk})$ is
the covariant metric tensor, the matrix inverse to $(g^{jk})$. Under a local  transformation $q'_j=f_j({\bf q})$
the contravariant tensor and  momenta transform according to
\begin{gather*}
(g')^{\ell h}=\sum_{j,k} \frac{\partial q'_\ell}{\partial q_j}\frac{\partial q'_h}{\partial q_k}g^{jk},\qquad
 p'_\ell =\sum_{j=1}^n \frac{\partial q_j}{\partial q'_\ell} p_j,
 \end{gather*}
so ${\cal H}$ is  coordinate independent. Here $m$ is a scaling parameter that can often be interpreted as the mass of the particle.
For the Hamiltonian~(\ref{relham}) the relation  between  momenta and the velocities is
$ p_j=m\sum\limits_{\ell=1}^ng_{j\ell}\dot{q}_\ell$,
so that
$T=\frac{1}{2m}\sum\limits_{j,k}g^{jk}({\bf q}) p_jp_k=\frac{m}{2}\sum\limits_{\ell,h=1}^n g_{\ell h}({\bf q})\dot{q}_\ell\dot{q}_h$.
Once the velocities are given, the momenta are scaled linearly in~$m$. In mechanics the exact value of~$m$ may be important, but for our purposes  can be scaled to any nonzero value using the above formulas. To make direct contact with mechanics we will often set~$m=1$;
for  structure calculations we will usually set~$m=1/2$.

 The {\it Poisson bracket} of two arbitrary functions ${\cal A}({\bf p}, {\bf q})$, ${\cal B}({\bf p}, {\bf  q})$  on the phase space is the function
 \begin{gather*}
 \{{\cal A},{\cal B}\}({\bf p},{\bf q})=\sum_{j=1}^n\left(\frac{\partial{\cal A}}{\partial p_j}
 \frac{\partial{\cal B}}{\partial q_j}
-\frac{\partial{\cal A}}{\partial q_j}\frac{\partial{\cal B}}{\partial p_j}\right).
\end{gather*}
The Poisson bracket obeys the following properties, for ${\cal A}$, ${\cal  B}$, ${\cal C}$
functions on the phase space and~$a$,~$b$ constants,
\begin{gather*}
 \text{anti-symmetry:}  \quad \{{\cal A},{\cal B}\}=-\{{\cal B},{\cal A}\}, \\
\text{Bilinearity:}  \quad \{{\cal A}, m{\cal B}+ n{\cal C}\}= m\{{\cal A}, {\cal B}\} + n\{{\cal A},{\cal C}\},  \\
\text{Jacobi identity:}  \quad \{{\cal A},\{ {\cal B},{\cal C}\}\}  + \{{\cal B}, \{{\cal C}, {\cal A}\}\}+\{{\cal C}, \{{\cal A},{\cal B}\}\}
 = 0, \\
\text{Leibniz  rule:} \quad \{{\cal A}, {\cal BC}\}= \{{\cal A},{\cal  B}\}{\cal C}+{\cal B}\{{\cal A},{\cal C}\}, \\
\text{chain\ rule:} \quad \{f({\cal A}),{\cal B}\}=f'({\cal A})\{{\cal A},{\cal B}\}.
\end{gather*}
In terms of the Kronecker delta $\delta_{jk}$, coordinates $ ({\bf q}, {\bf  p})$ satisfy canonical relations
$\{p_j, p_k\} = \{q_j, q_k\} = 0$, and $\{p_j, q_k\} = \delta_{jk}$.
Using the Poisson bracket, we can rewrite Hamilton's equations~(\ref{hameqns}) as
\begin{gather*}
\{{\cal H},q_j\}=\frac{dq_j}{dt}=\frac{\partial {\cal H}}{\partial p_j} , \qquad
 \{{\cal H}, p_j\}=\frac{dp_j}{dt}=\frac{\partial {\cal H}}{\partial q_j}.
\end{gather*}
For any function  ${\cal F}({\bf q},{\bf  p})$, its dynamics
along a trajectory ${\bf q}(t)$, ${\bf p}(t)$ is
$ \frac{d{\cal F}}{dt}=\{{\cal H}, {\cal F}\}$.
Thus  ${\cal F}({\bf q}, {\bf p})$  will be constant along a trajectory if and only $\{{\cal H},{\cal F}\}=0$.
  If $\{{\cal H},{\cal F}\}=0$,   we say that~$\cal F$  is a {\it constant of the motion}.

\subsection{Classical integrability}
 A system with Hamiltonian ${\cal H}$ is {\it integrable} if it admits $n$ constants of
the motion ${\cal L}_1 = {\cal H}$, ${\cal L}_2,\dots,{\cal L}_n$
that are in involution:
\begin{gather}\label{involution}
\{{\cal L}_j,{\cal L}_k\}=0, \qquad 1\le j,k\le n,
\end{gather}
and are functionally independent in the sense that $\det\big(\frac{\partial {\cal  L}_j}{\partial p_k}\big)\ne 0$.
Suppose~$\cal H$ is integrable with associated constants of the motion~${\cal L}_j$. Then by the inverse function theorem we can
solve the $n$ equations ${\cal L}_j({\bf q}, {\bf p}) = c_j$ for the momenta
to obtain $p_k = p_k({\bf q}, {\bf c})$, $ k = 1,\dots, n$, where ${\bf c} = (c_1,\dots, c_n)$ is a vector of constants.
For an integrable system, if
a particle with position $\bf q$  lies on the common intersection of the hypersurfaces ${\cal L}_j = c_j$
for constants $c_j$, then its momentum~$\bf p$ is completely determined.
Also, if a particle following a trajectory of an integrable system lies
on the common intersection of the hypersurfaces ${\cal L}_j = c_j$ at  time $t_0$, where the ${\cal L}_j$ are constants of the motion, then it
lies on the same common intersection for all $t$ near $t_0$.
Considering $p_j({\bf q},{\bf  c})$ and using the conditions~(\ref{involution}) and the chain rule it is straightforward
to verify $\frac{\partial p_j}{\partial q_k}=\frac{\partial p_k}{\partial q_j}$.
Therefore, there exists a function
$u({\bf q},{\bf  c})$ such that $p_\ell= \frac{\partial u}{\partial q_\ell}$, $\ell=1,\dots,n$. Note that
$ {\cal P}_j({\bf q},\frac{\partial u}{\partial {\bf q}})=c_j$, $j=1,\dots,n$,
and, in particular, $u$ satisf\/ies the {\it Hamilton--Jacobi equation}
\begin{gather}\label{hamjaceqn}
{\cal H}\left({\bf q},\frac{\partial u}{\partial {\bf q}}\right)=E,
\end{gather}
where $E = c_1$. By construction $\det(\frac{\partial u}{\partial q_j\partial c_k})\ne 0,$
 and such a solution of the Hamilton--Jacobi equation depending nontrivially on~$n$ parameters~$\bf c$ is called a {\it complete integral}.
 This argument is reversible:
a complete integral of~(\ref{hamjaceqn}) determines $n$ constants of
the motion in involution, ${\cal P}_1,\dots, {\cal P}_n$.

\begin{Theorem} \label{hjtheorem}
A system is integrable if and only if~\eqref{hamjaceqn}  admits
a complete integral.
\end{Theorem}

 A powerful method for demonstrating that a system is integrable is
to  exhibit a complete integral by using  additive separation of
variables.

It is a standard result in classical mechanics that  for an integrable system one can
 integrate Hamilton's equations and obtain the trajectories~\cite{Arnold, Goldstein}.
The Hamiltonian
formalism is  is well suited to exploiting symmetries of the system and  an important tool in laying the
framework for quantum mechanics.

\subsection{Classical superintegrability}

Let ${\cal F}=(f_1({\bf q},{\bf p}), \dots, f_N({\bf q},{\bf p}))$ be a set of $N$ functions def\/ined and locally analytic in some region of
 a $2n$-dimensional phase space. We say ${\cal F}$ is {\it functionally independent}  if the $N\times 2n$ matrix
$\left( \frac{\partial f_\ell}{\partial  q_j},\frac{\partial f_\ell}{\partial p_k}\right)$
has rank $N$ throughout the region. Necessarily  $N\le 2n$. The set is {\it functionally dependent} if the rank is strictly less than $N$ on the region.
In this case  there is a nonzero  analytic function $F$ of $N$ variables such that $F(f_1,\dots,f_n)=0$ identically on the region. Conversely, if~$F$
exists then the rank of the matrix is~$<N$.
We say that the Hamiltonian  system ${\cal H}$ is { (polynomially)  integrable} if there exists  a set of $N=n$  constants of the motion ${\cal L}_1={\cal H},\dots,{\cal L}_n$,
each polynomial in the momenta
 globally def\/ined (except possibly for isolated singularities),  that  is functionally independent and in involution: $\{{\cal L}_j,{\cal L}_k\}$=0 for $1\le j,k\le n$.
A~classical Hamiltonian system in $n$ dimensions is {\it maximally $($polynomially$)$ superintegrable} if it admits $2n-1$ functionally independent,
globally def\/ined constants, polynomial in the momenta (the maximum number possible~\cite{MPW}).
 At most $n$ functionally independent constants of the motion can be in mutual involution~\cite{Arnold}.
 However, several distinct $n$-subsets of the $2n-1$ polynomial constants of the motion
for a superintegrable system could be in involution. In that case the system is {\it multi-integrable}. An important  feature of
superintegrable systems is that the orbits traced out by the trajectories  can be determined algebraically, without the need for integration.
Along any trajectory each of the symmetries  is constant: ${\cal L}_s=c_s$, $s=1,\dots,2n-1$. Each equation ${\cal L}_s({\bf q},{\bf p})=c_s$ determines a
$(2n-1)$-dimensional hypersurface in the $2n$-dimensional phase space, and the trajectory must lie in that hypersurface. Thus the trajectory lies in the
 common intersection of $2n-1$
independent hypersurfaces; hence it  must be a curve. Another important feature of the trajectories is that all bounded orbits are periodic~\cite{nekhoroshev1972}.

The polynomial constants of the motion for a system with Hamiltonian ${\cal H}$ form the (polynomial) {\it symmetry algebra} $S_{\cal H}$ of the system,
closed under scalar multiplication and addition, multiplication and the Poisson bracket:
Indeed if  $\cal H$  is a Hamiltonian with constants of the motion~${\cal L}$,~${\cal K}$. Then $\alpha {\cal L}
+\beta {\cal K}$, ${\cal L}{\cal K}$ and $\{{\cal L},{\cal K}\}$ are also constants of the motion.
  The $n$ def\/ining constants of the motion of a polynomially integrable system do not generate a very interesting symmetry algebra, because all Poisson brackets  vanish. However, for
the $2n-1$ generators of a~polynomial superintegrable system the brackets cannot all vanish and the symmetry algebra has nontrivial structure.
The {\it degree} (or {\it order})  $O({\cal L})$ of a polynomial constant of the motion $\cal L$ is its degree as a polynomial in the momenta.  Here ${\cal H}$ has degree~2. The {\it degree} $O(F_k)$  of a~set of generators $F_k=\{ {\cal L}_1,\dots,{\cal L}_k\}$, is the maximum degree of the generators.
Let $S=S_{F_k}$ be a~symmetry algebra of a~Hamiltonian system generated by the set $F_k$. Clearly, many dif\/ferent sets~$F'_{k'}$ can generate the same symmetry algebra. Among all these there will be a~set of genera\-tors~$F^0_{k_0}$ for which $\ell=O(F^0_{k_0})$ is a~minimum. Here~$\ell$ is unique, although~$F^0_{k_0}$ is not. We def\/ine the {\it degree} (or {\it order})  of~$S$ to be~$\ell$.
The theory of 2nd degree superintegrable systems has been developed in papers such as~\cite{Dask2001,4,5,KKM20051,KKM20061,KKMW,VILE,RTW2009}; all such systems are known as are the structures of the symmetry algebras, all of which close at degree~6.

\subsection*{Acknowledgement} This work was partially supported by a grant from the Simons Foundation (\# 208754 to Willard Miller, Jr.). We thank Galliano Valent for correcting an error in an earlier draft.

\pdfbookmark[1]{References}{ref}
\LastPageEnding

\end{document}